\documentclass[%
 aip,
 amsmath,amssymb,
 reprint,%
groupedaddress,nofootinbib,floatfix]{revtex4-2}

\usepackage{graphicx}
\usepackage{dcolumn}
\usepackage{bm}

\usepackage{amsfonts}
\usepackage{amsthm}

\usepackage[colorlinks=true,linkcolor=blue,citecolor=blue,urlcolor=blue,plainpages=false,pdfpagelabels]{hyperref}
\usepackage{color}
\usepackage[dvipsnames]{xcolor}
\usepackage{mathtools}
\usepackage{bbm}

\usepackage{float}
\usepackage{array}
\usepackage{verbatim}

\usepackage[utf8]{inputenc}
\usepackage[T1]{fontenc}
\usepackage{etoolbox}

\usepackage{anyfontsize}  
\usepackage{microtype} 

\usepackage{newtxtext}
\usepackage{newtxmath}

\usepackage{boxedminipage}
\usepackage{enumitem}

\makeatletter
\def\@email#1#2{%
 \endgroup
 \patchcmd{\titleblock@produce}
  {\frontmatter@RRAPformat}
  {\frontmatter@RRAPformat{\produce@RRAP{*#1\href{mailto:#2}{#2}}}\frontmatter@RRAPformat}
  {}{}
}%
\makeatother

\makeatletter     
\renewcommand\onecolumngrid{
\do@columngrid{one}{\@ne}%
\def\set@footnotewidth{\onecolumngrid}
\def\footnoterule{\kern-6pt\hrule width 1.5in\kern6pt}%
}
\renewcommand\twocolumngrid{
        \def\footnoterule{
        \dimen@\skip\footins\divide\dimen@\thr@@
        \kern-\dimen@\hrule width.5in\kern\dimen@}
        \do@columngrid{mlt}{\tw@}
}%
\makeatother

\newcommand{\nocontentsline}[3]{}
\newcommand{\tocless}[2]{\bgroup\let\addcontentsline=\nocontentsline#1{#2}\egroup}
\newcommand{\toclesslab}[3]{\bgroup\let\addcontentsline=\nocontentsline#1{#2\label{#3}}\egroup}

\newcommand{\bra}[1]{\langle #1|}
\newcommand{\ket}[1]{|#1\rangle}

\newcommand{\braket}[2]{\langle #1|#2\rangle}
\newcommand{\ketbra}[2]{\ket{#1}\!\bra{#2}}

\newcommand{\abs}[1]{\left|#1\right|}
\newcommand{\e}{\mathrm{e}}
\newcommand{\I}{\mathrm{i}}
\newcommand{\Tr}{\mathrm{Tr}}
\renewcommand{\t}{{\scriptscriptstyle\mathsf{T}}}
\newcommand{\id}{\text{id}}
\newcommand{\conj}[1]{\overline{#1}}

\newcommand{\Lin}{\mathrm{L}}

\newcommand{\SF}[1]{\textnormal{\textsf{#1}}}

\newcommand{\GHZ}{\textnormal{GHZ}}
\newcommand{\ES}{\textnormal{ES}}

\theoremstyle{definition}
\newtheorem{theorem}{Theorem}[section]
\newtheorem{lemma}[theorem]{Lemma}

\newtheorem{definition}[theorem]{Definition}

\newtheorem{proposition}[theorem]{Proposition}
\newtheorem{remark}[theorem]{Remark}

\newenvironment{specbox*}[2]
{\begin{center}\begin{boxedminipage}{#1}
			\begin{center}\textbf{#2}\end{center} }
{\end{boxedminipage}\end{center}}

\definecolor{dred}  {RGB}{164,12,52}

\allowdisplaybreaks

\begin{document}


\title{On the design and analysis of near-term quantum network protocols using Markov decision processes}

\author{Sumeet Khatri}
\email{sumeet.khatri@fu-berlin.de}
\affiliation{Dahlem Center for Complex Quantum Systems, Freie Universit\"{a}t Berlin, 14195 Berlin, Germany}

\date{\today}

\begin{abstract}
	
	The quantum internet is one of the frontiers of quantum information science research. It will revolutionize the way we communicate and do other tasks, and it will allow for tasks that are not possible using the current, classical internet. The backbone of a quantum internet is entanglement distributed globally in order to allow for such novel applications to be performed over long distances. Experimental progress is currently being made to realize quantum networks on a small scale, but much theoretical work is still needed in order to understand how best to distribute entanglement, especially with the limitations of near-term quantum technologies taken into account. This work provides an initial step towards this goal. In this work, we lay out a theory of near-term quantum networks based on Markov decision processes (MDPs), and we show that MDPs provide a precise and systematic mathematical framework to model protocols for near-term quantum networks that is agnostic to the specific implementation platform. We start by simplifying the MDP for elementary links introduced in prior work, and by providing new results on policies for elementary links. In particular, we show that the well-known memory-cutoff policy is optimal. Then we show how the elementary link MDP can be used to analyze a quantum network protocol in which we wait for all elementary links to be active before creating end-to-end links. We then provide an extension of the MDP formalism to two elementary links, which is useful for analyzing more sophisticated quantum network protocols. Here, as new results, we derive linear programs that give us optimal steady-state policies with respect to the expected fidelity and waiting time of the end-to-end link.
	 
\end{abstract}

\maketitle

\tableofcontents

	\begin{figure*}
		\centering
		\includegraphics[width=0.95\textwidth]{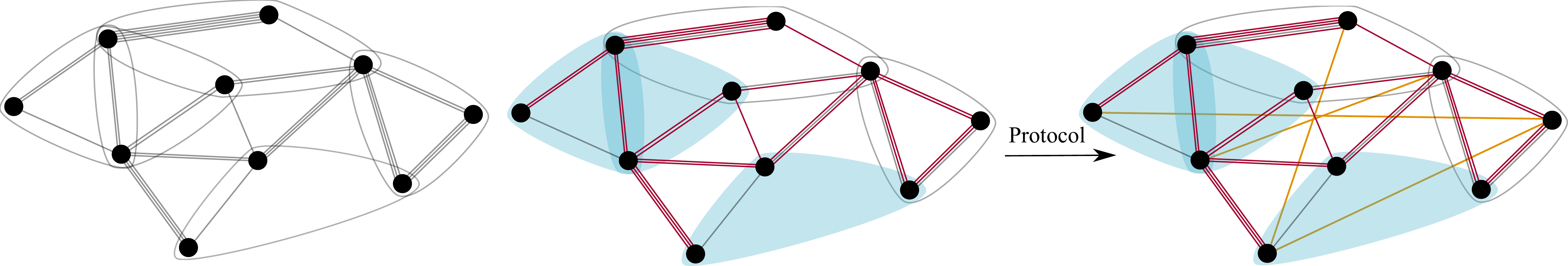}
		\caption{Graphical depiction of a quantum network and entanglement distribution. (Left) The physical layout of the quantum network is described by a hypergraph $G$, which should be thought of as fixed, in which the vertices represent the nodes (senders and receivers) in the network and the (hyper)edges represent quantum channels that are used to distribute entangled states (elementary links) shared by the corresponding nodes. (Center) At any point in time only a certain number of elementary links in the network may be active. By ``active'', we mean that an entangled state has been distributed successfully to the nodes and the corresponding quantum systems stored in the respective quantum memories. Active bipartite links are indicated by a red line, and active $k$-partite elementary links, $k\geq 3$, corresponding to the hyperedges are indicated by a blue bubble. (Right) An entanglement distribution protocol transforms elementary links to virtual links, which are indicated in orange, thus leading to a new graph for the network. The protocol is described mathematically by an LOCC channel.}\label{fig-network_physical} 
	\end{figure*}

\section{Introduction}\label{sec-introduction}

	The quantum internet~\cite{Kim08,Sim17,Cast18,WEH18,Dowling_book2} is envisioned to be a global-scale interconnected network of devices that exploits the uniquely quantum-mechanical phenomenon of entanglement. By operating in tandem with today's internet, it will allow people all over the world to perform quantum communication tasks such as quantum key distribution (QKD)~\cite{BB84,Eke91,GRG+02,SBPC+09,XXQ+20,PAB+19}, quantum teleportation~\cite{BBC+93,Vaidman94,BFK00}, quantum clock synchronization~\cite{JADW00,Preskill00,UD02,DB+18}, distributed quantum computation~\cite{CEHM99,CCT+20}, and distributed quantum metrology and sensing~\cite{DRC17,ZZS18,XZCZ19}. A quantum internet will also allow for exploring fundamental physics~\cite{Bruschi+14}, and for forming an international time standard~\cite{KKBJ+14}. Quantum teleportation and QKD are perhaps the primary use cases of the quantum internet in the near term. In fact, there are several metropolitan-scale QKD systems already in place~\cite{PPA+09,CWL+10,MP10,SLB+11,SFI+11,WCY+14,BLL+18,ZXCPP18}.
	
	Scaling up beyond the metropolitan level towards a global-scale quantum internet is a major challenge. All of the aforementioned tasks require the use of shared entanglement between distant locations on the earth, which typically has to be distributed using single-photonic qubits sent through either the atmosphere or optical fibers. It is well known that optical signals transmitted through either the atmosphere or optical fibers undergo an exponential decrease in the transmission success probability with distance~\cite{SveltoBook,KJK_book,KGMS88_book}, limiting direct transmission distances to roughly hundreds of kilometers. Therefore, one of the central research questions in the theory of quantum networks is how to overcome this exponential loss and thus to distribute entanglement over long distances efficiently and at high rates. 
	
	A quantum network can be modelled as a graph $G=(V,E)$, where the vertices $V$ represent the nodes in the network and the edges in $E$ represent quantum channels connecting the nodes; see Fig.~\ref{fig-network_physical}. Then, the task of entanglement distribution is to transform \textit{elementary links}, i.e., entanglement shared by neighbouring nodes, to \textit{virtual links}, i.e., entanglement between distant nodes; see the right-most panel of Fig.~\ref{fig-network_physical}. In this context, nodes that are not part of the virtual links to be created can act as \textit{quantum repeaters}, i.e., helper nodes whose purpose is to mitigate the effects of loss and noise along a path connecting the end nodes, thereby making the quantum information transmission more reliable~\cite{BDC98,DBC99}. Specifically, quantum repeaters perform entanglement distillation~\cite{BBP96,DAR96,BDSW96} (or some other form of quantum error correction), entanglement swapping~\cite{BBC+93,ZZH93}, and possibly some form of routing, in order to create the desired virtual links. Protocols for entanglement distribution in quantum networks have been described from an information-theoretic perspective in Refs.~[\onlinecite{Pir16,AML16,AK17,BA17,RKB+18,Pir19,Pir19b,DBWH19}], and limits on communication in quantum networks have been explored in Refs.~[\onlinecite{BCHW15,Pir16,AML16,STW16,TSW17,LP17,AK17,BA17,CM17,RKB+18,BAKE20,Pir19,Pir19b,DBWH19,HP22}]. Linear programs, and other techniques for obtaining optimal entanglement distribution rates in a quantum network, have been explored in Refs.~[\onlinecite{BAKE20,DPW20,CERW20,GEW20}]. However, information-theoretic analyses are agnostic to physical implementations, and generally speaking the protocols and the rates derived apply in an idealized scenario, in which quantum memories have high coherence times and quantum gate operations have no error.
	
	What are the fundamental limitations on \textit{near-term quantum networks}? Such quantum networks are characterized by the following elements:
	\begin{itemize}
		\item Small number of nodes;
		\item Imperfect sources of entanglement;
		\item Non-deterministic elementary link generation and entanglement swapping;
		\item Imperfect measurements and gate operations;
		\item Quantum memories with short coherence times;
		\item No (or limited) entanglement distillation/error correction.
	\end{itemize}
	A theoretical framework taking these practical limitations into account would act as a bridge between statements about what can be achieved in principle (which can be answered using information-theoretic methods) and statements that are directly useful for the purpose of implementation. The purpose of this work is to present the initial elements of such a theory of near-term quantum networks.
	
	The main contribution of this work is to frame quantum network protocols in terms of Markov decision processes (MDPs), and to place the Markov decision process for elementary links introduced in Ref.~[\onlinecite{Kha21b}] within an overall quantum network protocol. More specifically, the contributions of this work are as follows:
	\begin{enumerate}
		\item In Sec.~\ref{sec-MDP_elem_links}, we start by recapping the model for elementary link generation presented in Ref.~[\onlinecite{Kha21b}]. Then, as a new contribution, we show that the quantum decision process for elementary links introduced in Ref.~[\onlinecite{Kha21b}] can be written in a simpler manner as an MDP in terms of different variables. Furthermore, we emphasize that the figure of merit associated with the MDP, as introduced in Ref.~[\onlinecite{Kha21b}], takes into account the both the fidelity of the elementary link as well as the probability that it is active. To the best of our knowledge, such a figure of merit has not been considered in prior work. The simplified form of the MDP allows us to derive two new results. The first new result is Theorem~\ref{thm-elem_link_steady_state_dist}, which gives us an analytic expression for the steady-state value of an elementary link undergoing an arbitrary time-homogenous policy. The second new result is Theorem~\ref{thm-opt_pol_mem_cutoff}, in which we show that the so-called ``memory-cutoff policy''---in which the elementary link is kept for some fixed amount of time and then discarded and regenerated---is an optimal policy in the steady-state limit. We demonstrate the usefulness of the the MDP approach to modeling elementary links in Appendix~\ref{sec-elem_link_examples} and \ref{sec-sats}.
		
		\item In Sec.~\ref{sec-joining_dist_protocols}, we describe entanglement distillation protocols and protocols for joining elementary links (in order to create virtual links) in general terms as LOCC quantum instrument channels. We then present three joining protocols and write them down explicitly as LOCC channels. Doing so allows us to determine the output state of the protocol for \textit{any} set of input states, including input states that are noisy as a result of device imperfections, etc. This in turn allows us to compute the fidelity of the output state with respect to the ideal target state that would be obtained if the input states were ideal. Formulas for the fidelity at the output of the protocols are presented as Proposition~\ref{prop-ent_swap_post_fid}, Proposition~\ref{prop-GHZ_ent_swap_post_fid}, and Proposition~\ref{prop-graph_state_dist_post_fid}. In particular, Proposition~\ref{prop-ent_swap_post_fid} provides a formula for the fidelity at the output of the usual entanglement swapping protocol, which to the best of our knowledge is not explicitly found in prior works. Prior works typically use (as an approximation) the product of the individual elementary link fidelities in order to obtain the fidelity after entanglement swapping.
		
		\item In Sec.~\ref{sec-network_protocol}, we present a quantum network protocol that combines the Markov decision process for elementary links with known routing and path-finding algorithms. In essence, the protocol is a simple one in which we first wait for all of the relevant elementary links to become active, and then we perform the required joining operations to establish the virtual links; see Fig.~\ref{fig-QDP_protocol} for a summary. For this protocol, we provide a general method for determining waiting times and key rates for quantum key distribution. 
		
		\item In Sec.~\ref{sec-MDP_two_elem_links}, we provide a first step towards extending the elementary link MDP by defining an MDP for two elementary links with entanglement swapping. We then show how to approximate waiting times using a linear program, and we find that this linear programming approximation reproduces exactly the known analytic results on the waiting time for such a scenario~\cite{CJKK07}. However, our result is more general, allowing us to compute waiting times for arbitrary parameter regimes, while the analytic results are true only for restricted parameter regimes. Broadly speaking, having linear-programming approximations to the waiting time and other important quantities of interest (such as fidelity) will be important when considering MDPs for larger networks.
	\end{enumerate}
	
	This work is one in a long line of work on quantum repeaters, taking device imperfections and noise into account, beginning with the initial theoretical proposal~\cite{BDC98,DBC99}, and then resulting in a vast body of work~\cite{DLCZ01,CJKK07,SRM+07,SSM+07,BPv11,ZDB12,KKL15,ZBD16,EKB16,WZM+16,LZH+17,VK17,MMG19,ZPD+18,PWD18,DKD18,WPZD19,PD19,HBE20,GKL+03,RHG05,JTN+09,FWH+10,MSD+12,MKL+14,NJKL16,MLK+16,MEL17,CERW20,GEW20,STCMW20,DT21,RNX+21,BCMO21}. (See also Refs.~[\onlinecite{SSR+11,MATN15,VanMeter_book,ABC21,MPD+22}] and the references therein.) All of these proposals deal almost exclusively with a single line of repeaters connecting a sender and a receiver. However, for a quantum internet, we need to go beyond the linear topology to an arbitrary topology, and we need to consider multiple transmissions operating in parallel. While recent small-scale experiments~\cite{Hump+18,chung2021illinois,chung2022illinois,PHB+21,hermans2022teleportation,pompili2022demonstrationstack,bradley2022robustmemory} have demonstrated some of the key building blocks of a quantum internet, a unified and self-consistent theoretical framework will help to guide real-world implementations, especially when scaling up to larger distances and more nodes. It is our hope that this work provides a good starting point along this line of thought, and leads to a better understanding of how realistic, near-term quantum devices could be used to realize large-scale quantum networks, and eventually a global-scale quantum internet.

\section{Markov decision process for elementary links}\label{sec-MDP_elem_links}
	
	We start by presenting a Markov decision process (MDP) for elementary links, based on Ref.~[\onlinecite{Kha21b}]. To be specific, this is an MDP for an arbitrary edge of the graph corresponding to a quantum network. We start by describing the physical model of elementary link generation. Then, we define the MDP corresponding to this model of elementary link generation.

\subsection{Elementary link generation}\label{sec-practical_elem_link_generation}

	Our model for elementary link generation is the one considered in Ref.~[\onlinecite{Kha21b}] and illustrated in Fig.~\ref{fig-net_architecture}, based on the same model considered in prior work~\cite{AJP+03,JKR+16,DKD18,KMSD19}. Consider an arbitrary physical link in the network. For every such physical link, there is a source station that prepares and distributes an entangled state to the corresponding nodes. In general, all of these source stations operate independently of each other, distributing entangled states as they are requested. Specifically, we have the following.
	
	\begin{figure}
		\centering
		\includegraphics[scale=1]{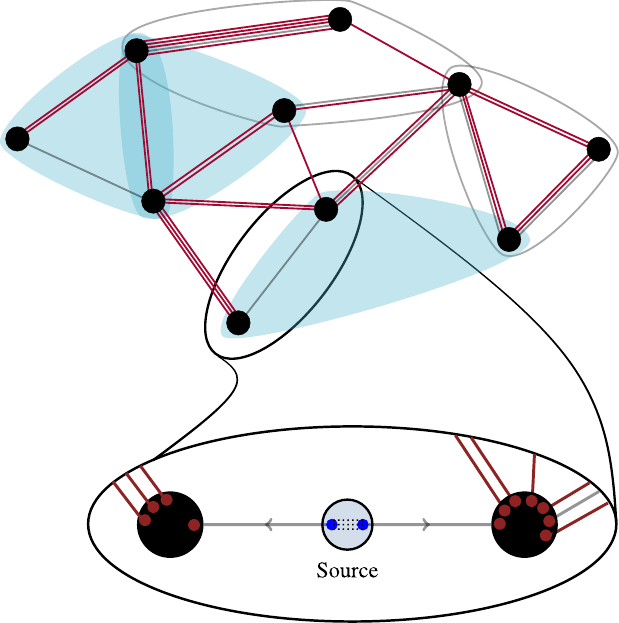}
		\caption{Our model for elementary link generation in a quantum network consists of source stations associated to every elementary link that distributes entangled states to the corresponding nodes~\cite{AJP+03,JKR+16,DKD18,KMSD19}. (Adapted from Ref.~[\onlinecite{Kha21b}].)}\label{fig-net_architecture}
	\end{figure}

	\begin{itemize}
		\item The source produces a $k$-partite quantum state $\rho^S$, $k\geq 2$, and sends it to the nodes via a quantum channel $\mathcal{S}$, leading to the state $\mathcal{S}(\rho^S)$. Here, $k$ is the number of nodes belonging to an edge, with $k=2$ corresponding to ordinary, bipartite edges (such as the red edges in Fig.~\ref{fig-net_architecture}) and $k\geq 3$ corresponding to hyperedges (such as the blue bubbles in Fig.~\ref{fig-net_architecture}).
		
		\item The nodes perform a heralding procedure, which is a protocol involving local operations and classical communication. It can be described by a quantum instrument $\{\mathcal{M}^0,\mathcal{M}^1\}$, where $\mathcal{M}^0$ and $\mathcal{M}^1$ are completely positive trace non-increasing maps such that $\mathcal{M}^0+\mathcal{M}^1$ is trace preserving. These maps capture not only the probabilistic nature of the heralding procedure but also the various imperfections of the devices that are used to perform the procedure. The map $\mathcal{M}^0$ corresponds to failure of heralding and $\mathcal{M}^1$ corresponds to success. The probability of successful transmission and heralding is
			\begin{equation}\label{eq-elem_link_success_prob}
				p=\Tr[(\mathcal{M}^1\circ\mathcal{S})(\rho^S)],
			\end{equation}
			and the states conditioned on success and failure are, respectively,
			\begin{align}
				\sigma^0&\coloneqq\frac{1}{p}(\mathcal{M}^1\circ\mathcal{S})(\rho^S),\label{eq-initial_link_state_success}\\
				\tau^{\varnothing}&\coloneqq\frac{1}{1-p}(\mathcal{M}^0\circ\mathcal{S})(\rho^S). \label{eq-initial_link_state_failure}
			\end{align}
			The superscript ``$0$'' in $\sigma^0$ indicates that, upon success of the heralding procedure, the quantum systems have been immediately stored in local quantum memories at the nodes and have not yet suffered from any decoherence.
			
		\item The state of the quantum systems after $m\in\{0,1,2,\dotsc\}$ time steps in the quantum memories is given by
			\begin{equation}\label{eq-elem_link_state_memory}
				\sigma(m)\coloneqq\mathcal{N}^{\circ m}(\sigma^0),
			\end{equation}
			where $\mathcal{N}$ is a quantum channel that describes the decoherence of the individual quantum memories at the nodes.
	\end{itemize}
	For specific, realistic noise models for the heralding and for the quantum memories, as well as for other realistic parameters for elementary link generation, we refer to Refs.~[\onlinecite{RGR+18,RYG+18,DSM+19,WAS20,CKD+20,GEW20,PHB+21}]. Also, in Appendix~\ref{sec-elem_link_examples}, we present two specific models of elementary link generation, as special cases of the abstract developments presented here.

\subsection{Definition of the MDP}\label{sec-elem_link_MDP_def}

	Having described the physical model of elementary link generation in the previous section, let us now proceed to the definition of the Markov decision process (MDP) for an elementary link. Note that while the formalism of the previous section gives us a mathematical description of the quantum state of an elementary link immediately after it is successfully generated, the MDP formalism provides us with a systematic framework to define actions on an elementary link and their effects on the quantum state over time. 
	
	Before starting, let us briefly summarize the definition of a Markov decision process (MDP); we refer to Appendix~\ref{sec-MDPs_overview} for more details and a detailed explanation of the notation being used. An MDP is a mathematical model of an agent performing actions on a system (usually called the environment). The system is described by a set $\SF{S}$ of \textit{(classical) states}, and the agent picks actions from a set $\SF{A}$. Corresponding to every action $a\in\SF{A}$ is a $\abs{\SF{S}}\times\abs{\SF{S}}$ \textit{transition matrix} $T^a$, such that the matrix element $T^a(s';s)$ is equal to the probability of transitioning to the state $s'\in\SF{S}$ given that the current state is $s\in\SF{S}$ and the action $a\in\SF{A}$ is taken. 
	
	The results of Ref.~[\onlinecite{Kha21b}] show us that, for the purposes of tracking the quantum state of an elementary link over time, as well as its fidelity to a target pure state, it is enough to keep track of the time that the quantum systems of the elementary link reside in their respective quantum memories. With this observation, we can define a simpler MDP for elementary links; see Fig.~\ref{fig-elem_link_MDP}.
	
	\begin{figure}
		\centering
		\includegraphics[width=0.7\columnwidth]{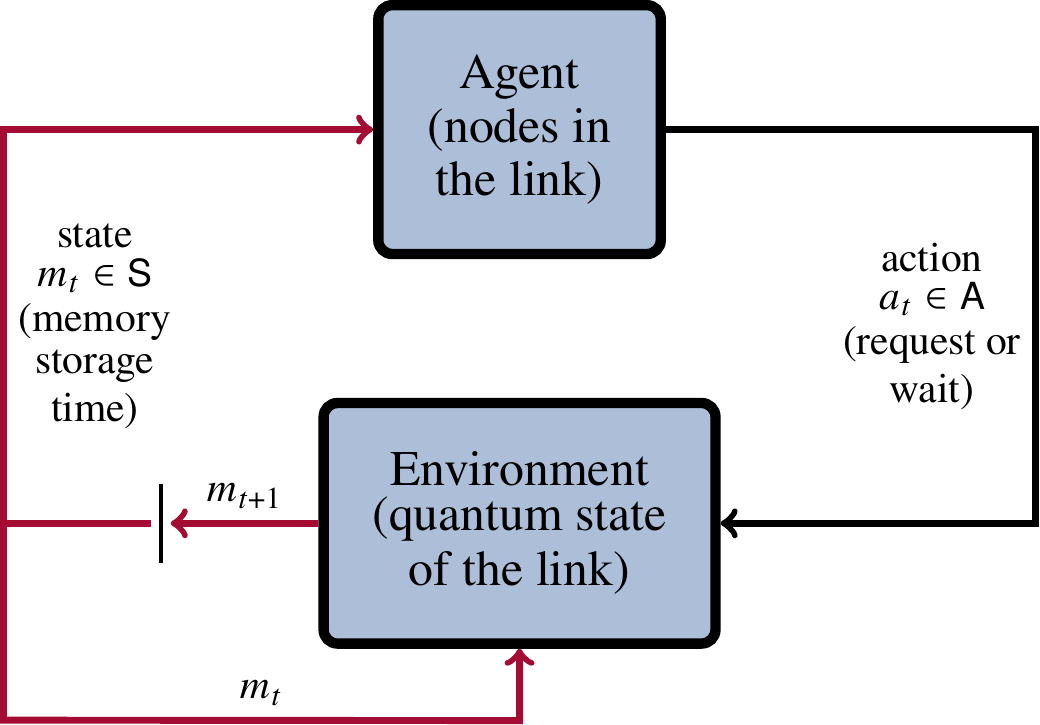}
		\caption{Schematic depiction of the Markov decision process (MDP) for elementary links presented in Sec.~\ref{sec-elem_link_MDP_def}. The MDP provides us with a systematic way of keeping track of the quantum state of an elementary link in a quantum network based on the actions at the nodes of the elementary link. Specifically, the states $m_t$ of the MDP encode information about the quantum state via \eqref{eq-elem_link_state_memory}.}\label{fig-elem_link_MDP}
	\end{figure}
	
	\begin{itemize}
		\item \textit{States}: The states in our elementary link MDP are defined by the set $\SF{S}=\{-1,0,1,\dotsc,m^{\star}\}$, which correspond to the number of time steps that the quantum systems of the elementary link have been sitting in their respective quantum memories. The state $-1$ corresponds to the elementary link being inactive, and $m^{\star}\in\mathbb{N}_0\coloneqq\{0,1,2,\dotsc\}$ corresponds to the coherence time of the quantum memory. Specifically, if $t_{\text{coh}}$ is the coherence time of the quantum memory (say, in seconds), and the duration of every time step (in seconds) is $\Delta t$ (based on the classical communication time between the nodes in the elementary link), then $m^{\star}=\frac{t_{\text{coh}}}{\Delta t}$. From now on, we refer to $m^{\star}$ as the \textit{maximum storage time} of the elementary link.
		
			We use $M(t)$, $t\in\mathbb{N}$, to refer to the random variables (taking values in $\SF{S}$) corresponding to the state of the MDP at time $t$. We also associate to the elements in $\SF{S}$ orthonormal vectors $\{\ket{m}\}_{m\in\SF{S}}$, and we emphasize that these vectors should not be thought of as representing quantum states but as representing the extreme points of a probability simplex associated with the set $\SF{S}$; see Appendix~\ref{sec-MDPs_overview} for details.
		
		\item \textit{Actions}: The set of actions is $\SF{A}=\{0,1\}$, where 0 corresponds to the action of ``wait'' and 1 corresponds to ``request''. In other words, at every time step, the agent can decide to keep their quantum systems currently in memory (``wait'') or to discard the quantum systems and perform the elementary link generation procedure again (``request'').
		
			The transition matrices $T^0$ and $T^1$ corresponding to the two actions are defined as follows:
			\begin{align}
				T^0&=\mathbbm{1}^{(-)}+B^{(+)}\label{eq-elem_link_transition_0}\\
				T^1&=\ketbra{g_p}{\gamma},\label{eq-elem_link_transition_1}
			\end{align}
			where
			\begin{align}
				\mathbbm{1}^{(-)}&\coloneqq\ketbra{-1}{-1},\\
				B^{(+)}&\coloneqq \sum_{m=0}^{m^{\star}-1}\ketbra{m+1}{m}+\ketbra{-1}{m^{\star}},\\
				\ket{g_p}&\coloneqq(1-p)\ket{-1}+p\ket{0},\label{eq-elem_link_MDP_gen_vector}\\
				\ket{\gamma}&=\sum_{m=-1}^{m^{\star}}\ket{m}.
			\end{align}
			(Note that we define our transition matrices such that probability vectors are applied to them from the right; see Appendix~\ref{sec-MDPs_overview} for details.) The transition matrix $T^0$ describes what happens to the elementary link when the action $a=0$ (``wait'') is taken by the agent: if the elementary link is currently inactive, then it stays inactive; if the elementary link is active, and it is in memory for less than $m^{\star}$ time steps, then the memory time is incremented by one; if the elementary link is active and it has been in memory for $m^{\star}$ time steps, then because the coherence time of the memory has been reached (as per the definition of $m^{\star}$), the elementary link becomes inactive. If the action $a=1$ (``request'') is taken, then regardless of the current state of the elementary link, the state changes to $-1$ (inactive) with probability $1-p$, meaning that the elementary link generation failed, or it changes to $0$ with probability $p$, meaning that the elementary link generation succeeded. These two possibilities are captured by the probability vector $\ket{g_p}$.
			 
			 We use $A(t)$, $t\in\mathbb{N}$, to refer to the random variable (taking values in the set $\SF{A}$) corresponding to the action taken at time $t$.
			 
			 We let $H(t)=(M(1),A(1),M(2),A(2),\dotsc,A(t-1),M(t))$ be the \textit{history}, consisting of a sequence of states and actions, up to time $t$, with $H(1)=M(1)$.
		
		\item \textit{Figure of merit}: Our figure of merit for an elementary link is the following function:
			\begin{align}
				f(m)&\coloneqq\left\{\begin{array}{l l} \bra{\psi}\sigma(m)\ket{\psi} & \text{if }m\in\{0,1,2,\dotsc,m^{\star}\}, \\ 0 & \text{if }m=-1. \end{array}\right.\\
				&=(1-\delta_{m,-1})\bra{\psi}\sigma(m)\ket{\psi},\label{eq-elem_link_fidelity}
			\end{align}
			where $\sigma(m)$ is defined in \eqref{eq-elem_link_state_memory} and $\ket{\psi}$ is a target state vector for the elementary link. (For example, if the elementary link contains two nodes, then $\ket{\psi}$ could be the state vector for the two-qubit maximally entangled state.) We emphasize that the function $f$ is not just the fidelity of the elementary link---it also depends implicitly on the probability that the elementary link is active, because if $f$ was simply the fidelity of the elementary link then instead of the definition $f(-1)=0$ we would have $f(-1)=\bra{\psi}\tau^{\varnothing}\ket{\psi}$, where $\tau^{\varnothing}=(1/(1-p))(\mathcal{M}^0\circ\mathcal{S})(\rho^S)$ is the quantum state corresponding to failure of the heralding procedure; see \eqref{eq-initial_link_state_failure}. We illustrate the importance of this distinction, and therefore the usefulness of this figure of merit for designing and evaluating protocols, in Sec.~\ref{sec-sats_policies_mem_cutoff}, specifically Fig.~\ref{fig-tstarInf_sats}. To the best of our knowledge, this figure of merit has not been considered in prior work.
	\end{itemize}
	
	
	A \textit{policy} is a sequence $\pi=(d_1,d_2,\dotsc)$ of \textit{decision functions} $d_t:\SF{S}\times\SF{A}\to[0,1]$, which indicate the probability of performing a particular action conditioned on the state of the system:
	\begin{equation}
		d_t(s)(a)=\Pr[A(t)=a|S(t)=s].
	\end{equation}
	For a particular policy $\pi=(d_1,d_2,\dotsc,d_{t-1})$, the probability of a particular history $h^t=(m_1,a_1,m_2,a_2,\dotsc,a_{t-1},m_t)$ of states and actions is (see Appendix~\ref{sec-MDPs_definitions})
	\begin{multline}\label{eq-MDP_elem_link_history_prob}
		\Pr[H(t)=h^t]_{\pi}\\=\Pr[M(1)=m_1]\prod_{j=1}^{t-1}T^{a_j}(m_{j+1};m_j)d_j(m_j)(a_j).
	\end{multline}
	Then, the quantum state of the elementary link is~\cite{Kha21b}
	\begin{align}
		\rho^{\pi}(t)&=\sum_{h^t}\Pr[H(t)=h^t]_{\pi}\,\ketbra{h^t}{h^t}\otimes\sigma(t|h^t),\label{eq-elem_link_initial_cq_state_2}\\
		\sigma(t|h^t)&=(1-\delta_{m_t,-1})\sigma(m_t)+\delta_{m_t,-1}\tau^{\varnothing},\label{eq-elem_link_cond_state}
	\end{align}
	where we recall that $\sigma(m_t)$ is given by \eqref{eq-elem_link_state_memory}.
	
	We are interested primarily in the expected value of the function $f$ defined in \eqref{eq-elem_link_fidelity} at times $t\in\mathbb{N}$:
	\begin{equation}\label{eq-network_QDP_total_fid_pol_elem}
		\widetilde{F}^{\pi}(t)\coloneqq\mathbb{E}[f(M(t))]_{\pi}=\sum_{m=0}^{m^{\star}}f(m)\Pr[M(t)=m]_{\pi},
	\end{equation}
	for policies $\pi=(d_1,d_2,\dotsc,d_{t-1})$. We are also interested in the probability that the elementary link is active at time $t\in\mathbb{N}$, which is given by
	\begin{equation}\label{eq-link_value}
		X^{\pi}(t)\coloneqq 1-\Pr[M(t)=-1]_{\pi}.
	\end{equation}
	From this, the expected fidelity of the elementary link is given by
	\begin{equation}
		F^{\pi}(t)\coloneqq\frac{\widetilde{F}^{\pi}(t)}{X^{\pi}(t)}.
	\end{equation}

\subsection{Optimal policies}\label{sec-elem_link_policies}

	We define an \textit{optimal policy} to be one that achieves the quantity $\sup_{\pi}\widetilde{F}^{\pi}(t)$, i.e., the maximum value of the function $\widetilde{F}^{\pi}$ defined in \eqref{eq-network_QDP_total_fid_pol_elem} among all policies~$\pi$. In the steady-state (infinite-time) limit, we are interested in the quantity
	\begin{multline}\label{eq-opt_elem_link_fidelity_avg_inf}
		\sup_d\lim_{t\to\infty}\widetilde{F}^{(d,d,\dotsc)}(t)\\=\sup_d\sum_{m=0}^{m^{\star}}f(m)\lim_{t\to\infty}\Pr[M(t)=m]_{(d,d,\dotsc)}
	\end{multline}
	(if the limit exists), which is the maximum value of $\widetilde{F}^{\pi}$ among all time-homogeneous (stationary) policies $\pi=(d,d,\dotsc)$, i.e., policies in which a fixed decision function $d$ is used at every time step.

	In Ref.~[\onlinecite{Kha21b}], it was shown that an optimal policy can be determined using a backward recursion algorithm. We restate this algorithm here for completeness.
	
	\begin{theorem}[Optimal finite-time policy for an elementary link~\cite{Kha21b}]\label{thm-opt_policy}
		For all $t\in\mathbb{N}$, the optimal value of an elementary link with success probability $p\in[0,1]$ and maximum storage time $m^{\star}\in\mathbb{N}_0$ is given by
		\begin{equation}\label{eq-opt_policy_BR_1}
			\sup_{\pi}\widetilde{F}^{\pi}(t)=\sum_{m_1\in\SF{S}}\max_{a_1\in\SF{A}}w_2(m_1,a_1),
		\end{equation}
		where
		\begin{equation}\label{eq-opt_policy_BR_2}
			w_j(h^{j-1},a_{j-1})=\sum_{m_j\in\SF{S}}\max_{a_j\in\SF{A}}w_{j+1}(h^{j-1},a_{j-1},m_j,a_j)
		\end{equation}
		for all $j\in\{2,3,\dotsc,t-1\}$, and
		\begin{multline}\label{eq-opt_policy_BR_3}
			w_t(h^{t-1},a_{t-1})\\=\sum_{m_t\in\SF{S}}\braket{m_1}{g_p}\left(\prod_{j=1}^{t-1}T^{a_j}(m_{j+1};m_j)\right)f(m_t).
		\end{multline}
		Furthermore, the optimal policy is deterministic and given by $\pi=(d_1^*,d_2^*,\dotsc,d_{t-1}^*)$, where
		\begin{equation}\label{eq-opt_policy_BR_4}
			d_j^*(h^j)=\max_{a\in\SF{A}}w_{j+1}(h^j,a)\quad\forall~j\in\{1,2,\dotsc,t-1\}.
		\end{equation}
	\end{theorem}
	
	Intuitively, the result of Theorem~\ref{thm-opt_policy} tells us that, for finite times, an optimal policy can be found by optimizing the individual actions going ``backwards in time'', by first optimizing the final action at time $t-1$ and then optimizing the action at time $t-2$, etc., and then finally optimizing the action at time $t=1$. This is indeed the case, because from \eqref{eq-opt_policy_BR_4} we see that the optimal action at the first time step is obtained using the function $w_2$, but from \eqref{eq-opt_policy_BR_2} we see that to calculate $w_2$ we need $w_3$, and to calculate $w_3$ we need $w_4$, etc., until we get to the function $w_t$ for the final time step, which we can calculate using \eqref{eq-opt_policy_BR_3}.
	
	While the optimal policy for finite times was determined in Ref.~[\onlinecite{Kha21b}], the steady-state value of the function $\widetilde{F}$ with respect to arbitrary stationary policies (i.e., the value in \eqref{eq-opt_elem_link_fidelity_avg_inf}) was not determined. We now show that the limit in \eqref{eq-opt_elem_link_fidelity_avg_inf} exists, and we determine its value for arbitrary decision functions.
	
	\begin{theorem}[Steady-state expected value of an elementary link]\label{thm-elem_link_steady_state_dist}
		Let $p\in[0,1]$ be the success probability of generating an elementary link in a quantum network, let $m^{\star}\in\mathbb{N}_0$ be the maximum storage time of the elementary link, and let $d$ be a decision function such that $d(m)(0)=\alpha(m)$, $m\in\{-1,0,1,\dotsc,m^{\star}\}$, is the probability of executing the action ``wait'' and $d(m)(1)=1-d(m)(0)=\overline{\alpha}(m)$ is the probability of executing the action ``request''. If the elementary link undergoes the stationary policy $(d,d,\dotsc)$, then
		\begin{equation}\label{eq-elem_link_steady_state_dist}
			\lim_{t\to\infty}\widetilde{F}^{(d,d,\dotsc)}(t)=\sum_{m=0}^{m^{\star}}f(m)s_d(m),
		\end{equation}
		where
		\begin{align}
			s_d(-1)&=\frac{1}{N_d}\left(1-p\left(1-\prod_{m'=0}^{m^{\star}}\alpha(m')\right)\right),\\
			s_d(0)&=\frac{1}{N_d}p\overline{\alpha}(-1),\\
			s_d(m)&=\frac{1}{N_d}p\overline{\alpha}(-1)\prod_{m'=0}^{m-1}\alpha(m'),\quad m\in\{1,\dotsc,m^{\star}\},
		\end{align}
		with
		\begin{multline}\label{eq-elem_link_steady_state_dist_normalization}
			N_d=1-p\left(1-\prod_{m'=0}^{m^{\star}}\alpha(m')\right)\\+p\overline{\alpha}(-1)\left(1+\sum_{m=1}^{m^{\star}}\prod_{m'=0}^{m-1}\alpha(m')\right).
		\end{multline}
	\end{theorem}
	
	\begin{proof}
		See Appendix~\ref{sec-elem_link_steady_state_dist_pf}.
	\end{proof}
	
	Using Theorem~\ref{thm-elem_link_steady_state_dist}, we can determine the optimal steady-state value of the function $\widetilde{F}^{(d,d,\dotsc)}$, and thus the optimal decision function $d$, by optimizing the quantity in \eqref{eq-elem_link_steady_state_dist} with respect to $m^{\star}$ independent variables $\alpha(-1),\alpha(0),\dotsc,\alpha(m^{\star})$ subject to the constraints $\alpha(m)\in[0,1]$ for all $m\in\{-1,0,1,\dotsc,m^{\star}\}$. (Recall from the statement of Theorem~\ref{thm-elem_link_steady_state_dist} that the variables $\alpha(m)$ are directly related to the decision function $d$.) Alternatively, we can use the following linear program in order to obtain an optimal policy.
	
	\begin{proposition}[Linear program for the optimal steady-state value of an elementary link]\label{thm-opt_elem_link_fidelity_avg_inf_LP}
		Consider an elementary link in a quantum network with generation success probability $p\in[0,1]$ and maximum storage time $m^{\star}\in\mathbb{N}_0$. Let $\ket{f}\coloneqq\sum_{m=-1}^{m^{\star}}f(m)\ket{m}$. The optimal steady-state value of the elementary link, namely, the quantity in \eqref{eq-opt_elem_link_fidelity_avg_inf}, is equal to the solution to the following linear program:
		\begin{equation}\label{eq-opt_elem_link_fidelity_avg_inf_LP}
			\begin{array}{l l} \text{maximize} & \braket{f}{v} \\[0.2cm] \text{subject to} & 0\leq \ket{w_a}\leq \ket{v}\leq 1\quad\forall~a\in\{0,1\}, \\[0.1cm] & \braket{\gamma}{v}=1, \\[0.1cm] & \ket{w_0}+\ket{w_1}=\ket{v}=T^0\ket{w_0}+T^1\ket{w_1},  \end{array}
		\end{equation}
		where the optimization is with respect to the $(m^{\star}+1)$-dimensional vectors $\ket{v},\ket{w_0},\ket{w_1}$, and the inequality constraints on the vectors are componentwise. For every feasible point of this linear program, we obtain a decision function $d$ as follows: $d(m)(a)=\frac{\braket{m}{w_a}}{\braket{m}{v}}$ for all $m\in\{-1,0,1,\dotsc,m^{\star}\}$ and $a\in\{0,1\}$. If $\braket{m}{v}=0$, then we set $d(m)(0)=\alpha(m)$ and $d(m)(1)=1-\alpha(m)$ for an arbitrary $\alpha(m)\in[0,1]$.
	\end{proposition}
	
	\begin{proof}
		The linear program in \eqref{eq-opt_elem_link_fidelity_avg_inf_LP} is a special case of the linear program presented in Proposition~\ref{prop-lin_prog_func_1} in Appendix~\ref{sec-MDPs_overview}. The main assumption of that result is that the MDP be ergodic, which is true in this case by Theorem~\ref{thm-elem_link_steady_state_dist}.
	\end{proof}

\subsection{The memory-cutoff policy and its optimality}\label{sec-elem_link_mem_cutoff_policy}
	
	An example of a stationary policy is the \textit{memory-cutoff policy}, which has been considered extensively in prior work~\cite{CJKK07,SRM+07,SSM+07,BPv11,JKR+16,DHR17,RYG+18,SSv19,KMSD19,SJM19,LCE20,Kha21b}. This is a deterministic policy that is defined by a cutoff time $t^{\star}\in\mathbb{N}_0\cup\{\infty\}$, where $\mathbb{N}_0\coloneqq\{0,1,2,\dotsc,\}$, such that $t^{\star}\leq m^{\star}$. Then, the decision function for this policy is defined by the values $d^{t^{\star}}\!(m)(0)$ and $d^{t^{\star}}\!(m)(1)=1-d^{t^{\star}}\!(m)(0)$ for all $m\in\{-1,0,1,\dotsc,t^{\star}\}$ as follows:
	\begin{equation}\label{eq-mem_cutoff_policy}
		d^{t^{\star}}\!(m)(0)=\left\{\begin{array}{l l} 0 & \text{if }m=-1,\,t^{\star},\\ 1 & \text{if }m\in\{0,1,\dotsc,t^{\star}-1\}, \end{array}\right.
	\end{equation}
	for all $t^{\star}\in\mathbb{N}_0$. In other words, the elementary link is kept in memory for $t^{\star}$ time steps, and then it is discarded and regenerated. If $t^{\star}=\infty$,
	\begin{equation}
		d^{\infty}(m)(0)=\left\{\begin{array}{l l} 0 & \text{if }m=-1, \\ 1 & \text{otherwise}, \end{array}\right.
	\end{equation}
	which means that the elementary link, once generated, is never discarded.
	
	For the memory-cutoff policy, we use the abbreviations $\widetilde{F}^{t^{\star}}\equiv\widetilde{F}^{(d^{t^{\star}},d^{t^{\star}},\dotsc)}$, $X^{t^{\star}}\equiv X^{(d^{t^{\star}},d^{t^{\star}},\dotsc)}$, and $F^{t^{\star}}\equiv F^{{(d^{t^{\star}},d^{t^{\star}},\dotsc)}}$. Using Theorem~\ref{thm-elem_link_steady_state_dist}, we have $N_d=1+t^{\star}p$ and $s_{d^{t^{\star}}}(m)\equiv s_{t^{\star}}(m)=\frac{p}{1+t^{\star}p}$ for all $m\in\{0,1,\dotsc,t^{\star}\}$, so that
	\begin{equation}\label{eq-avg_fid_tilde_tInfty}
		\lim_{t\to\infty}\widetilde{F}^{t^{\star}}\!(t)=\frac{p}{1+t^{\star}p}\sum_{m=0}^{t^{\star}}f(m),
	\end{equation} 
	for all $t^{\star}\in\mathbb{N}_0$, which agrees with Ref.~[\onlinecite[Eq.~(4.15)]{Kha21b}], which was obtained using different methods. We also obtain
	\begin{align}
		\lim_{t\to\infty}X^{t^{\star}}\!(t)&=\frac{(t^{\star}+1)p}{1+t^{\star}p},\\
		\lim_{t\to\infty}F^{t^{\star}}\!(t)&=\frac{1}{t^{\star}+1}\sum_{m=0}^{t^{\star}}f(m), \label{eq-avg_fid_tInfty}
	\end{align}
	for all $t^{\star}\in\mathbb{N}_0$.
	
	For $t^{\star}=\infty$, we have, for all $t\geq 1$~\cite{Kha21b},
	\begin{align}
		\widetilde{F}^{\infty}(t)&=\sum_{m=0}^{t-1}f(m)p(1-p)^{t-(m+1)},\label{eq-avg_fid_tilde_sats}\\
		X^{\infty}(t)&=1-(1-p)^t, \label{cor-link_status_Pr1}\\
		F^{\infty}(t)&=\sum_{m=0}^{t-1}f(m)\frac{p(1-p)^{t-(m+1)}}{1-(1-p)^t}.
	\end{align}
	
	

	It turns out that, in the steady-state limit, there always exists a cutoff such that the memory-cutoff policy achieves the optimal value of the elementary link.
	
	\begin{theorem}[Optimality of the memory-cutoff policy in the steady-state limit]\label{thm-opt_pol_mem_cutoff}
		Consider an elementary link in a quantum network with generation success probability $p\in[0,1]$ and maximum storage time $m^{\star}\in\mathbb{N}_0$. The optimal steady-state value of the elementary link, namely, the quantity in \eqref{eq-opt_elem_link_fidelity_avg_inf}, is achieved by a memory-cutoff policy, i.e.,
		\begin{equation}
			\sup_d \lim_{t\to\infty}\widetilde{F}^{(d,d,\dotsc)}(t)=\max_{t^{\star}\in\{0,1,\dotsc,m^{\star}\}} \frac{p}{1+t^{\star}p}\sum_{m=0}^{t^{\star}}f(m).
		\end{equation}
	\end{theorem}
	
	\begin{proof}
		See Appendix~\ref{sec-thm_opt_pol_mem_cutoff_pf}.
	\end{proof}

\section{Entanglement distillation and joining protocols}\label{sec-joining_dist_protocols}

	In the previous section, we discussed elementary links in a quantum network, how to model the generation of elementary links and how to model them in time in terms of a Markov decision process. The description of an elementary link in terms of a Markov decision process allows us to determine, as a function of time, the quantum state of an elementary link. Keeping in mind the overall goal of entanglement distribution, i.e., the creation of long-distance virtual links, the next step in an entanglement distribution protocol is to take elementary links and to improve their fidelity using entanglement distillation and then to join them in order to create the virtual links (using, e.g., entanglement swapping). In this section, we explain how to model entanglement distillation protocols and joining protocols using LOCC channels. We refer to Appendix~\ref{sec-LOCC_channels} for a detailed explanation of LOCC channels. The explicit description of these protocols as LOCC channels is important because, as we saw in the previous section, the quantum state of an elementary link will not always be the ideal entangled state with respect to which joining protocols are typically defined. It is therefore important to understand how the protocols will act when the input states are not ideal.

\subsection{Entanglement distillation}\label{sec-ent_distill}

	The term ``entanglement distillation'' refers to the task of taking many copies of a given quantum state $\rho_{AB}$ and transforming them, via an LOCC protocol, to several (fewer) copies of the maximally entangled state $\Phi_{AB}\coloneqq\frac{1}{d}\sum_{i,j=0}^{d-1}\ketbra{i,i}{j,j}$. Typically, with only a finite number of copies of the initial state $\rho_{AB}$, it is not possible to perfectly obtain copies of the maximally entangled state, so we aim instead for a state $\sigma_{AB}$ whose fidelity $F(\Phi_{AB},\sigma_{AB})$ to the maximally entangled state is higher than the fidelity $F(\Phi_{AB},\rho_{AB})$ of the initial state. Mathematically, the task of entanglement distillation corresponds to the transformation
		\begin{equation}
			\rho_{AB}^{\otimes n}\mapsto\mathcal{L}_{A^nB^n\to A^m B^m}(\rho_{AB}^{\otimes n})=\sigma_{AB}^{\otimes m},
		\end{equation}
		where $n,m\in\mathbb{N}$, $m<n$, and $\mathcal{L}_{A^nB^n\to A^m B^m}$ is an LOCC channel.
		
		Typically, in practice, we have $n=2$ and $m=1$, with the task being to transform two two-qubit states $\rho_{A_1B_1}^1$ and $\rho_{A_2B_2}^2$ to a two-qubit state $\sigma_{A_1B_1}$ having a higher fidelity to the maximally entangled state than the initial states. Protocols achieving this aim are typically probabilistic in practice, meaning that the state $\sigma_{A_1B_1}$ with higher fidelity is obtained only with some non-unit probability.

	\begin{figure}
		\centering
		\includegraphics[scale=1]{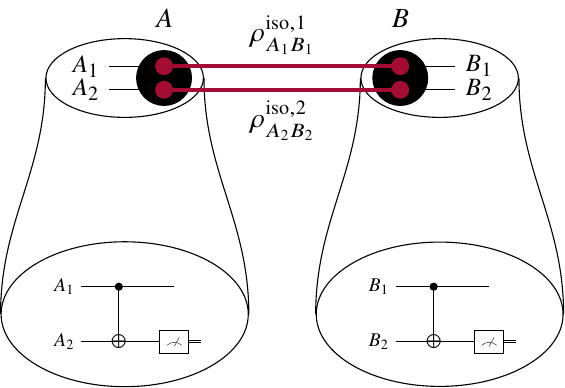}
		\caption{Depiction of the simple entanglement distillation protocol from Ref.~[\onlinecite{BBP96}]. The protocol takes two isotropic states $\rho_{A_jB_j}^{\text{iso},j}$, $j\in\{1,2\}$ (see \eqref{eq-isotropic_state}), and transforms them probabilistically to a state with higher fidelity.}\label{fig-ent_distill_example}
	\end{figure}
		
	We are not concerned with any particular entanglement distillation protocol in this work. All we are concerned with is their mathematical structure. In particular, entanglement distillation protocols that are probabilistic can be described mathematically as an LOCC instrument, which we now demonstrate with a simple example, depicted in Fig.~\ref{fig-ent_distill_example}, which comes from Ref.~[\onlinecite{BBP96}]. In this protocol, Alice and Bob first apply the CNOT gate to their qubits and follow it with a measurement of their second qubit in the standard basis. They then communicate the results of their measurement to each other. The protocol is considered successful if they both obtain the same outcome, and a failure otherwise. This protocol has the following corresponding LOCC instrument channel:
	\begin{align}
		&\mathcal{L}_{A_1A_2B_1B_2\to A_1B_1}\left(\rho_{A_1B_1}^1\otimes\rho_{A_2B_2}^2\right)\nonumber\\
		&\quad=\ketbra{0}{0}\otimes\left((K_{A}^0\otimes K_{B}^1)(\rho_{A_1B_1}^{\text{iso},1}\otimes\rho_{A_2B_2}^{\text{iso},2})(K_{A}^0\otimes K_{B}^1)^{\dagger}\right.\nonumber\\
		&\qquad\quad\left.+(K_{A}^1\otimes K_{B}^0)(\rho_{A_1B_1}^{\text{iso},1}\otimes\rho_{A_2B_2}^{\text{iso},2})(K_{A}^1\otimes K_{B}^0)^{\dagger}\right)\nonumber\\
		&\quad+\ketbra{1}{1}\otimes\left((K_{A}^0\otimes K_{B}^0)(\rho_{A_1B_1}^{\text{iso},1}\otimes\rho_{A_2B_2}^{\text{iso},2})(K_{A}^0\otimes K_{B}^0)^{\dagger}\right.\nonumber\\
		&\qquad\quad\left.+(K_{A}^1\otimes K_{B}^1)(\rho_{A_1B_1}^{\text{iso},1}\otimes\rho_{A_2B_2}^{\text{iso},2})(K_{A}^1\otimes K_{B}^1)^{\dagger}\right),\label{eq-ent_distill_channel}
	\end{align}
	where
	\begin{align}
		K_A^{x}&\equiv K_{A_1A_2\to A_1}^x \coloneqq\bra{x}_{A_2}\text{CNOT}_{A_1A_2}\quad\forall~x\in\{0,1\},\\
		K_B^x&\equiv K_{B_1B_2\to B_1}^x\coloneqq\bra{x}_{B_2}\text{CNOT}_{B_1B_2}\quad\forall~x\in\{0,1\}.
	\end{align}
	Furthermore, the states $\rho_{A_jB_j}^{\text{iso},j}$, $j\in\{1,2\}$, are defined as
	\begin{align}
		\rho_{A_jB_j}^{\text{iso},j}&\coloneqq\mathcal{T}_{A_jB_j}^U(\rho_{A_jB_j}^j)\\
		&\coloneqq\int_U \left(U_{A_j}\otimes\conj{U}_{B_j}\right)(\rho_{A_jB_j}^j)\left(U_{A_j}\otimes\conj{U}_{B_j}\right)^{\dagger},\label{eq-isotropic_state}
	\end{align}
	where $\mathcal{T}^U$ is the \textit{isotropic twirling channel}; see, e.g., Ref.~[\onlinecite[Example~7.25]{Wat18_book}].
		
	It is a straightforward calculation to show that if $f_1=\bra{\Phi}\rho_{A_1B_1}^1\ket{\Phi}$ and $f_2=\bra{\Phi}\rho_{A_2B_2}^2\ket{\Phi}$ are the fidelities of the initial states with the maximally entangled state, then the protocol depicted in Fig.~\ref{fig-ent_distill_example}, with corresponding LOCC channel given by \eqref{eq-ent_distill_channel}, succeeds with probability
	\begin{equation}
		p_{\text{succ}}=\frac{8}{9}f_1f_2-\frac{2}{9}(f_1+f_2)+\frac{5}{9},
	\end{equation}
	and the fidelity of the output state $\sigma_{A_1B_1}$ with the maximally entangled state (conditioned on success) is
	\begin{equation}
		\bra{\Phi}\sigma_{A_1B_1}\ket{\Phi}=\frac{1}{p_{\text{succ}}}\left(\frac{10}{9}f_1f_2-\frac{1}{9}(f_1+f_2)+\frac{1}{9}\right).
	\end{equation}
		
	The above example illustrates a general principle, which is that entanglement distillation protocols that are probabilistic (and heralded) can be described using LOCC instrument channels. Specifically, let $G=(V,E)$ be the graph corresponding to the physical links in a quantum network. Given an element $e\in E$ with $n$ parallel edges $e^1,e^2,\dotsc,e^n$, every probabilistic entanglement distillation protocol has the form of an LOCC instrument channel of the following form:
	\begin{multline}
		\mathcal{D}_{e^1\dotsb e^n\to e^1\dotsb e^{n'}}^{e}(\cdot)=\ketbra{0}{0}\otimes\mathcal{D}_{e^1\dotsb e^n\to e^1\dotsb e^{n'}}^{e;0}(\cdot)\\+\ketbra{1}{1}\otimes\mathcal{D}_{e^1\dotsb e^n\to e^1\dotsb e^{n'}}^{e;1}(\cdot),
	\end{multline}
	where $\mathcal{D}_{e^1\dotsb e^n\to e^1\dotsb e^{n'}}^{e;0}$ and $\mathcal{D}_{e^1\dotsb e^n\to e^1\dotsb e^{n'}}^{e;1}$ are completely positive trace non-increasing LOCC maps such that $\mathcal{D}_{e^1\dotsb e^n\to e^1\dotsb e^{n'}}^{e;0}+\mathcal{D}_{e^1\dotsb e^n\to e^1\dotsb e^{n'}}^{e;1}$ is a trace-preserving map, and thus an LOCC quantum channel. Specifically, $\mathcal{D}_{e^1\dotsb e^n\to e^1\dotsb e^{n'}}^{e;0}$ corresponds to failure of the protocol and $\mathcal{D}_{e^1\dotsb e^n\to e^1\dotsb e^{n'}}^{e;1}$ corresponds to success of the protocol.

\subsection{Joining protocols}\label{sec-joining_protocols}

	Let us now discuss joining protocols, such as entanglement swapping. We can describe such protocols using LOCC instrument channels, just as with entanglement distillation protocols. As above, let $G=(V,E)$ be the graph corresponding to the physical links in a quantum network. A path in a graph is a sequence $w=(v_1,e_1,v_2,e_2,\dotsc,e_{n-1},v_n)$ of vertices and edges that specifies how to get from the vertex $v_1$ to the vertex~$v_n$. Given a path $w$ of active elementary links in the network, the joining channel $\mathcal{L}_{w\to e'}$ that forms the new virtual link $e'$ is given in the probabilistic setting by
	\begin{equation}\label{eq-q_instr_channel_swapping}
		\mathcal{L}_{w\to e'}(\cdot)=\ketbra{0}{0}\otimes\mathcal{L}_{w\to e'}^0(\cdot)+\ketbra{1}{1}\otimes\mathcal{L}_{w\to e'}^1(\cdot),
	\end{equation}
	where $\mathcal{L}_{w\to e'}^0$ and $\mathcal{L}_{w\to e'}^1$ are completely positive trace non-increasing LOCC maps such that $\mathcal{L}_{w\to e'}^0+\mathcal{L}_{w\to e'}^1$ is a trace-preserving map, and thus an LOCC quantum channel. Specifically, $\mathcal{L}_{w\to e'}^0$ corresponds to failure of the joining protocol and $\mathcal{L}_{w\to e'}^1$ corresponds to success of the joining protocol. Given an input state $\rho_w$ corresponding to the given path $w$, the success probability of the joining protocol is $p_{\text{succ}}=\Tr\!\left[\mathcal{L}_{w\to e'}^1(\rho_w)\right]$, and the state conditioned on success is
	\begin{equation}\label{eq-network_QDP_virtual_link_fid_cond}
		\frac{1}{p_{\text{succ}}}\mathcal{L}_{w\to e'}^1(\rho_w).
	\end{equation}
	Note that as input states to the maps $\mathcal{L}_{w\to e'}^0$ and $\mathcal{L}_{w\to e'}^1$ we could have arbitrary states of the elementary links along the path $w$. In particular, depending on the elementary link policy, they could be states of the form \eqref{eq-elem_link_initial_cq_state_2}, which take into account the noise in the quantum memories and other device imperfections arising during the process of generating the elementary links.
	
	The precise joining protocol, and thus the explicit form for the maps $\mathcal{L}_{w\to e'}^0$ and $\mathcal{L}_{w\to e'}^1$, depends on the type of entanglement that is to be created. For bipartite entanglement, we consider entanglement swapping in Sec.~\ref{ex-ent_swap}. For tripartite GHZ entanglement, we describe a protocol in Sec.~\ref{ex-GHZ_ent_swap}, and for multipartite graph states we describe a protocol in Sec.~\ref{ex-graph_state_dist}.

\subsubsection{Entanglement swapping protocol}\label{ex-ent_swap}

	\begin{figure}
		\centering
		\includegraphics[scale=1]{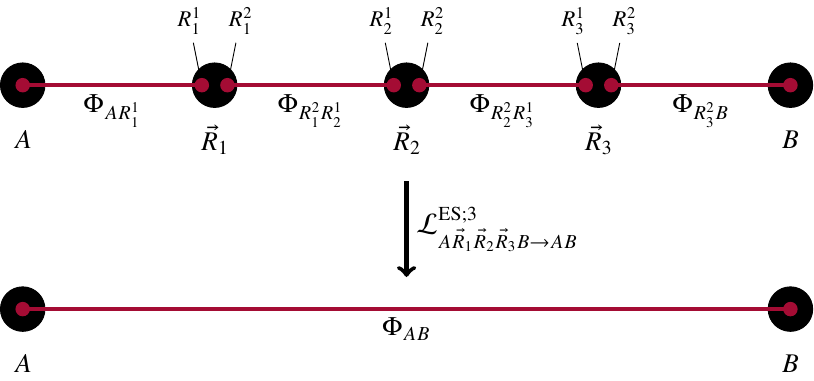}
		\caption{A chain of five nodes corresponding to the entanglement swapping protocol with $n=3$ intermediate nodes. The red lines represent maximally entangled states. The goal of the entanglement swapping protocol is to establish entanglement between $A$ and $B$. The protocol proceeds by first performing a Bell-basis measurement on the systems at the nodes $\vec{R}_j$, $1\leq j\leq n$, and communicating the results of the measurement to $B$, who applies a correction operation based on the outcomes.}\label{fig-ent_swap_chain}
	\end{figure}
	
	Let $\rho_{A\vec{R}_1\vec{R}_2\dotsb\vec{R}_nB}$ be a multipartite quantum state, where $n\geq 1$ and $\vec{R}_j\equiv R_j^1R_j^2$ is an abbreviation for two the quantum systems $R_j^1$ and $R_j^2$. The entanglement swapping protocol with $n$ intermediate nodes is defined by a Bell-basis measurement of the systems $\vec{R}_j$, i.e., a measurement described by the POVM $\{\Phi^{z,x}:z,x\in[d]\}$, where $[d]=\{0,1,\dotsc,d-1\}$, $\Phi^{z,x}=\ketbra{\Phi^{z,x}}{\Phi^{z,x}}$, and
	\begin{equation}
		\ket{\Phi^{z,x}}\coloneqq(Z^zX^x\otimes\mathbbm{1})\ket{\Phi}
	\end{equation}
	are the qudit Bell state vectors, with
	\begin{equation}\label{eq-max_ent_state}
		\ket{\Phi}\coloneqq\frac{1}{\sqrt{d}}\sum_{k=0}^{d-1}\ket{k,k}.
	\end{equation}
	The operators $Z$ and $X$ are the discrete Weyl operators~\cite{Wat18_book}, which are defined as
	\begin{equation}
		Z\coloneqq\sum_{k=0}^{d-1}\e^{\frac{2\pi\I k}{d}}\ketbra{k}{k},\quad X\coloneqq\sum_{k=0}^{d-1}\ketbra{k+1}{k}.
	\end{equation}
	Conditioned on the outcomes $(z_j,x_j)$ of the Bell measurement on $\vec{R}_j$, the unitary $Z_B^{z_1+\dotsb+z_n}X_B^{x_1+\dotsb+x_n}$ is applied to the system $B$, where the addition is performed modulo $d$. Let $\vec{z},\vec{x}\in[d]^{\times n}$, and define
	\begin{align}
		M_{\vec{R}_1\vec{R}_2\dotsb\vec{R}_n}^{\vec{z},\vec{x}}&\coloneqq \Phi_{\vec{R}_1}^{z_1,x_1}\otimes\Phi_{\vec{R}_2}^{z_2,x_2}\otimes\dotsb\otimes\Phi_{\vec{R}_n}^{z_n,x_n},\\[0.2cm]
		W_B^{\vec{z},\vec{x}}&\coloneqq Z_B^{z_1+\dotsb+z_n}X_B^{x_1+\dotsb+x_n},
	\end{align}
	where the addition in the second line is performed modulo $d$. Then, the LOCC quantum channel corresponding to the entanglement swapping protocol with $n\geq 1$ intermediate nodes is
	\begin{multline}\label{eq-ent_swap_channel}
		\mathcal{L}_{A\vec{R}_1\dotsb\vec{R}_nB\to AB}^{\ES;n}\left(\rho_{A\vec{R}_1\dotsb\vec{R}_nB}\right)\\\coloneqq\sum_{\vec{z},\vec{x}\in[d]^{\times n}}\Tr_{\vec{R}_1\dotsb\vec{R}_n}\!\!\left[M_{\vec{R}_1\dotsb\vec{R}_n}^{\vec{z},\vec{x}}W_B^{\vec{z},\vec{x}}\right.\\\left.\left(\rho_{A\vec{R}_1\dotsb\vec{R}_nB}\right)\left(W_B^{\vec{z},\vec{x}}\right)^\dagger\right].
	\end{multline}
	
	The standard entanglement swapping protocol~\cite{ZZH93} corresponds to the input state
	\begin{equation}
		\rho_{A\vec{R}_1\vec{R}_2\dotsb\vec{R}_nB}=\Phi_{AR_1^1}\otimes \Phi_{R_1^2R_2^1}\otimes\dotsb\otimes\Phi_{R_{n-1}^2R_n^1}\otimes\Phi_{R_n^2B}.
	\end{equation}
	This scenario is shown in Fig.~\ref{fig-ent_swap_chain}. Indeed, it can be shown that 
	\begin{multline}
		\mathcal{L}_{A\vec{R}_1\dotsb\vec{R}_nB\to AB}^{\ES;n}\left(\Phi_{AR_1^1}\otimes \Phi_{R_1^2R_2^1}\otimes\dotsb\right.\\\left.\otimes\Phi_{R_{n-1}^2R_n^1}\otimes\Phi_{R_n^2B}\right)=\Phi_{AB}.
	\end{multline}
	Furthermore, the standard teleportation protocol~\cite{BBC+93} corresponds to $n=1$ and the input state
	\begin{equation}
		\rho_{A\vec{R}_1B}=\sigma_{R_1^1}\otimes\Phi_{R_1^2B},
	\end{equation}
	where $A=\varnothing$ is a trivial (one-dimensional) system and $\sigma_{R_1^1}$ is an arbitrary $d$-dimensional quantum state, so that
	\begin{equation}
		\mathcal{L}_{\vec{R}_1\to B}^{\ES;1}(\sigma_{R_1^1}\otimes\Phi_{R_1^2B})=\sigma_B,
	\end{equation}
	as expected.
	
	\begin{proposition}[Fidelity after entanglement swapping]\label{prop-ent_swap_post_fid}
		For all $n\geq 1$ and all states $\rho_{AR_1^1}^1,\rho_{R_1^2R_2^1}^2,\dotsc,\rho_{R_n^2B}^{n+1}$, the fidelity of the maximally entangled state with the state after entanglement swapping of $\rho_{AR_1^1}^1,\rho_{R_1^2R_2^1}^2,\dotsc,\rho_{R_n^2B}^{n+1}$ is given by
		\begin{multline}\label{eq-ent_swap_post_fidelity}
			\bra{\Phi}_{AB}\mathcal{L}_{A\vec{R}_1\dotsb\vec{R}_nB\to AB}^{\ES;n}\left(\rho_{AR_1^1}^1\otimes\rho_{R_1^2R_2^1}^2\otimes\dotsb\otimes\rho_{R_n^2B}^{n+1}\right)\ket{\Phi}_{AB}\\=\sum_{\vec{z},\vec{x}\in[d]^{\times n}}^{d-1}\bra{\Phi^{z',x'}}\rho_{AR_1^1}^1\ket{\Phi^{z',x'}}\bra{\Phi^{z_1,x_1}}\rho_{R_1^2R_2^1}^2\ket{\Phi^{z_1,x_1}}\\\dotsb\bra{\Phi^{z_n,x_n}}\rho_{R_n^2B}^{n+1}\ket{\Phi^{z_n,x_n}},
		\end{multline}
		where $z'=-z_1-z_2-\dotsb-z_n$ and $x'=-x_1-x_2-\dotsb-x_n$.
	\end{proposition}
	
	\begin{proof}
		See Appendix~\ref{app-ent_swap_post_fid_pf}.
	\end{proof}
	
	We remark that a formula for the fidelity after entanglement swapping of two arbitrary bipartite qubit states can be found in Ref.~[\onlinecite{kirby2016entanglementswapping}].

\subsubsection{GHZ entanglement swapping protocol}\label{ex-GHZ_ent_swap}
	
	The previous example takes a chain of Bell states and transforms them into a Bell state shared by the end nodes of the chain. In this example, we look at a protocol that takes the same chain of Bell states and transforms them instead to a multi-qubit GHZ state, which is defined as~\cite{GHZ89}
	\begin{equation}\label{eq-GHZ_state}
		\ket{\GHZ_n}\coloneqq\frac{1}{\sqrt{2}}(\ket{0}^{\otimes n}+\ket{1}^{\otimes n}).
	\end{equation}
	We call this protocol the \textit{GHZ entanglement swapping protocol}.
	
	The protocol for transforming a chain of two Bell states to a three-party GHZ state is shown in Fig.~\ref{fig-ent_swap_GHZ}. First, the two qubits $R_1^1$ and $R_1^2$ in the central node are entangled with a CNOT gate, followed by a measurement of $R_1^2$ in the standard basis (with corresponding POVM $\{\ketbra{0}{0},\ketbra{1}{1}\}$). The result $x\in\{0,1\}$ is communicated to $B$, where the correction operation $X_B^x$ is applied. The LOCC channel corresponding to this protocol is
	\begin{multline}\label{eq-GHZ_ent_swap_channel_1}
		\mathcal{L}_{A\vec{R}_1B}^{\GHZ;1}\left(\rho_{A\vec{R}_1B}\right)\\=\sum_{x=0}^1 \left(K_{\vec{R}_1}^x\otimes X_B^x\right)\rho_{A\vec{R}_1B}\left(K_{\vec{R}_1}^x\otimes X_B^x\right)^\dagger,
	\end{multline}
	where
	\begin{align}
		K_{\vec{R}_1}^x&\coloneqq\bra{x}_{R_1^2}\text{CNOT}_{\vec{R}_1},\\
		\text{CNOT}_{\vec{R}_1}&\coloneqq\ketbra{0}{0}_{R_1^1}\otimes\mathbbm{1}_{R_1^2}+\ketbra{1}{1}_{R_1^1}\otimes X_{R_1^2}.
	\end{align}
	
	\begin{figure}
		\centering
		\includegraphics[scale=1]{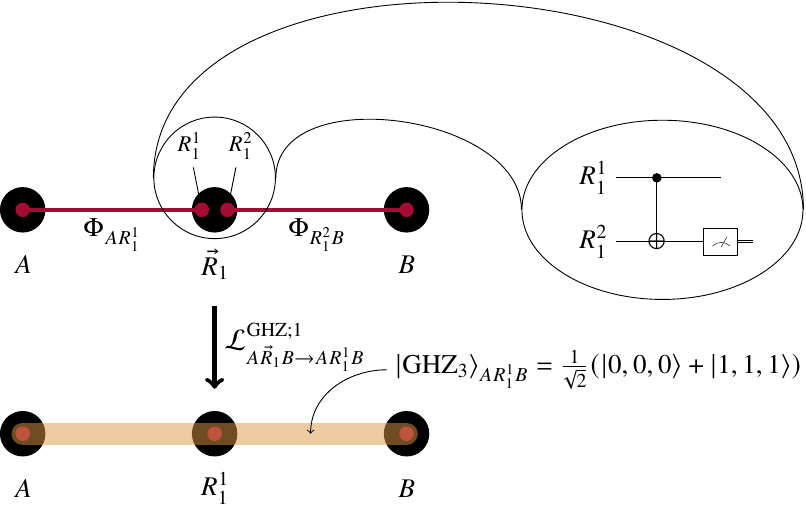}
		\caption{The GHZ entanglement swapping protocol with one intermediate node. The two qubits in the central node are entangled using the CNOT gate, after which the qubit $R_1^2$ is measured in the standard basis. The result $x\in\{0,1\}$ of the measurement is communicated to $B$, where the gate $X_B^x$ is applied.}\label{fig-ent_swap_GHZ}
	\end{figure}
	
	The protocol shown in Fig.~\ref{fig-ent_swap_GHZ}, with corresponding LOCC quantum channel in \eqref{eq-GHZ_ent_swap_channel_1}, can be easily extended to a scenario with $n>1$ intermediate nodes. In this case, the node $\vec{R}_1$ starts by applying the gate $\text{CNOT}_{\vec{R}_1}$ to its qubits and then measuring the qubit $R_1^2$ in the standard basis. The outcome of this measurement is sent to the node $\vec{R}_2$, and the corresponding correction operation is applied to the qubit $R_2^1$. Then, the gate $\text{CNOT}_{\vec{R}_2}$ is applied to the qubits at $\vec{R}_2$, followed by a standard-basis measurement of $R_2^2$ and communication of the outcome to $\vec{R}_3$ and a correction operation on $R_3^1$. This proceeds in sequence until the $n^{\text{th}}$ intermediate node $\vec{R}_n$, which sends its measurement outcome to $B$, which applies the appropriate correction operation. The LOCC channel for this protocol is
	\begin{multline}\label{eq-GHZ_ent_swap_channel}
		\mathcal{L}_{A\vec{R}_1\dotsb\vec{R}_nB\to AR_1^1\dotsb R_n^1B}^{\GHZ;n}\left(\rho_{A\vec{R}_1\dotsb\vec{R}_nB}\right)\\\coloneqq\sum_{\vec{x}\in\{0,1\}^n} P_{\vec{R}_1\dotsb\vec{R}_nB}^{\vec{x}}\left(\rho_{A\vec{R}_1\dotsb\vec{R}_nB}\right)P_{\vec{R}_1\dotsb\vec{R}_nB}^{\vec{x}~\dagger},
	\end{multline}
	where
	\begin{equation}
		P_{\vec{R}_1\dotsb\vec{R}_nB}^{\vec{x}}\coloneqq K_{\vec{R}_1}^{x_1}\otimes K_{\vec{R}_2}^{x_2}X_{R_2^1}^{x_1}\otimes\dotsb\otimes K_{\vec{R}_n}^{x_n}X_{R_n^1}^{x_{n-1}}\otimes X_B^{x_n}
	\end{equation}
	for all $\vec{x}\in\{0,1\}^n$. If the input state to this channel is
	\begin{equation}
		\rho_{A\vec{R}_1\dotsb\vec{R}_nB}=\Phi_{AR_1^1}\otimes \Phi_{R_1^2R_2^1}\otimes\dotsb\otimes\Phi_{R_{n-1}^2R_n^1}\otimes\Phi_{R_n^2B},
	\end{equation}
	then the output is a $(n+2)$-party GHZ state given by the state vector $\ket{\GHZ_{n+2}}_{AR_1^1\dotsb R_n^1B}$ as defined in \eqref{eq-GHZ_state}, i.e.,
	\begin{multline}
		\mathcal{L}_{A\vec{R}_1\dotsb\vec{R}_nB\to AR_1^1\dotsb R_n^1B}^{\GHZ;n}\left(\Phi_{AR_1^1}\otimes \Phi_{R_1^2R_2^1}\otimes\dotsb\right.\\\left.\otimes\Phi_{R_{n-1}^2R_n^1}\otimes\Phi_{R_n^2B}\right)=\ketbra{\GHZ_{n+2}}{\GHZ_{n+2}}.
	\end{multline}
	
	\begin{proposition}[Fidelity after GHZ entanglement swapping]\label{prop-GHZ_ent_swap_post_fid}
		For all $n\geq 1$, and for all states $\rho_{AR_1^1}^1,\rho_{R_1^2R_2^1}^2,\dotsc,\rho_{R_n^2B}^{n+1}$, the fidelity of the $(n+2)$-party GHZ state with the state after the GHZ entanglement swapping of $\rho_{AR_1^1}^1,\rho_{R_1^2R_2^1}^2,\dotsc,\rho_{R_n^2B}^{n+1}$ is
		\begin{multline}
			\bra{\GHZ_{n+2}}\mathcal{L}_{A\vec{R}_1\dotsb\vec{R}_nB\to AR_1^1\dotsb R_n^1B}^{\GHZ;n}\left(\rho_{AR_1^1}^1\otimes\right.\\\left.\rho_{R_1^2R_2^1}^2\otimes\dotsb\otimes\rho_{R_n^2B}^{n+1}\right)\ket{\GHZ_{n+2}}\\=\sum_{z_1,\dotsc,z_n=0}^1\bra{\Phi^{z_1+\dotsb+z_n,0}}\rho_{AR_1^1}\ket{\Phi^{z_1+\dotsb+z_n,0}}\\\bra{\Phi^{z_1,0}}\rho_{R_1^2R_2^1}^2\ket{\Phi^{z_1,0}}\dotsb\bra{\Phi^{z_n,0}}\rho_{R_n^2B}^{n+1}\ket{\Phi^{z_n,0}}.
		\end{multline}
	\end{proposition}
	
	\begin{proof}
		See Appendix~\ref{app-GHZ_ent_swap_post_fid_pf}.
	\end{proof}

\subsubsection{Graph state distribution protocol}\label{ex-graph_state_dist}

	We now consider an example of distributing an arbitrary graph state, which can be viewed as a special case of the procedure considered in Ref.~[\onlinecite{MMG19}], and it has been shown explicitly in Ref.~[\onlinecite[Sec.~III.B]{CC12}]. A graph state~\cite{BR01,RB01,Briegel09} is a multi-qubit quantum state defined using graphs. 
		
	Consider a graph $G=(V,E)$, which consists of a set $V$ of vertices and a set $E$ of edges. For the purposes of this example, $G$ is an undirected graph, and $E$ is a set of two-element subsets of $V$. The graph state $\ket{G}$ is an $n$-qubit quantum state $\ket{G}_{A_1\dotsb A_n}$, with $n=|V|$, that is defined as
	\begin{equation}
		 \ket{G}_{A_1\dotsb A_n}\coloneqq\frac{1}{\sqrt{2^n}}\sum_{\vec{\alpha}\in\{0,1\}^n}(-1)^{\frac{1}{2}\vec{\alpha}^{\t}A(G)\vec{\alpha}}\ket{\vec{\alpha}},
	\end{equation}
	where $A(G)$ is the adjacency matrix of $G$, which is defined as
	\begin{equation}
		A(G)_{i,j}=\left\{\begin{array}{l l} 1 & \text{if } \{v_i,v_j\}\in E,\\ 0 & \text{otherwise}, \end{array}\right.
	\end{equation}
	and $\vec{\alpha}$ is the column vector $(\alpha_1,\dotsc,\alpha_n)^{\t}$. It is easy to show that
	\begin{equation}
		\ket{G}_{A_1\dotsb A_n}=\text{CZ}(G)(\ket{+}_{A_1}\otimes\dotsb\otimes\ket{+}_{A_n}),
	\end{equation}
	where $\ket{+}\coloneqq\frac{1}{\sqrt{2}}(\ket{0}+\ket{1})$ and
	\begin{equation}
		\text{CZ}(G)\coloneqq\bigotimes_{\{v_i,v_j\}\in E}\text{CZ}_{A_iA_j},
	\end{equation}
	with $\text{CZ}_{A_iA_j}\coloneqq\ketbra{0}{0}_{A_i}\otimes\mathbbm{1}_{A_j}+\ketbra{1}{1}_{A_i}\otimes Z_{A_j}$ being the controlled-$Z$ gate.
	
	\begin{figure}
		\centering
		\includegraphics[scale=1]{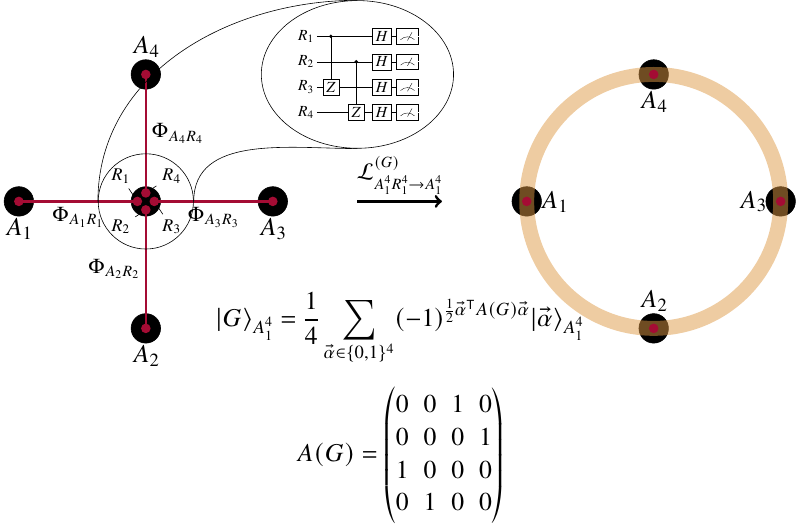}
		\caption{Depiction of a protocol for distributing a graph state among four nodes $A_1,A_2,A_3,A_4$, all of which initially share Bell states with the central node.}\label{fig-graph_state_dist}
	\end{figure}
	
	Now, consider the scenario depicted in Fig.~\ref{fig-graph_state_dist}, in which $n=4$ nodes share Bell states with a central node. The task is for the central node to distribute the graph state $\ket{G}$ to the outer nodes. One possible procedure is for the central node to locally prepare the graph state and then to teleport the individual qubits using the Bell states. However, it is possible to perform a slightly simpler procedure that does not require the additional qubits needed to prepare the graph state locally. In fact, the following deterministic procedure produces the required graph state $\ket{G}$ shared by the nodes $A_1,\dotsc,A_n$.
	\begin{enumerate}
		\item The central node applies $\text{CZ}(G)$ to the qubits $R_1,\dotsc,R_n$.
		\item On each of the qubits $R_1,\dotsc,R_n$, the central node performs the measurement defined by the POVM $\{\ketbra{+}{+},\ketbra{-}{-}\}$, where $\ket{\pm}=\frac{1}{\sqrt{2}}(\ket{0}\pm\ket{1})$. The outcome is an $n$-bit string $\vec{x}=(x_1,\dotsc,x_n)$, where $x_i=0$ corresponds to the ``$+$'' outcome and $x_i=1$ corresponds to the ``$-$'' outcome. The central node communicates outcome $x_i$ to the node $A_i$.
		\item The nodes $A_i$ apply $Z^{x_i}$ to their qubit. In other words, if $x_i=0$, then $A_i$ does nothing, and if $x_i=1$, then $A_i$ applies $Z$ to their qubit.
	\end{enumerate}
	Let us prove that this protocol achieves the desired outcome. First, observe that
	\begin{multline}
		\ket{\Phi}_{A_1R_1}\otimes\dotsb\otimes\ket{\Phi}_{A_nR_n}\\=\frac{1}{\sqrt{2^n}}\sum_{\vec{\alpha}\in\{0,1\}^n}\ket{\vec{\alpha}}_{A_1\dotsb A_n}\ket{\vec{\alpha}}_{R_1\dotsb R_n}.
	\end{multline}
	Then, after the first step, the state is
	\begin{multline}
		\frac{1}{\sqrt{2^n}}\sum_{\vec{\alpha}\in\{0,1\}^n}\ket{\vec{\alpha}}_{A_1\dotsb A_n}\text{CZ}(G)\ket{\vec{\alpha}}_{R_1\dotsb R_n}\\=\frac{1}{\sqrt{2^n}}\sum_{\vec{\alpha}\in\{0,1\}^n}(-1)^{\frac{1}{2}\vec{\alpha}^{\t}A(G)\vec{\alpha}}\ket{\vec{\alpha}}_{A_1\dotsb A_n}\ket{\vec{\alpha}}_{R_1\dotsb R_n},
	\end{multline}
	where we have used the fact that
	\begin{align}
		\text{CZ}(G)\ket{\vec{\alpha}}&=(-1)^{\sum_{i,j:\{v_i,v_j\}\in E}\alpha_i\alpha_j}\ket{\vec{\alpha}}\\
		&=(-1)^{\frac{1}{2}\vec{\alpha}^{\t}A(G)\vec{\alpha}}\ket{\vec{\alpha}}.
	\end{align}
	Then, we find that for every outcome string $(x_1,\dotsc,x_n)$ of the measurement on the qubits $R_1,\dotsc,R_n$ the corresponding (unnormalized) post-measurement state is
	\begin{equation}
		\frac{1}{2^n}\sum_{\vec{\alpha}\in\{0,1\}^n}(-1)^{\frac{1}{2}\vec{\alpha}^{\t}A(G)\vec{\alpha}}(-1)^{\alpha_1x_1+\dotsb+\alpha_nx_n}\ket{\vec{\alpha}}_{A_1\dotsb A_n}.
	\end{equation}
	Then, using the fact that $Z^x\ket{\alpha}=(-1)^{\alpha x}\ket{\alpha}$ for all $x,\alpha\in\{0,1\}$, we find that at the end of the second step the (unnormalized) state is
	\begin{multline}
		\frac{1}{2^n}(Z_{A_1}^{x_1}\otimes\dotsb\otimes Z_{A_n}^{x_n})\sum_{\vec{\alpha}\in\{0,1\}^n}(-1)^{\frac{1}{2}\vec{\alpha}^{\t}A(G)\vec{\alpha}}\ket{\vec{\alpha}}_{A_1\dotsb A_n}\\=\frac{1}{\sqrt{2^n}}(Z_{A_1}^{x_1}\otimes\dotsb\otimes Z_{A_n}^{x_n})\ket{G}_{A_1\dotsb A_n}
	\end{multline}
	for all $(x_1,\dotsc,x_n)\in\{0,1\}^n$. From this, we see that, up to local Pauli-$z$ corrections, the post-measurement state is equal to the desired graph state $\ket{G}$ with probability $\frac{1}{2^n}$ for every measurement outcome string $(x_1,\dotsc,x_n)$. Once all of the nodes $A_i$ receive their corresponding outcome $x_i$ and apply the correction $Z_{A_i}^{x_i}$, the nodes $A_1,\dotsc,A_n$ share the graph state $\ket{G}$. As a result of the classical communication of the measurement outcomes and the subsequent correction operations, the protocol is deterministic.
	
	The protocol described above has the following representation as an LOCC channel:
	\begin{multline}
		\mathcal{L}_{A_1^nR_1^n\to A_1^n}^{(G)}\left(\rho_{A_1^nR_1^n}\right)\\\coloneqq\sum_{\vec{x}\in\{0,1\}^n}\left(Z_{A_1^n}^{\vec{x}}\otimes\bra{\vec{x}}_{R_1^n}H^{\otimes n}\text{CZ}(G)_{R_1^n}\right)\left(\rho_{A_1^nR_1^n}\right)\\\left(Z_{A_1^n}^{\vec{x}}\otimes \text{CZ}(G)_{R_1^n}^{\dagger}H^{\otimes n}\ket{\vec{x}}_{R_1^n}\right),
	\end{multline}
	for every state $\rho_{A_1^nR_1^n}$, where $H=\ketbra{+}{0}+\ketbra{-}{1}$ is the Hadamard operator, and we have let
	\begin{equation}
		Z_{A_1\dotsb A_n}^{\vec{x}}\coloneqq Z_{A_1}^{x_1}\otimes\dotsb\otimes Z_{A_n}^{x_n}.
	\end{equation}
	We have also used the abbreviation $A_1^n\equiv A_1A_2\dotsb A_n$, and similarly for $R_1^n$. Using the fact that
	\begin{equation}\label{eq-graph_state_x_0}
		\text{CZ}(G)H^{\otimes n}\ket{\vec{x}}=Z^{\vec{x}}\ket{G}
	\end{equation}
	for all $\vec{x}\in\{0,1\}^n$, and letting
	\begin{equation}\label{eq-graph_state_x}
		\ket{G^{\vec{x}}}\coloneqq Z^{\vec{x}}\ket{G},
	\end{equation}
	we can write the channel in the following simpler form:
	\begin{multline}\label{eq-graph_state_dist_channel}
		\mathcal{L}_{A_1^nR_1^n\to A_1^n}^{(G)}\left(\rho_{A_1^nR_1^n}\right)\\=\sum_{\vec{x}\in\{0,1\}^n} \left(Z_{A_1^n}^{\vec{x}}\otimes\bra{G^{\vec{x}}}_{R_1^n}\right)\left(\rho_{A_1^nR_1^n}\right)\left(Z_{A_1^n}^{\vec{x}}\otimes\ket{G^{\vec{x}}}_{R_1^n}\right).
	\end{multline}
	From this, we see that the protocol can be thought of as measuring the systems $R_1,\dotsc, R_n$ according to the POVM $\left\{\ketbra{G^{\vec{x}}}{G^{\vec{x}}}\right\}_{\vec{x}\in\{0,1\}^n}$ and, conditioned on the outcome $\vec{x}$, applying the correction operation $Z^{\vec{x}}$ to the systems $A_1,\dotsc, A_n$. Note that $\left\{\ketbra{G^{\vec{x}}}{G^{\vec{x}}}\right\}_{\vec{x}\in\{0,1\}^n}$ is indeed a POVM due to the fact that
	\begin{equation}
		\ket{G^{\vec{x}}}=\text{CZ}(G)H^{\otimes n}\ket{\vec{x}}
	\end{equation}
	for all $\vec{x}\in\{0,1\}^n$, which follows from \eqref{eq-graph_state_x_0} and \eqref{eq-graph_state_x}, so that
	\begin{multline}
		\sum_{\vec{x}\in\{0,1\}^n}\ketbra{G^{\vec{x}}}{G^{\vec{x}}}\\=\text{CZ}(G)H^{\otimes n}\underbrace{\sum_{\vec{x}\in\{0,1\}^n}\ketbra{\vec{x}}{\vec{x}}}_{\mathbbm{1}}H^{\otimes n}\text{CZ}(G)^{\dagger}=\mathbbm{1}.
	\end{multline}
	
	\begin{remark}
		The set $\{\ket{G^{\vec{x}}}\}_{\vec{x}\in\{0,1\}^n}$ of state vectors defined in Eq.~\eqref{eq-graph_state_x} has been presented in Ref.~[\onlinecite[Sec.~II.A]{CC12}] and explicitly called the ``graph state basis'', labeled as $\ket{\boldsymbol{\mu}}$, for $\boldsymbol{\mu}\in\{0,1\}^n$.
	\end{remark}
	
	\begin{proposition}[Fidelity after graph state distribution]\label{prop-graph_state_dist_post_fid}
		For all $n\geq 2$, every graph $G$ with $n$ vertices, and all two-qubit states $\rho_{A_1R_1}^1$, $\rho_{A_2R_2}^2,\dotsc,\allowbreak\rho_{A_nR_n}^n$, the fidelity of the graph state $\ket{G}$ with the state after the graph state distribution protocol applied to $\rho_{A_1R_1}^1$, $\rho_{A_2R_2}^2,\dotsc,\allowbreak\rho_{A_nR_n}^n$ is
		\begin{multline}
			\bra{G}\mathcal{L}_{A_1^nR_1^n\to A_1^n}^{(G)}\left(\rho_{A_1R_1}^1\otimes\dotsb\otimes\rho_{A_nR_n}^n\right)\ket{G}\\=\sum_{\vec{x}\in\{0,1\}^n}\bra{\Phi^{z_1,x_1}}\rho_{A_1R_1}^1\ket{\Phi^{z_1,x_1}}\bra{\Phi^{z_2,x_2}}\rho_{A_2R_2}^2\ket{\Phi^{z_2,x_2}}\\\dotsb\bra{\Phi^{z_n,x_n}}\rho_{A_nR_n}^n\ket{\Phi^{z_n,x_n}},
		\end{multline}
		where the column vector $\vec{z}=(z_1,\dotsc,z_n)^{\t}$ is given by $\vec{z}=A(G)\vec{x}$, with $A(G)$ the adjacency matrix of $G$.
	\end{proposition}
	
	\begin{proof}
		See Appendix~\ref{app-graph_state_dist_post_fid_pf}.	
	\end{proof}

\section{Analysis of a quantum network protocol}\label{sec-network_protocol}

	In the previous two sections, we described in detail how to model elementary links in a quantum network using Markov decision processes. Then, we showed how to model entanglement distillation protocols and joining protocols (such as entanglement swapping) as LOCC channels. The upshot of these developments is that they give us a method for determining the quantum states of elementary and virtual links in a quantum network that depend explicitly on the underlying device parameters and noise processes that characterize the device, thereby allowing us to perform a more realistic analysis of entanglement distribution protocols, as we now show in this section.
	
	In this section, we analyze a simple entanglement distribution protocol. Recall from Sec.~\ref{sec-introduction} that entanglement distribution refers to the task of creating virtual links---entanglement between non-adjacent nodes---from elementary links, which are entangled states shared by adjacent (physically connected) nodes. An entanglement distribution protocol can be thought of as a graph transformation, as done in Refs.~[\onlinecite{SMI+17,CRDW19}] and depicted in Fig.~\ref{fig-network_physical}. Starting with the graph $G=(V,E)$ of physical links in the network, the goal is to realize a new graph $G_{\text{target}}=(V,E_{\text{target}})$ consisting of virtual links in addition to elementary links, such as the graph in the right-most panel of Fig.~\ref{fig-network_physical}.
	
	The protocol that we consider consists of two steps: generate elementary links, and then perform joining protocols based on the given target graph. The protocol is described more formally in Fig.~\ref{fig-QDP_protocol}. Starting with the graph $G=(V,E)$ of elementary links, all of the elementary links independently undergo policies $\pi_{e}$, with $e\in E$. After $t\geq 1$ time steps, an algorithm finds paths for creating the virtual links specified by the target graph $G_{\text{target}}$ and the corresponding joining protocols are performed. If entire target network cannot be achieved in $t$ time steps, then a decision is made to either conclude the protocol with the current configuration or to continue for another $t$ time steps under the same policies.
	
	\begin{figure}
		\centering
		\includegraphics[scale=1]{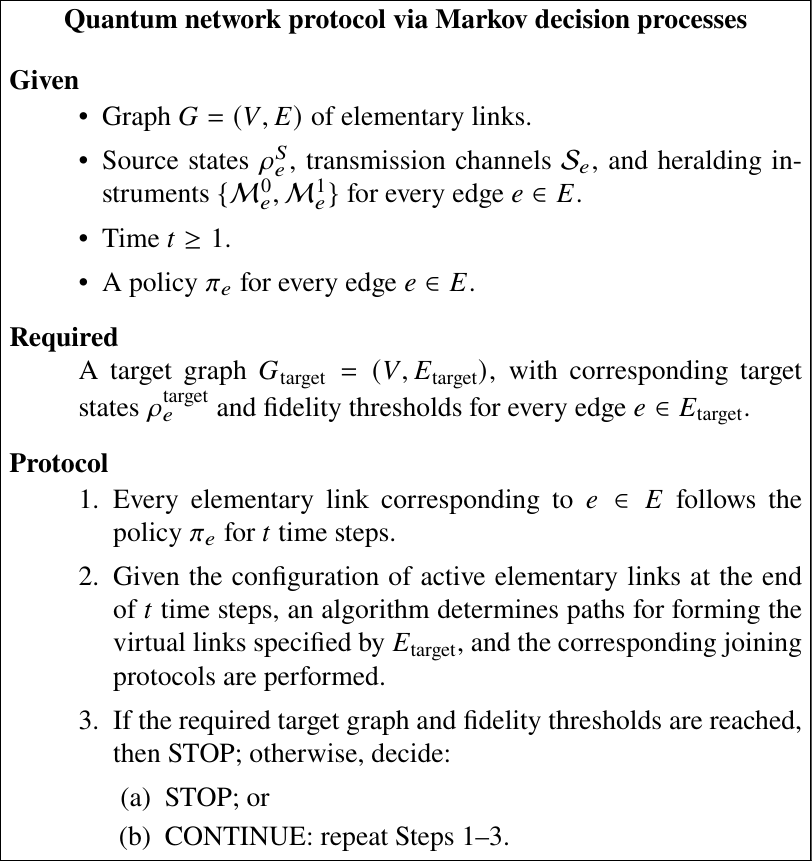}
		\caption{Outline of a quantum network protocol based on Markov decision processes. Every elementary link in the network follows a policy for $t\geq 1$ time steps. At the end of the $t$ time steps, the appropriate paths in the network are found and the corresponding joining protocols are performed in order to achieve the network corresponding to the target graph $G_{\text{target}}$.}\label{fig-QDP_protocol}
	\end{figure}
	
	\begin{remark}\label{rem-QDP_network_extensions}
		Note that in the protocol described in Fig.~\ref{fig-QDP_protocol}, the virtual links are created only when all of the required elementary links are active. This is of course not the most general procedure, because it is in general possible to join some of the elementary links along a path while waiting for the others to become active. To handle such general procedures requires developing MDPs for systems of multiple elementary links. While this is the subject of ongoing future work, we provide an example of how to extend the elementary-link MDP framework of Sec.~\ref{sec-MDP_elem_links} to a system of two elementary links, in which entanglement swapping is included, in Sec.~\ref{sec-MDP_two_elem_links}. We also note that the protocol in Fig.~\ref{fig-QDP_protocol} uses fixed routing and path-finding algorithms from Refs.~[\onlinecite{SMI+17,PKT+19,CRDW19}]. It is possible, in principle, to develop an MDP that takes into account routing. Doing so would allow us to obtain protocols that simultaneously optimize the actions of the elementary links, the joining operations, and the actions corresponding to routing, either directly using dynamic programming algorithms such as the one in Theorem~\ref{thm-opt_policy}, or through reinforcement learning. These possibilities, and other possibilities for developing more sophisticated protocols using MDPs, are interesting directions for future work.
	\end{remark}

\subsection{Fidelity}\label{sec-protocol_fidelity}

	In order to quantify the performance of the protocol described in Fig.~\ref{fig-QDP_protocol}, it is natural to ask what the fidelity of the resulting states of the elementary and virtual links are to prescribed target states. Thus, let us begin by showing, in general terms, how we could calculate the fidelity after $t$ time steps of our protocol.
	
	First, we note that all of the elementary links are independent of each other. This is due to the fact that we assume that every node has a separate quantum system for every one of the elementary links associated to that node. Furthermore, we assume that every elementary link undergoes its own policy independent of the other elementary links. Therefore, after $t$ time steps the quantum state of the network is
	\begin{equation}\label{eq-network_cq_state_QDP}
		\rho_G^{\vec{\pi}}(t)=\bigotimes_{e\in E}\rho_{e}^{\pi_{e}}(t),
	\end{equation}
	where $\vec{\pi}=\{\pi_e:e\in E\}$ is a collection of policies for the individual elementary links, and every state $\rho_{e}^{\pi_{e}}(t)$ is given by \eqref{eq-elem_link_initial_cq_state_2}, namely,
	\begin{equation}\label{eq-elem_link_cq_state}
		\rho_e^{\pi_e}(t)=\sum_{h^t}\Pr[H_e(t)=h^t]_{\pi_e}\,\ketbra{h^t}{h^t}\otimes\sigma_e(t|h^t).
	\end{equation}
	Recall from \eqref{eq-MDP_elem_link_history_prob} that $\Pr[H_e(t)=h^t]_{\pi_e}$ is the probability of the history $h^t$ with respect to the policy $\pi_e$, and $\sigma_e(t|h^t)$ is the quantum state of the elementary link conditioned on the history $h^t$, given by \eqref{eq-elem_link_cond_state}.
	
	The state in \eqref{eq-elem_link_cq_state} is a classical-quantum state that contains both classical information about the history of elementary link as well as the quantum state of the elementary link conditioned on every history. If we condition on an elementary link corresponding to $e\in E$ being active at time $t$, then the expected quantum state of the elementary link at time $t$ is~\cite{Kha21b}
	\begin{equation}
		\overline{\rho}_e^{\pi_e}(t)\coloneqq\frac{1}{X_e^{\pi_e}(t)}\sum_{h^t:m_t\neq -1}\Pr[H_e(t)=h^t]_{\pi_e}\,\sigma_e(t|h^t).
	\end{equation}
	From these states, we can calculate the quantum states of the virtual links in the target graph that are created via joining protocols. In general, the states are of the form \eqref{eq-network_QDP_virtual_link_fid_cond}. As a concrete example, let us consider the usual entanglement swapping protocol from Sec.~\ref{ex-ent_swap}. Let $w=(v_1,e_1,v_2,e_2,\dotsc,e_{n},v_{n+1})$ be a path between two non-neigbouring nodes $v_1$ and $v_{n+1}$, such that the entanglement swapping protocol along this path creates the virtual link given by the edge $\{v_1,v_{n+1}\}$. The quantum state at the input of the entanglement swapping protocol is $\bigotimes_{j=1}^{n}\overline{\rho}_{e_j}^{\pi_{e_j}}(t)$, and the output state, conditioned on success of the protocol is $\mathcal{L}^{\text{ES};n}(\bigotimes_{j=1}^{n}\overline{\rho}_{e_j}^{\pi_{e_j}}(t))$, where we recall the definition of $\mathcal{L}^{\text{ES};n}$ in \eqref{eq-ent_swap_channel}.
	
	After the appropriate joining protocols are performed, and conditioned on their success, we obtain the target graph $G_{\text{target}}=(V,E_{\text{target}})$, and the corresponding quantum state has the form $\bigotimes_{e\in E_{\text{target}}} \omega_e$, where if $e$ is a virtual link, obtained via a joining protocol, then $\omega_e$ is given by \eqref{eq-network_QDP_virtual_link_fid_cond}. Now, the target quantum state is simply a tensor product of the target states corresponding to the edges of the target graph, i.e., $\bigotimes_{e\in E_{\text{target}}}\omega_e^{\text{target}}$. Therefore, by multiplicativity of fidelity with respect to the tensor product, the fidelity of the quantum state after the protocol is equal to $\prod_{e\in E_{\text{target}}}F(\omega_e,\omega_e^{\text{target}})$. For the virtual links, individual fidelities in this product can be calculated using the formulas presented in Sec.~\ref{sec-joining_protocols}.
	

\subsection{Waiting time}

	In addition to the fidelity, another relevant figure of merit is the \textit{expected waiting time}, which is a figure of merit that indicates how long it takes (on average) to establish an elementary or virtual link. This figure of merit has been considered in prior work in the context of both a linear chain of quantum repeaters and general quantum networks~\cite{CJKK07,BPv11,SSv19,VK19,BCE19,KMSD19,CBE21}. 
	
	When defining the waiting times, we imagine a scenario in which elementary link generation is continuously occurring in the network~\cite{CRDW19} and that an end-user request for entanglement occurs at a time $t_{\text{req}}\geq 0$. The waiting time is then the number of time steps from time $t_{\text{req}}$ onward that it takes to establish the entanglement.
	
	\begin{definition}[Elementary link waiting time]\label{def-network_QDP_elem_link_waiting_time}
		Let $G=(V,E)$ be the graph corresponding to the elementary links of a quantum network and let $e\in E$. For all $t_{\text{req}}\geq 0$, the waiting time for the elementary link corresponding to the edge $e$ is defined to be
		\begin{equation}
			W_{e}(t_{\text{req}})\coloneqq\sum_{t=t_{\text{req}}+1}^{\infty} t X_{e}(t)\prod_{i=t_{\text{req}}+1}^{t-1}(1-X_{e}(i)).
		\end{equation}
		Then, the expected waiting time is
		\begin{multline}\label{eq-waiting_time_prob_late_request}
			\mathbb{E}[W_{e}(t_{\text{req}})]_{\pi}=\sum_{t=t_{\text{req}}+1}^{\infty} t\Pr[X_{e}(t_{\text{req}}+1)=0,\\\dotsc,X_{e}(t_{\text{req}}+t)=1]_{\pi},
		\end{multline}
		where $\pi$ is an arbitrary policy for the elementary link corresponding to the edge $e$.
	\end{definition}
	
	We make the following definition for the waiting time for a collection of elementary links.
	
	\begin{definition}[Collective elementary link waiting time]\label{def-collective_elem_link_waiting_time}
		Let $G=(V,E)$ be the graph corresponding to the elementary links of a quantum network, and let $t_{\text{req}}\geq 0$. For every subset $E'\subseteq E$, the waiting time for the elementary links corresponding to the elements of $E'$ is defined to be
		\begin{equation}
			W_{E'}(t_{\text{req}})\coloneqq\sum_{t=t_{\text{req}}+1}^{\infty} tX_{E'}(t)\prod_{i=t_{\text{req}}+1}^{t-1}(1-X_{E'}(i))
		\end{equation}
		where $X_{E'}(t)\coloneqq\prod_{e\in E'}X_e(t)$.
	\end{definition}

	In other words, the collective elementary link waiting time is the time it takes for all of the elementary links given by $E'$ to be simultaneously active, and its expected value is
	\begin{multline}\label{eq-exp_waiting_time_general}
		\mathbb{E}[W_{E'}(t_{\text{req}})]_{\pi}=\sum_{t=t_{\text{req}}+1}^{\infty} t\Pr[X_{E'}(t_{\text{req}}+1)=0,\\\dotsc,X_{E'}(t_{\text{req}}+t)=1]_{\vec{\pi}},
	\end{multline}
	where $\vec{\pi}=\left(\pi_{e}:e\in E'\right)$ is an arbitrary collection of policies for the elementary links corresponding to $E'$. If we consider a collection of elementary links, all undergoing the $t^{\star}=\infty$ memory-cutoff policy, then
	\begin{multline}\label{eq-exp_waiting_time_tInfty}
		\mathbb{E}[W_{E'}(t_{\text{req}})]_{\infty}=\sum_{k=1}^M \binom{M}{k}(-1)^{k+1}\left(1+\frac{(1-p_k)^{t_{\text{req}}+1}}{p_k}\right),\\p_k\coloneqq 1-(1-p)^k.
	\end{multline}
	Proofs of this result using various different techniques can be found in Refs.~[\onlinecite{BPv11,Prax13,KMSD19}]. In Appendix~\ref{sec-exp_waiting_time_tInfty_pf}, we prove this result within the framework introduced here by explicitly evaluating the formula in \eqref{eq-exp_waiting_time_general}.

	\begin{definition}[Virtual link waiting time]\label{def-network_QDP_virtual_waiting_time}
		Let $G=(V,E)$ be the graph corresponding to the elementary links of a quantum network, and let $t_{\text{req}}\geq 0$. Given a pair $v_1,v_n\in V$ of distinct non-adjacent vertices and a path $w=(v_1,e_1,v_2,e_2,\dotsc,e_{n-1},v_n)$ between them for some $n\geq 2$, the \textit{virtual link waiting time} along this path is defined to be the amount of time it takes to establish the virtual link given by the edge $\{v_1,v_n\}$: 
		\begin{equation}
			W_{\{v_1,v_n\};w}(t_{\text{req}})\coloneqq W_{E_w}(t_{\text{req}})\sum_{t=t_{\text{req}}+1}^{\infty}t Y_{w}(1-Y_{w})^{t-1},
		\end{equation}
		where $E_w=\{e_1,e_2,\dotsc,e_{n-1}\}$ is the set of edges corresponding to the path $w$, $W_{E_w}(t_{\text{req}})$ is the collective elementary link waiting time from Definition~\ref{def-collective_elem_link_waiting_time}, and $Y_{E_w}$ is a binary random variable for the success of the joining protocol along the path $w$, so that $Y_w=1$ corresponds to success of the joining protocol and $Y_w=0$ to failure. We define $Y_w$ and $W_{E_w}$ to be independent random variables.
	\end{definition}
	
	The formula for the virtual link waiting time in Definition~\ref{def-network_QDP_virtual_waiting_time} is based on the formula in Ref.~[\onlinecite{CJKK07}]. It corresponds to the simple strategy of waiting for all of the elementary links along the path $w$ to be established and then performing the measurements for the joining protocol. Note that this strategy is consistent with our overall quantum network protocol in Fig.~\ref{fig-QDP_protocol}.

\subsection{Key rates for quantum key distribution}
	
	In order to determine secret key rates between arbitrary pairs of nodes in a quantum network, we need to keep track of the quantum state of the relevant elementary links as a function of time. The following discussion and formulas for secret key rates are based on Ref.~[\onlinecite{GKF+15}].
	
	Suppose that $K$ is a function that gives the number of secret key bits per entangled state shared by the nodes of either an elementary link or virtual link. ($K$ is, for example, the formula for the asymptotic secret key rate of the BB84, six-state, or device-independent protocol.) Then, suppose that $G=(V,E)$ is the graph corresponding to the elementary links of a quantum network.
	Consider a collection $e'\coloneqq\{v_1,\dotsc,v_k\}\notin E$ of distinct nodes corresponding to a virtual link for some $k\geq 2$, and let $w$ be a path in the physical graph leading to the virtual link given by $e'$. An entanglment swapping protocol is performed along the path $w$ in order to establish the bipartite virtual link. Conditioned on success of the joining protocol, the quantum state of the virtual link is given by \eqref{eq-network_QDP_virtual_link_fid_cond}, namely,
	\begin{equation}\label{eq-network_QDP_virtual_link_output_cond}
		\frac{1}{p_{\text{succ}}}\mathcal{L}_{w\to e'}^1(\rho_w),
	\end{equation}
	where
	\begin{equation}
		p_{\text{succ}}=\Tr\!\left[\mathcal{L}_{w\to e'}^1(\rho_w)\right]
	\end{equation}
	is the success probability of the joining protocol. Then, the secret key rate (in units of secret key bits per second) for the virtual link along the path $w$ is
	\begin{equation}
		\widetilde{K}_{e';w}=p_{\text{succ}}\nu_{e'}^{\text{rep}}K.
	\end{equation}
	Here, $K$ is calculated using the state in \eqref{eq-network_QDP_virtual_link_output_cond}. The repetition rate $\nu_{e'}^{\text{rep}}$ in this case is a function of the end-to-end classical communication time required for executing the joining protocol.

\section{A Markov decision process beyond the elementary link level}\label{sec-MDP_two_elem_links}

	The developments so far in this work constitute an analysis of quantum networks using a Markov decision process (MDP) for elementary links. As we have seen, the framework of MDPs is useful because it allows us to model noise processes and imperfections that are present in near-term quantum technologies, and thus allows us to understand the limits on the performance of near-term quantum networks. An important question is how useful the MDP formalism will be in practice when scaling up to model systems of more than one elementary link. In this section, we provide an MDP for a system of two elementary links, taking entanglement swapping into account. We note that in recent work~\cite{SvL21} an MDPs for repeater chains with two, three and four elementary links have been considered, but the definition of the MDP here differs from from the one in Ref.~[\onlinecite{SvL21}], because here we take decoherence of the quantum memories into account.
	
	We start this section by defining the basic elements of the MDP, and then we show how to obtain optimal policies using linear programming. In particular, we formulate the optimal expected waiting time to obtain the end-to-end virual link and the optimal expected fidelity of the end-to-end virtual link as linear programs. Then, we show that prior analytical results on the expected waiting time for two elementary links under the memory-cutoff policy~\cite{CJKK07}, known only in the ``symmetric'' scenario when the two elementary links have the same transmission-heralding success probability and the same memory cutoff, can be reproduced. However, we note that our linear programming procedure can be applied even in non-symmetric scenarios.
	
\subsection{An MDP for two elementary links}\label{sec-MDP_two_elem_links_def}
	
	Let $p_1$ and $p_2$ be the success probabilities for generating the two elementary links, and let $q$ be the probability of successful entanglement swapping. Note that $p_1$ and $p_2$ are defined exactly as in Sec.~\ref{sec-practical_elem_link_generation}. In particular,
	\begin{align}
		p_1&=\Tr[(\mathcal{M}_1^1\circ\mathcal{S}_1)(\rho_1^S)],\\
		p_2&=\Tr[(\mathcal{M}_2^1\circ\mathcal{S}_2)(\rho_2^S)],
	\end{align}
	where $\mathcal{M}_j^1$, $j\in\{1,2\}$, are the completely positive maps corresponding to success of the heralding proecedure for the $j^{\text{th}}$ elementary link, $\mathcal{S}_j$ is the transmission channel from the source to the nodes for the $j^{\text{th}}$ elementary link, and $\rho_j^S$ is the state produced by the source associated with the $j^{\text{th}}$ elementary link; see Fig.~\ref{fig-two_elem_link_MDP}. We also define the states
	\begin{align}
		\sigma_j^0&=\frac{1}{p_j}(\mathcal{M}_j^1\circ\mathcal{S}_j)(\rho_j^S),\\
		\sigma_j(m)&=\mathcal{N}_j^{\circ m}(\sigma_j^0),\quad j\in\{1,2\},
	\end{align}
	where $\mathcal{N}_j$ is the quantum channel describing the decoherence of the quantum memories associated with the $j^{\text{th}}$ elementary link.
	
	\begin{figure}
		\centering
		\includegraphics[scale=1]{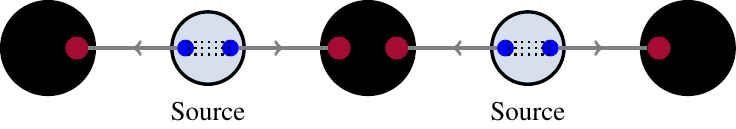}
		\caption{Two elementary links with entanglement swapping at the central node.}\label{fig-two_elem_link_MDP}
	\end{figure}
	
	Now, recall that in the case of one elementary link considered in Sec.~\ref{sec-elem_link_MDP_def}, the state variable was just the memory time $M(t)$, referring to the time for which the quantum state of the elementary link was held in the memories of the nodes, and the actions consisted of either keeping the elementary link or discarding it and generating a new one. Now, in the case of two elementary links, we must keep track of the memory time of both elementary links, and we also store information about whether or not the virtual (end-to-end) link is active. The actions are similar to before, consisting of the same elementary link actions as before, but now we define an additional action for performing the entanglement swapping operation. Formally, we have the following.
	\begin{itemize}
		\item \textit{States}: The states of the MDP are elements of the set $\SF{S}=\SF{X}\times\SF{M}_1\times\SF{M}_2$, where $\SF{X}=\{0,1\}$ indicates whether or not the end-to-end link is active, $\SF{M}_1=\{-1,0,1,\dotsc,m_1^{\star}\}$ is the set of possible states of the first elementary link (with the elements of the set having the same interpretation as in the elementary link MDP), and $\SF{M}_2=\{-1,0,1,\dotsc,m_2^{\star}\}$ is the set of possible states of the second elementary link. In particular, $m_1^{\star}$ and $m_2^{\star}$ are the maximum storage times of the two elementary links, corresponding to their coherence times; see Sec.~\ref{sec-elem_link_MDP_def}. To these states, we associate the (standard) probability simplex spanned by the orthonormal vectors $\ket{x}\otimes\ket{m_1}\otimes\ket{m_2}$, with $x\in\SF{X}$, $m_1\in\SF{M}_1$, and $m_2\in\SF{M}_2$, and we often use the abbreviation $\ket{s}\equiv\ket{x,m_1,m_2}\equiv\ket{x}\otimes\ket{m_1}\otimes\ket{m_2}$ for every $s=(x,m_1,m_2)\in\SF{S}$.
		
			We use $S(t)=(X(t),M_1(t),M_2(t))$, $t\in\mathbb{N}$, to refer to the random variables (taking values in $\SF{S}$) corresponding to the state of the MDP. 
		
		\item \textit{Actions}: The set of actions is $\SF{A}=\{00,01,10,11,\Join\}$, where the different actions have the following meanings:
				\begin{itemize}
					\item $00$: Keep both elementary links.
					\item $01$: Keep the first elementary link, discard and regenerate the second.
					\item $10$: Discard and regenerate the first elementary link, keep the second.
					\item $11$: Discard and regenerate both elementary links.
					\item $\Join$: Perform entanglement swapping.
				\end{itemize}
				
				We use $A(t)$, $t\in\mathbb{N}$, to refer to the random variables (taking values in the set $\SF{A}$) corresponding to the actions taken.
				
				We let $H(t)=(S(1),A(1),S(2),A(2),\dotsc,A(t-1),S(t))$ be the history, consisting of a sequence of states and actions, up to time $t\in\mathbb{N}$, with $H(1)=S(1)$.
			
		\item \textit{Figure of merit}: For the elementary link MDP defined in Sec.~\ref{sec-elem_link_MDP_def}, recall that the figure of merit was essentially the fidelity of the elementary link, but scaled by a factor corresponding to the probability that the elementary link is active. We define the figure of merit here in an analogous fashion as follows:
			\begin{widetext}
			\begin{equation}\label{eq-two_elem_link_fidelity}
				f(x,m_1,m_2)=\left\{\begin{array}{l l} \bra{\psi}\mathcal{L}^{\text{ES};1}(\sigma_1(m_1)\otimes\sigma_2(m_2))\ket{\psi} & \text{if } x=1,\,m_1,m_2\geq 0, \\ 0 & \text{otherwise}, \end{array}\right.
			\end{equation}
			\end{widetext}
			where we recall that $\mathcal{L}^{\text{ES};1}$ is the entanglement swapping channel for one intermediate node, as defined in Sec.~\ref{ex-ent_swap}, and $\ket{\psi}$ is a target pure state vector, which in this context is typically the maximally entangled state vector $\ket{\Phi}$, as defined in \eqref{eq-max_ent_state}.
	\end{itemize}
	
	Let us now proceed to the definition of the transition matrices for our MDP. Unlike the elementary link scenario, in this scenario of two elementary links we want not only for the fidelity and success probability of the end-to-end link to be high, but we also want the average amount of time it takes to generate the end-to-end link to be low---in other words, we want the expected waiting time to be low as well. Therefore, in order to address the expected waiting time in our MDP, we define the transition matrices in such a way that states corresponding to an active end-to-end link (i.e., states $s=(x,m_1,m_2)\in\SF{S}$ such that $x=1$) are \textit{absorbing states}. By doing this, the expected waiting time is nothing but the expected time to absorption, which is a standard result in the theory of Markov chains; see, e.g., Ref.~[\onlinecite{Stewart09_book}]. We note that this idea of relating the expected waiting time of a quantum repeater chain to the absorption time of a Markov chain has already been used in Ref.~[\onlinecite{SSv19}]; however, here, we apply this idea in the more general context of an MDP, while also taking memory decoherence and other device imperfections explicitly into account.
	
	Let $T_j^a$ denote the transition matrix for the $j^{\text{th}}$ elementary link, as defined in \eqref{eq-elem_link_transition_0} and \eqref{eq-elem_link_transition_1}, for $a\in\{0,1\}$. Then, using those elementary link transition matrices, we define the transition matrices for our MDP for two elementary links as follows:
	\begin{align}
		T^{00}&\coloneqq\ketbra{0}{0}\otimes T_1^0\otimes T_2^0+\ketbra{1}{1}\otimes\mathbbm{1}_1\otimes\mathbbm{1}_2,\\
		T^{01}&\coloneqq\ketbra{0}{0}\otimes T_1^0\otimes T_2^1+\ketbra{1}{1}\otimes\mathbbm{1}_1\otimes\mathbbm{1}_2,\\
		T^{10}&\coloneqq\ketbra{0}{0}\otimes T_1^1\otimes T_2^0+\ketbra{1}{1}\otimes\mathbbm{1}_1\otimes\mathbbm{1}_2,\\
		T^{11}&\coloneqq\ketbra{0}{0}\otimes T_1^1\otimes T_2^1+\ketbra{1}{1}\otimes\mathbbm{1}_1\otimes\mathbbm{1}_2,\\
		T^{\Join}&\coloneqq\ketbra{0}{0}\otimes\left((1-q)\ketbra{g_{p_1},g_{p_2}}{\gamma_1^+,\gamma_2^+}\right.\nonumber\\ &\qquad\left.+S_1\otimes\ketbra{-1}{-1}+\ketbra{-1}{-1}\otimes S_2\right.\nonumber\\ &\qquad +\ketbra{-1,-1}{-1,-1}+\ketbra{-1,-1}{-1,m_2^{\star}}\nonumber\\ &\qquad \left.+\ketbra{-1,-1}{m_1^{\star},-1}\right)\nonumber\\ &\qquad +\ketbra{1}{0}\otimes q\mathbbm{1}_1^+\otimes\mathbbm{1}_2^++\ketbra{1}{1}\otimes\mathbbm{1}_1\otimes\mathbbm{1}_2,
	\end{align}
	where
	\begin{align}
		\ket{\gamma_j^+}&=\sum_{m=0}^{m_j^{\star}}\ket{m},\\
		\mathbbm{1}_j&=\sum_{m=-1}^{m_j^{\star}}\ketbra{m}{m},\\
		\mathbbm{1}_j^+&=\sum_{m=0}^{m_j^{\star}}\ketbra{m}{m},\\
		S_j&=\sum_{m=0}^{m_j^{\star}-1}\ketbra{m+1}{m},
	\end{align}
	and $\ket{g_{p_j}}$, $j\in\{1,2\}$, is defined exactly as in \eqref{eq-elem_link_MDP_gen_vector}.
	
	First, let us observe that every transition matrix has a block structure, with the blocks defined by the transitions of the status of the end-to-end link. Specifically, we can write every transition matrix $T^a$ as
	\begin{equation}
		T^a=\begin{pmatrix} T_{0\to 0}^{a} & T_{1\to 0}^a \\ T_{0\to 1}^a & T_{1\to 1}^a \end{pmatrix},\quad a\in\SF{A},
	\end{equation}
	where the sub-blocks $T_{x\to x'}^a$ is the block corresponding to the transition of the status of the virtual link from $x\in\{0,1\}$ to $x'\in\{0,1\}$. (We note, as before, that probability vectors are applied to transition matrices from the right; see Appendix~\ref{sec-MDPs_overview}.) From this, we see that for the actions $00,01,10,11$, the transition matrices are of the following block-diagonal form:
	\begin{equation}
		T^{jk}=\begin{pmatrix} T_1^j\otimes T_2^k & 0 \\ 0 & \mathbbm{1}_1\otimes\mathbbm{1}_2 \end{pmatrix},\quad j,k\in\{0,1\}.
	\end{equation}
	Therefore, for these transition matrices, because the entanglement swapping action is not performed, the transition from $x=0$ to $x=1$ is not possible. Consequently, if the end-to-end is initially inactive ($x=0$), then it stays inactive and each elementary link transitions independently according to the elementary link transition matrices from Sec.~\ref{sec-elem_link_MDP_def}. If the end-to-end link is intially active ($x=1$), then nothing happens to the states of the elementary links, in accordance with the definition of an absorbing state. For the action $\Join$ of entanglement swapping, we have three non-zero blocks. The block $T_{0\to 0}^{\Join}$ means that the end-to-end link is initially inactive and stays inactive, which can happen in one of several ways:
	\begin{itemize}
		\item Both elementary links are initially active but the entanglement swapping fails, after which both elementary links are regenerated. This possibility is given by the term $(1-q)\ketbra{g_{p_1},g_{p_2}}{\gamma_1^+,\gamma_2^+}$.
		
		\item Both elementary links are initially inactive. In this case, they both remain inactive after the entanglement swapping action, and this is given by the term $\ketbra{-1,-1}{-1,-1}$.
		
		\item One of the elementary links is active but the other is not. In this case, the memory time of the active elementary link is incremented by one, corresponding to the ``shift'' operator $S_j$ on the active elementary link, while the inactive elementary link remains inactive. These possibilities are given by the terms $S_1\otimes\ketbra{-1}{-1}$ and $\ketbra{-1}{-1}\otimes S_2$.
		
		\item One of the elementary links is inactive and the other has reached is maximum storage time. In this case, the inactive elementary link remains inactive, and the other elementary link transitions to the $-1$ state, because the maximum time $m_j^{\star}$ was reached. These possibilities are given by the terms $\ketbra{-1,-1}{m_1^{\star},-1}$ and $\ketbra{-1,-1}{-1,m_2^{\star}}$.
		
	\end{itemize}
	The block $T_{0\to 1}^{\Join}$ corresponds to a transition from the end-to-end link initially being inactive to being active, which happens when the entanglement swapping succeeds. Since the entanglement swapping is possible only when both elementary links are active, and because we want to keep track of the memory times of the elementary links at the moment the entanglement swapping is performed, this block is given by $q\mathbbm{1}_1^+\otimes\mathbbm{1}_2^+$. Finally, the block $T_{1\to 1}^{\Join}$ corresponds to the end-to-end link being active already; thus, in accordance with the definition of an absorbing state, this block is given simply by $\mathbbm{1}_1\otimes\mathbbm{1}_2$, as with the other actions.

	Now, just as we defined a memory-cutoff policy for elementary links in Sec.~\ref{sec-elem_link_mem_cutoff_policy}, we can define a memory-cutoff policy for the system of two elementary links that we are considering here. Suppose that the first elementary link has cutoff time $t_1^{\star}\leq m_1^{\star}$ and the second elementary link has cutoff time $t_2^{\star}\leq m_2^{\star}$. Then, we define the decision function such that, if both elementary links are active, then an entanglement swap is attempted; otherwise, one of the actions $01$, $10$, or $11$ is performed, depending on which elementary links are active. This leads to the following definition of the deterministic decision function.
	\begin{multline}\label{eq-two_elem_link_mem_cutoff_policy}
		d(0,m_1,m_2)\\=\left\{\begin{array}{l l} 01, & m_1\in\{0,\dotsc,t_1^{\star}-1\},\,m_2=-1,\\[0.1cm]
		10, & m_1=-1,\,m_2\in\{0,\dotsc,t_2^{\star}-1\},\\[0.1cm]
		11, & (m_1,m_2)=(-1,-1),\,(-1,t_2^{\star}),\,(t_1^{\star},-1),\\[0.1cm]
		\Join, & m_1\in\{0,\dotsc,t_1^{\star}\},\,m_2\in\{0,\dotsc,t_2^{\star}\}, \end{array}\right.
	\end{multline}
	for all $m_1\in\SF{M}_1$ and $m_2\in\SF{M}_2$. Note that it is only necessary to define the decision function on the transient states $(0,m_1,m_2)$ and not the absorbing states $(1,m_1,m_2)$, because the figures of merit that we are concerned with (such as the expected value of the function $f$ in \eqref{eq-two_elem_link_fidelity} and the expected waiting time to absorption) do not depend on the values of the decision function on absorbing states.

\subsection{Optimal policies via linear programming}

	Having defined the basic elements of the MDP for two elementary links with entanglement swapping, let us now look at optimal policies. We are concerned both with the figure of merit defined in \eqref{eq-two_elem_link_fidelity} and with the expected waiting time to obtain an end-to-end link. In Appendix~\ref{sec-MDP_linear_prog}, we show that both quantities can be bounded using linear programs. In fact, the results in Appendix~\ref{sec-MDP_linear_prog} go beyond the MDP for two elementary links that we consider here, because the linear programs apply to general MDPs with arbitrary state and action sets and transition matrices.
	
	\begin{theorem}[Linear program for the optimal expected value for two elementary links]\label{thm-two_elem_links_opt_fidelity_lin_prog}
		Given a system of two elementary links, along with the associated MDP defined in Sec.~\ref{sec-MDP_two_elem_links_def}, the optimal expected value of the function $f$ defined in \eqref{eq-two_elem_link_fidelity} is given by the following linear program:
		\begin{equation}\label{eq-two_elem_links_opt_fidelity_lin_prog}
			\begin{array}{l l} \textnormal{maximize} & \braket{f}{1,x} \\[0.2cm] \textnormal{subject to} & 0\leq\ket{w_a}\leq\ket{x}\quad\forall~a\in\textnormal{\textsf{A}},\\[0.2cm] & 0\leq\ket{v_a}\leq\ket{y}\quad\forall~a\in\textnormal{\textsf{A}},\\[0.2cm] & \displaystyle\sum_{a\in\textnormal{\textsf{A}}}\sum_{i=0}^1 T_{0\to i}^a\ket{w_a}=\ket{0}\ket{x},\\[0.5cm] & \displaystyle\sum_{a\in\textnormal{\textsf{A}}}\ket{w_a}=\ket{x},\\[0.5cm] & \displaystyle\ket{y}-\sum_{a\in\textnormal{\textsf{A}}}T_{0\to 0}^a\ket{v_a}=\ket{g_{p_1}}\ket{g_{p_2}},\\[0.5cm] & \displaystyle \sum_{a\in\textnormal{\textsf{A}}}\ket{v_a}=\ket{y},\\[0.5cm] & \displaystyle T_{0\to 1}^{\Join}\ket{v_{\Join}}=\ket{x}, \end{array}
		\end{equation}
		where the optimization is with respect to the $(m_1^{\star}+2)\cdot(m_2^{\star}+2)$-dimensional vectors $\ket{x}$, $\ket{y}$, $\ket{w_a}$, $\ket{v_a}$, $a\in\SF{A}$, and the inequality constraints are component-wise. Every set of feasible points $\ket{x}$, $\ket{y}$, $\ket{w_a}$, $\ket{v_a}$, $a\in\SF{A}$, of this linear program defines a stationary policy with decision function $d$, whose values for the transient states $(0,m_1,m_2)$ are as follows:
		\begin{equation}
			d(0,m_1,m_2)(a)=\frac{\braket{0,m_1,m_2}{w_a}}{\braket{0,m_1,m_2}{x}},
		\end{equation}
		for all $m_1\in\SF{M}_1$, $m_2\in\SF{M}_2$, and $a\in\SF{A}$. If $\braket{0,m_1,m_2}{x}=0$, then we can set $d(0,m_1,m_2)$ to be an arbitrary probability distribution over the set $\SF{A}$ of actions.
	\end{theorem}
	
	\begin{remark}\label{rem-decision_func_transient}
		Note that in the theorem statement above we defined the action of the decision function only for the transient states. For the absorbing states, we can set the decision function to be arbitrary, because neither the expected value of the MDP nor the expected waiting time to absorption is affected by the value of the decision function on absorbing states; see Appendix~\ref{sec-MDPs_figs_of_merit}.
	\end{remark}
	
	\begin{theorem}[Linear program for the optimal expected waiting time for two elementary links]\label{thm-two_elem_links_opt_waiting_time_lin_prog}
		Given a system of two elementary links, along with the associated MDP defined in Sec.~\ref{sec-MDP_two_elem_links_def}, the optimal expected waiting time is given by the following linear program:
		\begin{equation}\label{eq-two_elem_links_opt_waiting_time_lin_prog}
			\begin{array}{l l} \textnormal{minimize} & \braket{\gamma}{x} \\[0.2cm] \textnormal{subject to} & 0\leq\ket{w_a}\leq\ket{x}\quad\forall~a\in\textnormal{\textsf{A}}, \\[0.3cm] & \displaystyle\ket{x}-\sum_{a\in\textnormal{\textsf{A}}}T_{0\to 0}^a\ket{w_a}=\ket{g_{p_1}}\ket{g_{p_2}}, \\[0.5cm] & \displaystyle\sum_{a\in\textnormal{\textsf{A}}}\ket{w_a}=\ket{x}, \end{array}
		\end{equation}
		where the optimization is with respect to the $(m_1^{\star}+2)\cdot(m_2^{\star}+2)$-dimensional vectors $\ket{x}$ and $\ket{w_a}$, $a\in\SF{A}$, and the inequality constraints are component-wise. Every set of feasible points $\ket{x}$, $\ket{w_a}$, $a\in\SF{A}$, of this linear program defines a stationary policy with decision $d$, whose values for the transient states $(0,m_1,m_2)$ (see Remark~\ref{rem-decision_func_transient} above) are as follows:
		\begin{equation}
			d(0,m_1,m_2)(a)=\frac{\braket{0,m_1,m_2}{w_a}}{\braket{0,m_1,m_2}{x}},
		\end{equation}
		for all $m_1\in\SF{M}_1$, $m_2\in\SF{M}_2$, and $a\in\SF{A}$. If $\braket{0,m_1,m_2}{x}=0$, then we can set $d(0,m_1,m_2)$ to be an arbitrary probability distribution over the set $\SF{A}$ of actions.
	\end{theorem}
	
	\begin{figure}
		\centering
		\includegraphics[scale=1]{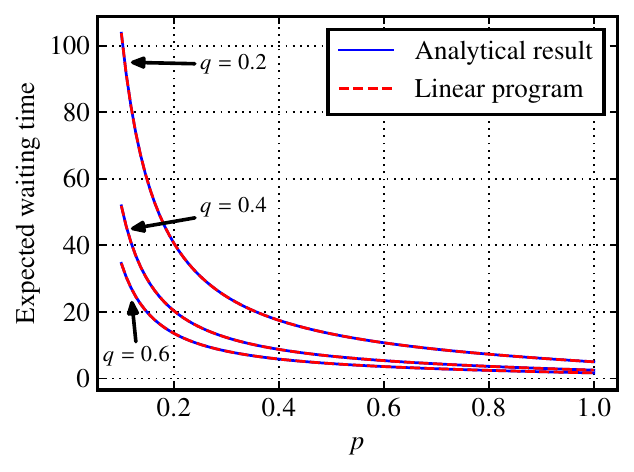}
		\caption{The expected waiting time for an end-to-end for a system of two elementary links, as depicted in Fig.~\ref{fig-two_elem_link_MDP}. We let $p_1=p_2=p$ be the transmission-heralding success probability for both elementary links, and we denote by $q$ the success probability for entanglement swapping. We compare the known analytical result for this scenario [\onlinecite[Eq.~(5)]{CJKK07}], with cutoff $t^{\star}=5$ (see \eqref{eq-two_elem_links_waiting_time_prior}), to the solution obtained by the linear program in \eqref{eq-two_elem_links_opt_waiting_time_lin_prog}, with maximum storage time $m^{\star}=5$.}\label{fig-LP_vs_analytic}
	\end{figure}

	We now show that the linear program in \eqref{eq-two_elem_links_opt_waiting_time_lin_prog} reproduces the known analytical result in Ref.~[\onlinecite[Eq.~(5)]{CJKK07}] for the expected waiting time for two elementary links with the same success probability $p$ and cutoff time $t^{\star}$:
	\begin{equation}\label{eq-two_elem_links_waiting_time_prior}
		\frac{3-2p(1-(1-p)^{t^{\star}})-2(1-p)^{t^{\star}}}{qp(2-p(1-2(1-p)^{t^{\star}})-2(1-p)^{t^{\star}})}.
	\end{equation}
	In Fig.~\ref{fig-LP_vs_analytic}, we plot this function along with the optimal value obtained for the linear program in \eqref{eq-two_elem_links_opt_waiting_time_lin_prog}. We find that the two curves coincide for all values of the transmission-heralding probability $p$ and the entanglement swapping success probability $q$ considered. This provides us not only with a sanity check on the linear program, but it also provides evidence that the memory-cutoff policy in \eqref{eq-two_elem_link_mem_cutoff_policy} is optimal, at least in the ``symmetric'' scenario, in which both elementary links have the same transmission-heralding success probability. We also note that the result in \eqref{eq-two_elem_links_waiting_time_prior} holds only in this symmetric scenario, while the linear program in \eqref{eq-two_elem_links_opt_waiting_time_lin_prog} can be used to determine the optimal expected waiting time in arbitrary parameter regimes.

\section{Summary and outlook}\label{sec-summary_outlook}

	The central topic of this work is the theory of near-term quantum networks---specifically, how to describe them and how to develop protocols for entanglement distribution in practical scenarios with near-term quantum technologies. The goal in this area of research is to develop protocols that can handle multiple-user requests, work for any given network topology, and can adapt to changes in topology and attacks to the network infrastructure, with the ultimate goal being the realization of the quantum internet. In this work, we have laid some of the foundations for this research program. The core idea is that Markov decision processes (MDPs) provide a natural setting in which to analyze near-term quantum network protocols. We illustrated this idea in this work by first analyzing the MDP for elementary links first introduced in Ref.~[\onlinecite{Kha21b}], simplifying its formulation and presenting some new results about it. Notably, in Theorem~\ref{thm-opt_pol_mem_cutoff}, we show that the memory-cutoff policy is optimal in the steady-state limit. We then showed how the elementary link MDP can be used as part of an overall quantum network protocol. Finally, we provided a first step towards using the MDP formalism for more realistic, larger networks, by providing an MDP for two elementary links. We showed that important figures of merit such as the fidelity of the end-to-end link as well as the expected waiting time for the end-to-end link, can be obtained using linear programs.
	
	Moving forward, there are many interesting directions to pursue. The MDPs introduced in this work are not entirely general, because they do not model protocols for arbitrary repeater chains nor arbitrary networks. Thus, to start with, extending the MDP for two elementary links to repeater chains of arbitrary length is an interesting direction for future work. In this direction, we expect that linear, and possibly even semi-definite relaxations of the expected value of the end-to-end link and of the expected waiting time, such as those in Theorem~\ref{thm-two_elem_links_opt_fidelity_lin_prog} and Theorem~\ref{thm-two_elem_links_opt_waiting_time_lin_prog}, are going to be crucial in the analysis of longer repeater chains, because the size of the MDP (the number of states and actions) will grow exponentially with the number of elementary links.
	
	Going beyond repeater chains to general quantum networks, it is of interest to examine protocols involving multiple cooperating agents. When we say that agents ``cooperate'', we mean that they are allowed to communicate with each other. In the context of quantum networks, agents who cooperate have knowledge beyond that of their own nodes. If every agent cooperates with an agent corresponding to a neighbouring elementary link, then the agents would have knowledge of the network in their local vicinity, and this would in principle improve waiting times and rates for entanglement distribution. Furthermore, the quantum state of the network would not be a simple tensor product of the quantum states corresponding to the individual edges, as we have in \eqref{eq-network_cq_state_QDP} when all the agents are independent. See Refs.~[\onlinecite{PKT+19,CRDW19}] for a discussion of nodes with local and global knowledge of a quantum network in the context of routing.
	
	Finally, another interesting direction for future work is to develop quantum network protocols based on decision processes that incorporate queuing models for requests for links of a specific type between specific nodes; see, e.g., Refs.~[\onlinecite{PGGT20,DT21}]. Then, one can calculate quantities such as the time needed to fulfill all requests. We can also calculate the ``capacity'' of the network, defined in the context of queuing systems as the maximum number of requests that can be fulfilled per unit time.
	
\bigskip

\begin{acknowledgments}
	
	Much of this work is based on the author's PhD thesis research~\cite{Kha21}, which was conducted at the Hearne Institute for Theoretical Physics, Department of Physics and Astronomy, Louisiana State University. During this time, financial support was provided by the National Science Foundation and the National Science and Engineering Research Council of Canada Postgraduate Scholarship. The author also acknowledges support from the BMBF (QR.X). The plots in this work were made using the Python package matplotlib~\cite{matplotlib}.

\end{acknowledgments}

	




\onecolumngrid

\appendix

\newpage

\toclesslab\section{Overview of Markov decision processes}{sec-MDPs_overview}

	In this section, we provide a brief overview of the concepts from the theory of Markov decision processes (MDPs) that are relevant for this work. We mostly follow the definitions and results as presented in Ref.~[\onlinecite{Put14_book}] while using the notation defined in Sec.~\ref{sec-notation}.
	
	\bigskip

\toclesslab\subsection{Notation}{sec-notation}
	
	Throughout this work, we deal with probability distributions defined on a discrete, finite set of points. It is very helpful to write these probability distributions as vectors in a (standard) probability simplex. We do this as follows. Consider a finite set $\SF{X}$. To this set, we associate the orthonormal vectors $\{\ket{x}\}_{x\in\SF{X}}$ in $\mathbb{R}^{|\SF{X}|}$, which means that $\braket{x}{x'}=\delta_{x,x'}$ for all $x,x'\in\SF{X}$. The probability simplex corresponding to $\SF{X}$ is then formally defined as all convex combinations of the vectors in $\{\ket{x}\}_{x\in\SF{X}}$:
	\begin{equation}
		\Delta_{\SF{X}}\coloneqq\left\{\sum_{x\in\SF{X}}p_x\ket{x}:0\leq p_x\leq 1,\,\sum_{x\in\SF{X}}p_x=1\right\}.
	\end{equation}
	This set is in one-to-one correspondence with the set of all probability distributions defined on $\SF{X}$. Specifically, let $P:\SF{X}\to[0,1]$ be a probability distribution (probability mass function) on $\SF{X}$, i.e., $P(x)\in[0,1]$ for all $x\in\SF{X}$ and $\sum_{x\in\SF{X}}P(x)=1$. The unique probability vector $\ket{P}_{\SF{X}}\in\Delta_{\SF{X}}$ corresponding to $P$ is
	\begin{equation}
		\ket{P}_{\SF{X}}\coloneqq\sum_{x\in\SF{X}}P(x)\ket{x}.
	\end{equation}
	We drop the subscript $\SF{X}$ from $\ket{P}_{\SF{X}}$ whenever the underlying set $\SF{X}$ is clear from context. It is important to note and to emphasize that the vector $\ket{P}$ does not represent a quantum state---the bra-ket notation is used merely for convenience. Normalization of the probability vector is then captured by defining the following vector:
	\begin{equation}
		\ket{\gamma_{\SF{X}}}\coloneqq\sum_{x\in\SF{X}}\ket{x}.
	\end{equation}
	We often omit the subscript $\SF{X}$ in $\ket{\gamma_{\SF{X}}}$ when the underlying set $\SF{X}$ is clear from context. Then
	\begin{equation}
		\braket{\gamma}{P}=\sum_{x\in\SF{X}}P(x)=1.
	\end{equation}
	It is often the case that a probability distribution is associated with a random variable $X$ taking values in $\SF{X}$, so that $P(x)\equiv P_X(x)=\Pr[X=x]$ for all $x\in\SF{X}$. In this case, for brevity, we sometimes write the probability vector as
	\begin{equation}
		\ket{X}\equiv\ket{P_X}=\sum_{x\in\SF{X}}\Pr[X=x]\ket{x}.
	\end{equation}
	
	Now, consider another random variable $Y$ taking values in the finite set $\SF{Y}$. We regard stochastic matrices mapping $X$ to $Y$ (i.e., matrices of conditional probabilities $\Pr[Y=y|X=x]$) as linear operators with domain $\Delta_{\SF{X}}$ and codomain $\Delta_{\SF{Y}}$: 
	\begin{equation}
		T_{Y|X}\coloneqq\sum_{\substack{x\in\SF{X}\\y\in\SF{Y}}}\Pr[Y=y|X=x]\ketbra{y}{x},
	\end{equation}
	and we denote the matrix elements by
	\begin{equation}
		T_{Y|X}(y;x)\coloneqq\bra{y}T_{Y|X}\ket{x}=\Pr[Y=y|X=x]\quad\forall~x\in\SF{X},\,y\in\SF{Y}.
	\end{equation}
	We then have, by definition of a stochastic matrix,
	\begin{equation}
		\bra{\gamma_{\SF{Y}}}T_{Y|X}=\bra{\gamma_{\SF{X}}},
	\end{equation}
	which captures the fact that the columns of a stochastic matrix sum to one. Then, if $\ket{P_X}\in\Delta_{\SF{X}}$ is a probability distribution corresponding to $X$, then the action of the matrix $T_{Y|X}$ on $\ket{P_X}$, which results in the probability distribution $\ket{P_Y}\in\Delta_{\SF{Y}}$ corresponding to $Y$, can be written as
	\begin{equation}
		\ket{P_Y}\coloneqq T_{Y|X}\ket{P_X}.
	\end{equation}
	In particular, for all $y\in\SF{Y}$,
	\begin{align}
		P_Y(y)&\coloneqq\braket{y}{P_Y}\\
		&=\bra{y}T_{Y|X}\ket{P_X}\\
		&=\sum_{x\in\SF{X}}\Pr[Y=y|X=x]P_X(x).
	\end{align}
	
	Finally, we discuss joint probability distributions. Consider two finite sets $\SF{X}$ and $\SF{Y}$ and the set $\Delta_{\SF{X}\times\SF{Y}}\subset\mathbb{R}^{|\SF{X}\times\SF{Y}|}$ of all (joint) probability distributions on $\SF{X}\times\SF{Y}$. Now, because $\mathbb{R}^{|\SF{X}\times\SF{Y}|}\cong\mathbb{R}^{|\SF{X}|}\otimes\mathbb{R}^{|\SF{Y}|}$, we can regard $\Delta_{\SF{X}\times\SF{Y}}$ as the convex span (convex hull) of tensor product orthonormal vectors $\ket{x}\otimes\ket{y}$, $x\in\SF{X}$, $y\in\SF{Y}$. Thus, every $\ket{Q}_{\SF{XY}}\in\Delta_{\SF{X}\times\SF{Y}}$ can be written as
	\begin{equation}
		\ket{Q}_{\SF{XY}}=\sum_{(x,y)\in\SF{X}\times\SF{Y}}Q_{x,y}\ket{x}\otimes\ket{y}.
	\end{equation}
	We frequently use the abbreviation $\ket{x,y}\equiv\ket{x}\otimes\ket{y}$ in this paper. Then, marginal distributions can be obtained as follows:
	\begin{align}
		\ket{Q}_{\SF{X}}&\coloneqq(\mathbbm{1}_{\SF{X}}\otimes\bra{\gamma_{\SF{Y}}})\ket{Q}_{\SF{XY}}=\sum_{x\in\SF{X}}\left(\sum_{y\in\SF{Y}}Q_{x,y}\right)\ket{x},\\
		\ket{Q}_{\SF{Y}}&\coloneqq(\bra{\gamma_{\SF{X}}}\otimes\mathbbm{1}_{\SF{Y}})\ket{Q}_{\SF{XY}}=\sum_{y\in\SF{Y}}\left(\sum_{x\in\SF{X}}Q_{x,y}\right)\ket{y},
	\end{align}
	where
	\begin{equation}
		\mathbbm{1}_{\SF{X}}\coloneqq\sum_{x\in\SF{X}}\ketbra{x}{x},\quad\mathbbm{1}_{\SF{Y}}\coloneqq\sum_{y\in\SF{Y}}\ketbra{y}{y}.
	\end{equation}
	These concepts for probability distributions defined on two sets can be readily extended to probability distributions defined on sets of the form $\SF{X}_1\times\SF{X}_2\times\dotsb\times\SF{X}_n$ for all $n\geq 2$. 
	
	\bigskip

\toclesslab\subsection{Definitions}{sec-MDPs_definitions}

	A \textit{Markov decision process (MDP)} is a stochastic process that models the evolution of a system with which an agent is allowed to interact. Formally, an MDP is defined as a collection
	\begin{equation}
		\left<\SF{S},\SF{A},\{T^a\}_{a\in\SF{A}},r\right>
	\end{equation}
	consisting of the following elements:
	\begin{itemize}
		\item A set $\SF{S}$ of the allowed \textit{states} of the system. We consider finite state sets throughout this work. The sequence $(S(t):t\in\mathbb{N})$ of random variables taking values in $\SF{S}$ describes the state of the system at all times $t\in\mathbb{N}$.
		
		\item A set $\SF{A}$ of \textit{actions} that the agent is allowed to perform on the system. We consider finite action sets throughout this work. The sequence $(A(t):t\in\mathbb{N})$ of random variables taking values in $\SF{S}$ describes the action taken by the agent at all times $t\in\mathbb{N}$.
		
		\item A set $\{T^a\}_{a\in\SF{A}}$ of \textit{transition matrices}, which are stochastic matrices with domain $\Delta_{\SF{S}}$ and codomain $\Delta_{\SF{S}}$. Specifically,
			\begin{equation}
				T^a=\sum_{s,s'\in\SF{S}}\Pr[S(t+1)=s'|S(t)=s,A(t)=a]\ketbra{s'}{s}
			\end{equation}
			for all $t\in\mathbb{N}$. These matrices determine how the system evolves from one time to the next conditioned on the actions of the agent.
			
		\item A function $r:\SF{S}\times\SF{A}\to\mathbb{R}$ that quantifies the \textit{reward} that the agent receives at every time step based on the current state of the system and the action that it takes.
	\end{itemize}
	
	The \textit{history} up to time $t\in\mathbb{N}$ of an MDP is the random sequence $H(t)\coloneqq(S(1),A(1),\dotsc,A(t-1),S(t))$, with $H(1)=S(1)$. By the Markovian nature of an MDP, the probability distribution of every history $h^t=(s_1,a_1,\dotsc,a_{t-1},s_t)$ is equal to
	\begin{equation}
		\Pr[H(t)=h^t]=\Pr[S(1)=s_1]\prod_{j=1}^{t-1} T^{a_{j}}(s_{j+1};s_{j})d_j(s_{j})(a_{j}),
	\end{equation}
	where
	\begin{equation}\label{eq-decision_function}
		d_j(s_{j})(a_{j})\coloneqq\Pr[A(j)=a_j|S(j)=s_j]
	\end{equation}
	is the probability distribution of actions at time $j$ conditioned on the current state of the system. We refer to $d_j:\SF{S}\times\SF{A}\to[0,1]$ as a \textit{decision function}. Note that $\sum_{a\in\SF{A}}d_j(s)(a)=1$ for all $s\in\SF{S}$. The sequence
	\begin{equation}
		\pi=(d_1,d_2,\dotsc)
	\end{equation}
	of decision functions at all times $t\in\mathbb{N}$ is known as a \textit{policy} of the agent. In the context of this work, policies should be thought of as synonymous with \textit{protocols} for quantum networks.
	
	Given a decision function $d$, we define the following linear operators acting on $\Delta_{\SF{S}}$:
	\begin{equation}\label{eq-decision_function_matrix}
		D_a^d\coloneqq\sum_{s\in\SF{S}}d(s)(a)\ketbra{s}{s},\quad\forall~a\in\SF{A}.
	\end{equation}
	Then, it is straightforward to show that the linear operator
	\begin{equation}\label{eq-transition_matrix_decision}
		P^d\coloneqq\sum_{a\in\SF{A}}T^a D_a^d
	\end{equation}
	from $\Delta_{\SF{S}}$ to $\Delta_{\SF{S}}$ is a stochastic matrix with elements
	\begin{equation}
		\bra{s'}P^d\ket{s}=\Pr[S(t+1)=s'|S(t)=s]
	\end{equation}
	for all $t\in\mathbb{N}$ and all $s,s'\in\SF{S}$.
	
	\begin{remark}
		Observe that for a fixed decision function $d$, the set $\{D_a^d\}_{a\in\SF{A}}$ of linear operators defined in \eqref{eq-decision_function_matrix} forms a positive operator-valued measure (POVM). Indeed, by definition, all of the operators are positive semidefinite; furthermore, by definition of the decision function in \eqref{eq-decision_function},
		\begin{align}
			\sum_{a\in\SF{A}}D_a^d&=\sum_{a\in\SF{A}}\sum_{s\in\SF{S}} d(s)(a)\ketbra{s}{s}\\
			&=\sum_{s\in\SF{S}}\underbrace{\left(\sum_{a\in\SF{A}} d(s)(a)\right)}_{=1~\forall\,s\in\SF{S}}\ketbra{s}{s}\\
			&=\sum_{s\in\SF{S}}\ketbra{s}{s}\\
			&=\mathbbm{1}_{\SF{S}}.
		\end{align}
	\end{remark}
	
	The transition matrices $P^d$ as defined in \eqref{eq-transition_matrix_decision} allow us to determine the probability distribution of the state of the system at every time $t\in\mathbb{N}$ for a given policy. Specifically, for a policy $\pi=(d_1,d_2,\dotsc)$,
	\begin{align}
		\ket{S(t)}_{\pi}&\coloneqq\sum_{s\in\SF{S}}\Pr[S(t)=s]_{\pi}\ket{s}\label{eq-state_time_policy_1}\\
		&=P^{d_{t-1}}\dotsb P^{d_2}P^{d_1}\ket{S(1)},\label{eq-state_time_policy_2}
	\end{align}
	where
	\begin{equation}
		\ket{S(1)}\coloneqq\sum_{s\in\SF{S}}\Pr[S(1)=s]\ket{s}
	\end{equation}
	is the probability distribution for the system at the initial time $t=1$.

\toclesslab\subsubsection{MDPs with absorbing states}{sec-MDPs_absorbing}

	We call a state $s\in\SF{S}$ \textit{absorbing} if $T^a\ket{s}=\ket{s}$ for all $a\in\SF{A}$. In other words, once the system reaches the state $s$ it always stays there, meaning that $P^d\ket{s}=\ket{s}$ for all decision functions $d$. Every state that is not absorbing is called \textit{transient} if there is non-zero probability that, starting from such a state, the system will eventually reach an absorbing state. We can partition the set $\SF{S}$ of all states into disjoint sets: $\SF{S}=\SF{S}_{\text{tra}}\cup\SF{S}_{\text{abs}}$, where $\SF{S}_{\text{abs}}$ is the set of absorbing states and $\SF{S}_{\text{tra}}$ is the set of transient states. We can then rewrite the set $\{\ket{s}\}_{s\in\SF{S}}$ as $\{\ket{0,s}\}_{s\in\SF{S}_{\text{tra}}}\cup\{\ket{1,s}\}_{s\in\SF{S}_{\text{abs}}}$, leading to the following block structure for the transition matrices $T^a$:
	\begin{equation}
		T^a=\ketbra{0}{0}\otimes T_{0\to 0}^a+\ketbra{0}{1}\otimes T_{1\to 0}^a+\ketbra{1}{0}\otimes T_{0\to 1}^a+\ketbra{1}{1}\otimes T_{1\to 1}^a,
	\end{equation}
	where $T_{0\to 0}^a$ is the block describing transitions between transient states, $T_{1\to 0}^a$ is the block describing transition between an absorbing state and a transient state, $T_{0\to 1}^a$ is the block describing transitions between a transient state and an absorbing state, and $T_{1\to 1}^a$ is the block describing transitions between absorbing states. Note that by our definition of an absorbing state, $T_{1\to 0}^a=0$ and $T_{1\to 1}^a=\mathbbm{1}_{\SF{S}_{\text{abs}}}$ for all $a\in\SF{A}$. Similarly, for a decision function $d$, we can write the matrices $D_a^d$, $a\in\SF{A}$, in block form as
	\begin{align}
		D_a^d&=\ketbra{0}{0}\otimes D_a^d(0)+\ketbra{1}{1}\otimes D_a^d(1),\\
		D_a^d(0)&=\sum_{s\in\SF{S}_{\text{trans}}} d(s)(a)\ketbra{s}{s},\\
		D_a^d(1)&=\sum_{s\in\SF{S}_{\text{abs}}}d(s)(a)\ketbra{s}{s}.
	\end{align}
	Consequently, the transition matrix $P^d$ in \eqref{eq-transition_matrix_decision} has the form
	\begin{align}
		P^d&=\ketbra{0}{0}\otimes Q^d+\ketbra{1}{0}\otimes R^d+\ketbra{1}{1}\otimes\mathbbm{1}_{\SF{S}_{\text{abs}}}\\
		&\equiv\begin{pmatrix} Q^d & 0 \\ R^d & \mathbbm{1}_{\SF{S}_{\text{abs}}}\end{pmatrix},\label{eq-transition_matrix_MDP_abs}
	\end{align}
	where
	\begin{align}
		Q^d&=\sum_{a\in\SF{A}}T_{0\to 0}^aD_a^d(0),\label{eq-transition_matrix_MDP_abs_tra}\\
		R^d&=\sum_{a\in\SF{A}}T_{0\to 1}^aD_a^d(0).
	\end{align}

\toclesslab\subsection{Figures of merit}{sec-MDPs_figs_of_merit}

	While the primary figure of merit in a Markov decision process is the expected reward, in this work we are mostly interested in what we call \textit{functions of state} (such as the fidelity) and the absorption time (corresponding to the waiting time for a virtual link).

	\paragraph*{Functions of state.} In this work, we are also interested in functions $f:\textsf{S}\to\mathbb{R}$ of the state of the system. We can associate to such functions the vector
	\begin{equation}
		\ket{f}\coloneqq\sum_{s\in\SF{S}}f(s)\ket{s}.
	\end{equation}
	Then, for a policy $\pi=(d_1,d_2,\dotsc)$, we are interested in the expected value of the random variable $f(S(t))$ for all $t\in\mathbb{N}$, i.e., the quantity
	\begin{equation}
		\mathbb{E}[f(S(t))]_{\pi}=\sum_{s\in\textsf{S}}f(s)\Pr[S(t)=s]_{\pi}.
	\end{equation}
	Using \eqref{eq-state_time_policy_1} and \eqref{eq-state_time_policy_2}, we immediately obtain
	\begin{align}
		\mathbb{E}[f(S(t))]_{\pi}&=\braket{f}{S(t)}_{\pi}\\
		&=\bra{f}P^{d_{t-1}}\dotsb P^{d_2}P^{d_1}\ket{S(1)}.
	\end{align}
	With respect to stationary policies $\pi=(d,d,\dotsc)$, we are also interested in the asymptotic quantity
	\begin{equation}\label{eq-MDP_func_inf}
		\lim_{t\to\infty}\mathbb{E}[f(S(t))]_{(d,d,\dotsc)}=\lim_{t\to\infty}\bra{f}(P^d)^{t-1}\ket{S(1)},
	\end{equation}
	if the limit exists, along with the optimal value
	\begin{equation}\label{eq-MDP_opt_func_inf}
		\sup_d\lim_{t\to\infty}\bra{f}(P^d)^{t-1}\ket{S(1)}.
	\end{equation}
	
	\paragraph*{Expected waiting time to absorption.} Finally, for MDPs with absorbing states, we are interested in the expected waiting time to absorption with respect to stationary policies, i.e., the expected number of time steps needed to reach an absorbing state when starting from a transient state and following a stationary policy. It is a standard result of the theory of Markov chains (see, e.g., Ref.~[\onlinecite[Theorem~9.6.1]{Stewart09_book}]) that in this setting, with a transition matrix as in \eqref{eq-transition_matrix_MDP_abs}, if the initial distribution of states is given by the probability vector $\ket{0,S_{\text{tr}}(1)}$ (entirely in the transient block), then the expected waiting time to absorption with respect to the policy $(d,d,\dotsc)$ is $\bra{\gamma}_{\SF{S}_{\text{tra}}}(\mathbbm{1}_{\SF{S}_{\text{tra}}}-Q^d)^{-1}\ket{S_{\text{tra}}(1)}$. We are then interested in the following optimal value:
	\begin{equation}\label{eq-MDP_opt_absorb_time}
		\inf_d\bra{\gamma}_{\SF{S}_{\text{tra}}}(\mathbbm{1}_{\SF{S}_{\text{tra}}}-Q^d)^{-1}\ket{S_{\text{tra}}(1)}.
	\end{equation}

\toclesslab\subsection{Linear programs}{sec-MDP_linear_prog}

	We now present linear programs for estimating the values of the figures of merit presented in the previous section in the steady-state limit with a time-homogeneous (stationary) policy. Linear programs have been used for MDPs in various different ways~\cite{Der62,KA68,OM68,FS02_book,Put14_book}. The linear programs we consider here are similar to those in the aforementioned references, but we present them here in the notation introduced at the beginning of this section. We start by considering a general MDP, not necessarily with absorbing states.
	
	\begin{proposition}[Linear program for the steady-state expected function value]\label{prop-lin_prog_func_1}
		Consider an MDP as defined in Sec.~\ref{sec-MDPs_definitions} along with a function $f:\SF{S}\to\mathbb{R}$ of the state of the MDP. Among decision functions $d$ for which the limit $\lim_{t\to\infty}(P^d)^{t-1}$ exists, the optimal steady-state expected value of $f$, namely the quantity in \eqref{eq-MDP_opt_func_inf}, is equal to the solution of the following linear program:
		\begin{equation}
			\begin{array}{l l} \textnormal{maximize} & \braket{f}{v} \\[0.2cm] \textnormal{subject to} & 0\leq\ket{w_a}\leq\ket{v}\leq 1\quad\forall~a\in\SF{A}, \\[0.2cm] & \displaystyle\sum_{s\in\SF{S}}\braket{s}{v}=1, \\[0.5cm] & \displaystyle\sum_{a\in\SF{A}}\ket{w_a}=\ket{v}=\sum_{a\in\SF{A}}T^a\ket{w_a}, \end{array}
		\end{equation}
		where the optimization is with respect to $\abs{\SF{S}}$-dimensional vectors $\ket{v}$ and $\ket{w_a}$, $a\in\SF{A}$, and the inequality constraints are component-wise. Every set of feasible points $\ket{v}$, $\ket{w_a}$, $a\in\SF{A}$, of this linear program defines a stationary policy with decision function $d$ as follows:
		\begin{equation}\label{eq-lin_prog_func_1_decision_func}
			d(s)(a)=\frac{\braket{s}{w_a}}{\braket{s}{v}},\quad\forall~s\in\SF{S},\,a\in\SF{A}.
		\end{equation}
	\end{proposition}
	
	\begin{proof}
		By the assumption that $\lim_{t\to\infty}(P^d)^{t-1}$ exists and is unique, we have that $\lim_{t\to\infty}(P^d)^{t-1}=\ketbra{v}{\gamma}$ for some probability vector $\ket{v}$ such that $\braket{s}{v}\in(0,1]$ for all $s\in\SF{S}$ and $P^d\ket{v}=\ket{v}$. (Note that all elements $\braket{s}{v}$ are strictly greater than zero; see, e.g., Ref.~[\onlinecite[Theorem~A.2]{Put14_book}].) Therefore, $\lim_{t\to\infty}\mathbb{E}[f(S(t))]_{(d,d,\dotsc)}=\braket{f}{v}$. Furthermore, using the fact that $P^d=\sum_{a\in\SF{A}}T^aD_a^d$, we have $\sum_{a\in\SF{A}}T_aD_a^d\ket{v}=\ket{v}$. Now, let $\ket{w_a}\coloneqq D_a^d\ket{v}$. By recalling that $D_a^d=\sum_{s\in\SF{S},a\in\SF{A}}d(s)(a)\ketbra{s}{s}$, we see that $\braket{s}{w_a}=d(s)(a)\braket{s}{v}$, so that the elements $\braket{s}{w_a}$ can be thought of as the joint probabilities $\Pr[S(t)=s,A(t)=a]$ (in the steady state). This means that $\braket{s}{w_a}\in[0,1]$, and also that $\braket{s}{w_a}\leq\braket{s}{v}$, for all $s\in\SF{S}$ and $a\in\SF{A}$. Then, using the fact that $\sum_{a\in\SF{A}}D_a^d=\mathbbm{1}_{\SF{S}}$, we obtain $\sum_{a\in\SF{A}}\ket{w_a}=\ket{v}$. By uniqueness of the stationary probability vector $\ket{v}$, the result follows.
		
		The construction of the decision function $d$ in \eqref{eq-lin_prog_func_1_decision_func} follows from Ref.~[\onlinecite[Theorem~8.8.2]{Put14_book}]. Both $\braket{s}{w_a}$ and $\braket{s}{v}$ are obtained from the linear program, and because $\braket{s}{v}$ is strictly positive, we can divide in order to get $d(s)(a)$, and the condition $\sum_{a\in\SF{A}}\ket{w_a}$ guarantees that $\sum_{a\in\SF{A}}d(s)(a)=1$ for all $s\in\SF{S}$, as required for a conditional probability. This completes the proof.
	\end{proof}

	We now consider the optimal expected value of a function $f:\textsf{S}\to\mathbb{R}$ in the steady-state limit when there are absorbing states in the MDP. We now show how to obtain an upper bound using a linear program.
	
	\begin{proposition}[Linear program for the steady-state expected function value for an MDP with absorbing states]\label{prop-lin_prog_func_2}
		Consider an MDP with absorbing states, as defined in Sec.~\ref{sec-MDPs_absorbing}, along with a function $f:\SF{S}\to\mathbb{R}$ of the state of the MDP. Then, the optimal steady-state expected value of $f$, namely the quantity in \eqref{eq-MDP_opt_func_inf}, is given by the solution to the following linear program:
		\begin{equation}\label{eq-lin_prog_func_2}
			\begin{array}{l l} \textnormal{maximize} & \braket{f}{1,x} \\[0.2cm] \textnormal{subject to} & 0\leq\ket{w_a}\leq\ket{x}\quad\forall~a\in\textnormal{\textsf{A}},\\[0.2cm] & 0\leq\ket{v_a}\leq\ket{y}\quad\forall~a\in\textnormal{\textsf{A}},\\[0.2cm] & \displaystyle\sum_{a\in\textnormal{\textsf{A}}}\sum_{i=0}^1 T_{0\to i}^a\ket{w_a}=\ket{0}\ket{x},\\[0.5cm] & \displaystyle\sum_{a\in\textnormal{\textsf{A}}}\ket{w_a}=\ket{x},\\[0.5cm] & \displaystyle\ket{y}-\sum_{a\in\textnormal{\textsf{A}}}T_{0\to 0}^a\ket{v_a}=\ket{S_{\text{tra}}(1)},\\[0.5cm] & \displaystyle \sum_{a\in\textnormal{\textsf{A}}}\ket{v_a}=\ket{y},\\[0.5cm] & \displaystyle\sum_{a\in\textnormal{\textsf{A}}}T_{0\to 1}^a\ket{v_a}=\ket{x}, \end{array}
		\end{equation}
		where $\ket{S_{\text{tra}}(1)}$ is the initial $\abs{\SF{S}_{\text{tra}}}$-dimensional probability vector of the MDP (entirely in the transient block), and the optimization is with respect to $\abs{\SF{S}_{\text{tra}}}$-dimensional vectors $\ket{x},\ket{y},\ket{v_a},\ket{w_a}$, $a\in\SF{A}$, and the inequality constraints are component-wise. Every set of feasible points $\ket{x},\ket{y},\ket{v_a},\ket{w_a}$, $a\in\SF{A}$, of this linear program defines a stationary policy with decision function $d$ for the transient states as follows:
		\begin{equation}\label{eq-lin_prog_func_2_decision_function}
			\forall~s\in\SF{S}_{\text{tra}},\,a\in\SF{A}: d(s)(a)=\frac{\braket{s}{w_a}}{\braket{s}{x}}.
		\end{equation}
		If $\braket{s}{x}=0$, as well as for $s\in\SF{S}_{\text{abs}}$, we can set $d(s)$ to an arbitrary probability distribution.
	\end{proposition}
	
	\begin{proof}
		For every decision function $d$, for the transition matrix in \eqref{eq-transition_matrix_MDP_abs} it is known that~\cite{Stewart09_book}
		\begin{equation}
			\lim_{t\to\infty}(P^d)^{t-1}=\begin{pmatrix} 0 & 0 \\ R^d(\mathbbm{1}_{\SF{S}_{\text{tra}}}-Q^d)^{-1} & \mathbbm{1}_{\mathsf{S}_{\text{abs}}} \end{pmatrix}.
		\end{equation}
		Therefore,
		\begin{align}
			\lim_{t\to\infty}\mathbb{E}[f(S(t))]_{(d,d,\dotsc)}&=\bra{f}\left(\ketbra{1}{0}\otimes(R^d(\mathbbm{1}_{\SF{S}_{\text{tra}}}-Q^d)^{-1})+\ketbra{1}{1}\otimes\mathbbm{1}_{\mathsf{S}_{\text{abs}}}\right)\ket{0}\ket{S_{\text{tra}}(1)}\\
			&=\bra{f}\left(\ketbra{1}{0}\otimes(R^d(\mathbbm{1}_{\SF{S}_{\text{tra}}}-Q^d)^{-1})\right)\ket{0}\ket{S_{\text{tra}}(1)}
		\end{align}
		Now, let
		\begin{align}
			\ket{x}&=R^d(\mathbbm{1}_{\SF{S}_{\text{tra}}}-Q^d)^{-1}\ket{S_{\text{tra}}(1)},\\
			\ket{y}&=(\mathbbm{1}_{\SF{S}_{\text{tra}}}-Q^d)^{-1}\ket{S_{\text{tra}}(1)}.
		\end{align}
		Then, it is easy to verify that
		\begin{align}
			P^d\ket{0}\ket{x}&=\ket{0}\ket{x},\label{eq-lin_prog_func_2_pf1}\\
			(\mathbbm{1}_{\SF{S}_{\text{tra}}}-Q^d)\ket{y}&=\ket{S_{\text{tra}}(1)},\label{eq-lin_prog_func_2_pf2}\\
			R^d\ket{y}&=\ket{x}.\label{eq-lin_prog_func_2_pf3}
		\end{align}
		Then, as in the proof of Proposition~\ref{prop-lin_prog_func_1}, we define
		\begin{align}
			\ket{w_a}&=D_a^d(0)\ket{x},\\
			\ket{v_a}&=D_a^d(0)\ket{y},
		\end{align}
		for all $a\in\SF{A}$. Then, by the same arguments as in the proof of Proposition~\ref{prop-lin_prog_func_1}, we have that $0\leq\ket{w_a}\leq\ket{x}$ and $0\leq\ket{v_a}\ket{y}$ for all $a\in\SF{A}$, and $\sum_{a\in\SF{A}}\ket{w_a}=\ket{x}$, $\sum_{a\in\SF{A}}\ket{v_a}=\ket{y}$. Combining the constraints in \eqref{eq-lin_prog_func_2_pf1}--\eqref{eq-lin_prog_func_2_pf3} along with the definitions and constraints for the vectors $\ket{w_a}$ and $\ket{v_a}$ leads to the constraints in \eqref{eq-lin_prog_func_2}. Optimizing with respect to these constraints therefore results in a value that cannot be less than the value in \eqref{eq-MDP_opt_func_inf}, leading to the desired result. The construction of the decision function in \eqref{eq-lin_prog_func_2_decision_function} results from the definition of $\ket{w_a}$ and the reasoning analogous to that given in the proof of Proposition~\ref{prop-lin_prog_func_1}. This completes the proof.
	\end{proof} 
	
	Finally, we show how to bound the expected absorption time of an MDP with absorbing states using a linear program.
	
	\begin{proposition}[Linear program for the expected waiting time to absorption]\label{prop-lin_prog_func_3}
		Consider an MDP with absorbing states, as defined in Sec.~\ref{sec-MDPs_absorbing}. The minimum expected waiting time to reach an absorbing state, namely the quantity in \eqref{eq-MDP_opt_absorb_time}, is given by the following linear program:
		\begin{equation}\label{eq-lin_prog_func_3}
			\begin{array}{l l} \textnormal{minimize} & \braket{\gamma}{x} \\[0.2cm] \textnormal{subject to} & 0\leq\ket{w_a}\leq\ket{x}\quad\forall~a\in\textnormal{\textsf{A}}, \\[0.3cm] & \displaystyle\ket{x}-\sum_{a\in\textnormal{\textsf{A}}}T_{0\to 0}^a\ket{w_a}=\ket{S_{\text{tra}}(1)}, \\[0.5cm] & \displaystyle\sum_{a\in\textnormal{\textsf{A}}}\ket{w_a}=\ket{x}, \end{array}
		\end{equation}
		where $\ket{S_{\text{tra}}(1)}$ is the initial $\abs{\SF{S}_{\text{tra}}}$-dimensional probability vector of the MDP (entirely in the transient block), and the optimization is with respect to $\abs{\SF{S}_{\text{tra}}}$-dimensional vectors $\ket{x}$, $\ket{w_a}$, $a\in\SF{A}$, and the inequality constraints are component-wise. Every set of feasible points $\ket{x}$, $\ket{w_a}$, $a\in\SF{A}$, of this linear program defines a stationary policy with decision function $d$ for the transient states as follows:
		\begin{equation}\label{eq-lin_prog_func_3_decision_func}
			d(s)(a)=\frac{\braket{s}{w_a}}{\braket{s}{x}}\quad\forall~s\in\SF{S}_{\text{tra}},\,a\in\SF{A}.
		\end{equation}
		If $\braket{s}{x}=0$, then we can set $d(s)$ to be an arbitrary probability distribution.
	\end{proposition}
	
	\begin{proof}
		We start with the fact that, for a given decision function $d$, the expected waiting time to absorption is given by $\bra{\gamma}_{\SF{S}_{\text{tra}}}(\mathbbm{1}_{\SF{S}_{\text{tra}}}-Q^d)^{-1}\ket{S_{\text{tra}}(1)}$. Now, let
		\begin{equation}
			\ket{x}=(\mathbbm{1}_{\SF{S}_{\text{tra}}}-Q^d)^{-1}\ket{S_{\text{tra}}(1)}.
		\end{equation}
		Then, we have that $(\mathbbm{1}_{\SF{S}_{\text{tra}}}-Q^d)\ket{x}=\ket{S_{\text{tra}}(1)}$. Using the definition of $Q^d$ in \eqref{eq-transition_matrix_MDP_abs_tra}, we obtain $\ket{x}-\sum_{a\in\SF{A}}T_{0\to 0}^a\ket{w_a}=\ket{S_{\text{tra}}(1)}$, where $\ket{w_a}=D_a^d(0)\ket{x}$, $a\in\SF{A}$. The definition of $\ket{w_a}$ is the same as in the proof of Proposition~\ref{prop-lin_prog_func_2}, thus for the same reasons as in that proof we have that $0\leq\ket{w_a}\leq\ket{x}$ for all $a\in\SF{A}$ and $\sum_{a\in\SF{A}}\ket{w_a}=\ket{x}$. It is then clear that, by optimizing the quantity $\braket{\gamma}{x}$ with respect to $\ket{x}$ and $\ket{w_a}$, $a\in\SF{A}$, the result can be no greater than the optimal expected time to reach an absorbing state, leading to the desired result. Then, given feasible points $\ket{x}$ and $\ket{w_a}$, $a\in\SF{A}$, the function $d$ defined by \eqref{eq-lin_prog_func_3_decision_func} is a valid decision function whenever $\braket{s}{x}$ is non-zero. This completes the proof. 
	\end{proof}

\toclesslab\section{Quantum states and channels}{sec-prelim_q_states_chan}
	
	In this section, we summarize some standard material on quantum states and channels, which can be found in Refs.~[\onlinecite{NC00_book,Hol12_book,Wilde17_book,Wat18_book,KW20_book}]. Given a quantum system $A$ with associated Hilbert space $\SF{H}_A$, the \textit{quantum state} of $A$ is given by a density operator acting on $\SF{H}_A$: a linear operator $\rho_A:\SF{H}_A\to\SF{H}_A$ that is positive semi-definite and has unit trace, i.e., $\rho_A\geq 0$ and $\Tr[\rho_A]=1$.
	
	A type of quantum state that we frequently encounter in this work is a \textit{classical-quantum} state, which is a quantum state of the form
	\begin{equation}\label{eq-classical_quantum_state}
		\rho_{XA}=\sum_{x\in\SF{X}}p(x)\ketbra{x}{x}_X\otimes\rho_A^x,
	\end{equation}
	where $\SF{X}$ is a finite set, $p:\SF{X}\to[0,1]$ is a probability mass function, and $\{\rho_A^x\}_{x\in\SF{X}}$ is a set of quantum states. Classical-quantum states can be used to model scenarios in which classical information accompanies the state of a quantum system. Specifically, if a quantum system $A$ is prepared in a state from the set $\{\rho_A^x\}_{x\in\SF{X}}$ according to the probability distribution defined by $p$, then knowledge of the label $x\in\SF{X}$ is stored in the classical register~$X$.
	
	Given two quantum states $\rho$ and $\sigma$, their \textit{fidelity} is defined to be ~\cite{Uhl76}
	\begin{equation}
		F(\rho,\sigma)\coloneqq\left(\Tr\!\left[\sqrt{\sqrt{\rho}\sigma\sqrt{\rho}}\right]\right)^2.
	\end{equation}
	The fidelity quantifies the closeness of two quantum states. In particular, $F(\rho,\sigma)=1$ if and only if $\rho=\sigma$, and $F(\rho,\sigma)=0$ if and only if $\rho$ and $\sigma$ are supported on orthogonal subspaces. If one of the states, say $\sigma$, is pure, then it is straightforward to show that
	\begin{equation}
		F(\rho,\ketbra{\psi}{\psi})=\bra{\psi}\rho\ket{\psi}.
	\end{equation}
	The fidelity is also \textit{multiplicative}, meaning that
	\begin{equation}\label{eq-fidelity_multiplicative}
		F(\rho_1\otimes\rho_2,\sigma_1\otimes\sigma_2)=F(\rho_1,\sigma_1)F(\rho_2,\sigma_2),
	\end{equation}
	for all states $\rho_1,\rho_2,\sigma_1,\sigma_2$.
	
	A \textit{quantum channel} is a mathematical description of the evolution of a quantum system. Let $\Lin(\SF{H}_A)$ denote the vector space of linear operators acting on the Hilbert space $\SF{H}_A$. A linear map $\mathcal{T}:\Lin(\SF{H}_A)\to\Lin(\SF{H}_B)$ is often called a \textit{superoperator}, and it is such that
	\begin{equation}
		\mathcal{T}(\alpha X+\beta Y)=\alpha\mathcal{T}(X)+\beta\mathcal{T}(Y)
	\end{equation}
	for all $\alpha,\beta\in\mathbb{C}$ and all $X,Y\in\Lin(\SF{H}_A)$. It is often helpful to explicitly indicate the input and output Hilbert spaces of a superoperator $\mathcal{T}:\Lin(\SF{H}_A)\to\Lin(\SF{H}_B)$ by writing $\mathcal{T}_{A\to B}$. The identity superoperator is denoted by $\id_A$, and it satisfies $\id_A(X)=X$ for all $X\in\Lin(\SF{H}_A)$.
	
	A quantum channel $\mathcal{N}_{A\to B}$ is a linear, \textit{completely positive}, and \textit{trace-preserving} superoperator acting on the vector space $\Lin(\SF{H}_A)$ of linear operators of the Hilbert space $\SF{H}_A$ of the quantum system $A$. Given an input state $\rho_A$ of the system $A$, the output is the state of a new quantum system $B$ given by $\mathcal{N}_{A\to B}(\rho_A)$.
	\begin{itemize}
		\item A superoperator $\mathcal{N}$ is \textit{completely positive} if the map $\id_{k}\otimes\mathcal{N}$ is positive for all $k\in\mathbb{N}$, where $\id_k:\Lin(\mathbb{C}^k)\to\Lin(\mathbb{C}^k)$ is the identity superoperator. In other words, $(\id_k\otimes\mathcal{N})(X)\geq 0$ for every linear operator $X\geq 0$.
		\item A superoperator $\mathcal{N}$ is \textit{trace preserving} if $\Tr[\mathcal{N}(X)]=\Tr[X]$ for every linear operator $X$.
	\end{itemize}

\toclesslab\subsection{Quantum instruments}{sec-quantum_instruments}
	
	Let $A$ be a quantum system with associated Hilbert space $\SF{H}_A$. A \textit{measurement} of $A$ is defined by a finite set $\{M_A^x\}_{x\in\SF{X}}$ of linear operators acting on $\SF{H}_A$, called a \textit{positive operator-valued measure (POVM)}, that satisfies the following two properties.
	\begin{itemize}
		\item $M_A^x \geq 0$ for all $x\in\SF{X}$;
		\item $\displaystyle \sum_{x\in\SF{X}} M_A^x=\mathbbm{1}_A$.
	\end{itemize}
	Elements of the set $\SF{X}$ label the possible outcomes of the measurement. Given a state $\rho_A$, the probability of obtaining the outcome $x\in\SF{X}$ is given by the Born rule as $\Tr[M_A^x\rho_A]$.
	
	As a generalization of a measurement, a \textit{quantum instrument} is a finite set $\{\mathcal{M}^x\}_{x\in\SF{X}}$ of completely positive trace non-increasing maps such that the sum $\sum_{x\in\SF{X}}\mathcal{M}^x$ is a trace-preserving map, and thus a quantum channel. (A trace non-increasing map $\mathcal{N}_{A\to B}$ satisfies $\Tr[\mathcal{N}_{A\to B}(X_A)]\leq\Tr[X_A]$ for every positive semi-definite linear operator $X_A$.) The \textit{quantum instrument channel} $\mathcal{M}$ associated to the quantum instrument $\{\mathcal{M}^x\}_{x\in\SF{X}}$ is defined as
	\begin{equation}\label{eq-quantum_instrument_channel}
		\mathcal{M}(\cdot)\coloneqq\sum_{x\in\SF{X}} \ketbra{x}{x}\otimes\mathcal{M}^x(\cdot).
	\end{equation}
	A quantum instrument $\{\mathcal{M}^x\}_{x\in\SF{X}}$ can be thought of as a generalized form of a measurement, in which the completely positive maps $\mathcal{M}^x$ represent the evolution of the quantum system conditioned on the outcome $x$. The trace non-increasing property of the maps $\mathcal{M}^x$ represents the fact that the outcome $x$ occurs probabilistically. Specifically, the probability of obtaining the outcome $x$ is equal to $\Tr[\mathcal{M}^x(\rho)]$, which can be thought of as a generalized form of the Born rule stated above. The quantum instrument channel in \eqref{eq-quantum_instrument_channel} can be thought of as an operation that stores both the outcome $x$ of the instrument in the classical register as well as the corresponding output state.

\toclesslab\subsection{LOCC channels}{sec-LOCC_channels}

	\begin{figure*}
		\centering
		\includegraphics[scale=1]{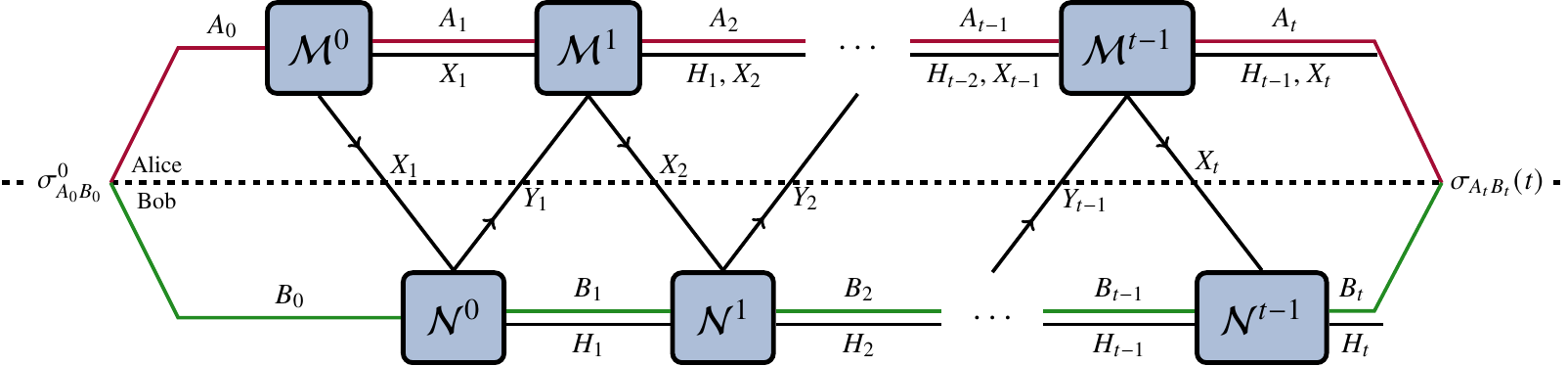}
		\caption{Depiction of a $t$-round LOCC protocol between Alice and Bob. The channels $\mathcal{M}^0,\mathcal{M}^1,\dotsc,\mathcal{M}^{t-1}$ represent Alice's local operations, the channels $\mathcal{N}^0,\mathcal{N}^1,\dotsc,\mathcal{N}^{t-1}$ represent Bob's local operations, and the registers $H_1,H_2,\dotsc,H_t$ represent the classical communication to and from Alice and Bob. The final quantum state $\sigma_{A_tB_t}(t)$ shared by Alice and Bob has the form shown in \eqref{eq-LOCC_protocol_final_state}.}\label{fig-LOCC}
	\end{figure*}

	Consider two parties, Alice and Bob, who are spatially separated. Suppose that they have the ability to perform arbitrary quantum operations (quantum channels, measurements, instruments) in their respective labs and that they are connected by a classical communication channel. It is often the case that Alice and Bob are also connected by a quantum channel and/or share an entangled quantum state, and their task is to make use of these \textit{resources} as sparingly as possible in order to accomplish their desired goal. Their local operations and classical communication (LOCC) can be used freely to help with achieving the goal, which could be feedback-assisted quantum communication~\cite{BDSW96} (see also Ref.~[\onlinecite{KW20_book}]), which includes quantum teleportation and entanglement swapping \r~\cite{BBC+93,Vaidman94,BFK00,ZZH93}, or it could be entanglement distillation~\cite{BDSW96}. In the network setting, the task is repeater-assisted quantum communication. Here, we focus on the basic mathematical definition of an LOCC channel and provide some examples. For more mathematical details about LOCC channels, we refer to Ref.~[\onlinecite{CLM+14}].
	
	Consider the scenario shown in Fig.~\ref{fig-LOCC}, which is an \textit{LOCC protocol} with $t$ rounds. The LOCC channel corresponding to this protocol, i.e., the channel mapping the input systems $A_0B_0$ to the output systems $A_tB_t$ at the end of the $t^{\text{th}}$ round of the protocol, can be derived as follows.
	
	We start with the initial state $\sigma_{A_0B_0}^0$ shared by Alice and Bob. In the first round, Alice acts on her system $A_0$ with a quantum instrument channel $\mathcal{M}_{A_0\to X_1A_1}^0$, which leads to the following output:
	\begin{equation}
		\sigma_{A_0B_0}^0\mapsto \mathcal{M}_{A_0\to X_1A_1}(\sigma_{A_0B_0}^0)=\sum_{x_1\in\SF{X}_1}\ketbra{x_1}{x_1}_{X_1}\otimes\mathcal{M}_{A_0\to A_1}^{0;x_1}(\sigma_{A_0B_0}^0),
	\end{equation}
	where $\SF{X}_1$ is a finite set and $\{\mathcal{M}_{A_0\to A_1}^{0;x_1}\}_{x_1\in\SF{X}_1}$ is a quantum instrument. The classical register $X_1$ is then communicated to Bob, who applies the conditional quantum instrument channel given by
	\begin{equation}
		\mathcal{N}_{X_1B_0\to X_1Y_1B_1}^0(\ketbra{x_1}{x_1}_{X_1}\otimes \tau_{B_0})=\ketbra{x_1}{x_1}_{X_1}\otimes\mathcal{N}_{B_0\to Y_1B_1}^{0;x_1}(\tau_{B_0}),
	\end{equation}
	where $\mathcal{N}_{B_0\to Y_1B_1}^{0;x_1}$ is a quantum instrument channel, i.e.,
	\begin{equation}
		\mathcal{N}_{B_0\to Y_1B_1}^{0;x_1}(\tau_{B_0})=\sum_{y_1\in\SF{Y}_1}\ketbra{y_1}{y_1}_{Y_1}\otimes\mathcal{N}_{B_0\to B_1}^{0;x_1,y_1}(\tau_{B_0}),
	\end{equation}
	with $\{\mathcal{N}_{B_0\to B_1}^{0;x_1,y_1}\}_{y_1\in\SF{Y}_1}$ being a quantum instrument, i.e., every $\mathcal{N}_{B_0\to B_1}^{0;x_1,y_1}$ is a completely positive trace non-increasing map and $\sum_{y_1\in\SF{Y}_1}\mathcal{N}_{B_0\to B_1}^{0;x_1,y_1}$ is a trace-preserving map. Bob sends the outcome $y_1\in\SF{Y}_1$ of the quantum instrument to Alice. This completes the first round, and the quantum state shared by Alice and Bob is
	\begin{align}
		\rho_{X_1Y_1A_1B_1}(1)&\coloneqq\left(\mathcal{M}_{A_0\to X_1A_1}^0\otimes\mathcal{N}_{X_1B_0\to X_1Y_1B_1}^0\right)(\sigma_{A_0B_0}^0)\\
		&=\sum_{\substack{x_1\in\SF{X}_1\\y_1\in\SF{Y}_1}}\ketbra{x_1,y_1}{x_1,y_1}_{X_1Y_1}\otimes\left(\mathcal{M}_{A_0\to A_1}^{0;x_1}\otimes\mathcal{N}_{B_0\to B_1}^{0;x_1,y_1}\right)(\sigma_{A_0B_0}^0)\\
		&\quad=\sum_{\substack{x_1\in\SF{X}_1\\y_1\in\SF{Y}_1}}\ketbra{x_1,y_1}{x_1,y_1}_{X_1Y_1}\otimes\widetilde{\sigma}_{A_1B_1}(1;x_1,y_1),
	\end{align}
	where in the last line we let
	\begin{equation}
		\widetilde{\sigma}_{A_1B_1}(1;x_1,y_1)\coloneqq\left(\mathcal{M}_{A_0\to A_1}^{0;x_1}\otimes\mathcal{N}_{B_0\to B_1}^{0;x_1,y_1}\right)(\sigma_{A_0B_0}^0).
	\end{equation}
	
	Now, in the second round, Alice applies the conditional quantum instrument channel given by
	\begin{equation}
		\mathcal{M}_{X_1Y_1A_1\to X_1Y_1X_2A_2}^1(\ketbra{x_1,y_1}{x_1,y_1}_{X_1Y_1}\otimes\rho_{A_1})\coloneqq\ketbra{x_1,y_1}{x_1,y_1}_{X_1Y_1}\otimes\mathcal{M}_{A_1\to X_2A_2}^{1;x_1,y_1}(\rho_{A_1}),
	\end{equation}
	where $\mathcal{M}_{A_1\to X_2A_2}^{1;x_1,y_1}$ is a quantum instrument channel in which the underlying quantum instrument $\{\mathcal{M}_{A_1\to A_2}^{1;x_1,y_1,x_2}\}_{x_2\in\SF{X}_2}$ is conditioned on the histories $\{(x_1,y_1):x_1\in\SF{X}_1,\,y_1\in\SF{Y}_1\}$ of hers and Bob's prior outcomes. She sends the outcome $x_2\in\SF{X}_2$ of the quantum instrument to Bob, who then applies the conditional quantum instrument channel given by
	\begin{equation}
		\mathcal{N}_{X_1Y_1X_2B_1\to X_1Y_1X_2Y_2B_2}^1(\ketbra{x_1,y_1,x_2}{x_1,y_1,x_2}_{X_1Y_1X_2}\otimes\tau_{B_1})\coloneqq\ketbra{x_1,y_1,x_2}{x_1,y_1,x_2}_{X_1Y_1X_2}\otimes\mathcal{N}_{B_1\to Y_2B_2}^{1;x_1,y_1,x_2}(\tau_{B_1}),
	\end{equation}
	where $\mathcal{N}_{B_1\to Y_2B_2}^{1;x_1,y_1,x_2}$ is a quantum instrument channel in which the underlying quantum instrument $\{\mathcal{N}_{B_1\to B_2}^{1;x_1,y_1,x_2,y_2}\}_{y_2\in\SF{Y}_2}$ depends on the prior outcomes $x_1\in\SF{X}_1,\,y_1\in\SF{Y}_1,\,x_2\in\SF{X}_2$. The outcome of the instrument is $y_2\in\SF{Y}_2$, so that, at the end of the second round, the state shared by Alice and Bob is
	\begin{equation}
		\rho_{X_1Y_1X_2Y_2A_2B_2}(2)=\sum_{h^2}\ketbra{h^2}{h^2}_{H_2}\otimes\widetilde{\sigma}_{A_2B_2}(2;h^2),
	\end{equation}
	where 
	we used the abbreviations
	\begin{align}
		H_2&\equiv X_1Y_1X_2Y_2,\\
		h^2&\equiv (x_1,y_1,x_2,y_2)\in\SF{X}_1\times\SF{Y}_1\times\SF{X}_2\times\SF{Y}_2,
	\end{align}
	and
	\begin{equation}
		\widetilde{\sigma}(2;h^2)\coloneqq\left(\mathcal{M}_{A_1\to A_2}^{1;x_1,y_1,x_2}\circ\mathcal{M}_{A_0\to A_1}^{0;x_1}\otimes\mathcal{N}_{B_1\to B_2}^{1;h^2}\circ\mathcal{N}_{B_0\to B_1}^{0;x_1,y_1}\right)(\sigma_{A_0B_0}^0).
	\end{equation}
	
	Let $h^j\coloneqq(x_1,y_1,\dotsc,x_j,y_j)\in\SF{X}_1\times\SF{Y}_1\times\dotsb\times\SF{X}_j\times\SF{Y}_j$ be the history up to $j$ steps, $j\in\{1,2,\dotsc,\}$. Proceeding in the manner presented above, at the $j^{\text{th}}$ step of the protocol, Alice and Bob apply conditional quantum instrument channels of the form
	\begin{equation}
		\mathcal{M}_{H_{j}A_j\to H_jX_{j+1}A_{j+1}}^j\left(\ketbra{h^j}{h^j}_{H_j}\otimes\rho_{A_j}\right)=\ketbra{h^j}{h^j}_{H_j}\otimes\mathcal{M}_{A_j\to X_{j+1}A_{j+1}}^{j;h^j}(\rho_{A_j})
	\end{equation}
	and
	\begin{equation}
		\mathcal{N}_{H_jX_{j+1}B_j\to H_{j+1}B_{j+1}}^{j}\left(\ketbra{h^j,x_{j+1}}{h^j,x_{j+1}}_{H_jX_{j+1}}\otimes\sigma_{B_j}\right)=\ketbra{h^j,x_{j+1}}{h^j,x_{j+1}}_{H_jX_{j+1}}\otimes\mathcal{N}_{B_j\to Y_{j+1}B_{j+1}}^{j;h^j,x_{j+1}}(\sigma_{B_j}),
	\end{equation}
	where
	\begin{equation}
		\mathcal{M}_{A_j\to X_{j+1}A_{j+1}}^{j;h^j}(\rho_{A_j})=\sum_{x_{j+1}\in\SF{X}_{j+1}}\ketbra{x_{j+1}}{x_{j+1}}_{X_{j+1}}\otimes\mathcal{M}_{A_j\to A_{j+1}}^{j;h^j,x_{j+1}}(\rho_{A_j})
	\end{equation}
	and
	\begin{equation}
		\mathcal{N}_{B_j\to Y_{j+1}B_{j+1}}^{j;h^j,x_{j+1}}(\sigma_{B_j})=\sum_{y_{j+1}\in\SF{Y}_{j+1}}\ketbra{y_{j+1}}{y_{j+1}}_{Y_{j+1}}\otimes\mathcal{N}_{B_j\to B_{j+1}}^{j;h^j,x_{j+1},y_{j+1}}(\sigma_{B_j}).
	\end{equation}
	Therefore, at the end of the $t^{\text{th}}$ round, the classical-quantum state shared by Alice and Bob is
	\begin{equation}\label{eq-cq_state_LOCC}
		\rho_{H_tA_tB_t}(t)=\sum_{h^t}\ketbra{h^t}{h^t}_{H_t}\otimes\widetilde{\sigma}_{A_tB_t}(t;h^t),
	\end{equation}
	where $H_t\equiv X_1Y_1\dotsb X_tY_t$, $h^t=(x_1,y_1,x_2,y_2,\dotsc,x_t,y_t)$,
	\begin{align}\label{eq-cq_state_LOCC_3}
		\widetilde{\sigma}(t;h^t)&=\left(\mathcal{S}_{A_0\to A_t}^{t;h^t}\otimes\mathcal{T}_{B_0\to B_t}^{t;h^t}\right)(\sigma_{A_0B_0}^0),\\
		\mathcal{S}_{A_0\to A_t}^{t;h^t}&\coloneqq\mathcal{M}_{A_{t-1}\to A_t}^{t-1;h_{t-1}^t,x_t}\circ\dotsb\circ\mathcal{M}_{A_1\to A_2}^{1;h_1^t,x_2}\circ\mathcal{M}_{A_0\to A_1}^{0;x_1},\\[0.2cm]
		\mathcal{T}_{B_0\to B_t}^{t;h^t}&\coloneqq\mathcal{N}_{B_{t-1}\to B_t}^{t-1;h^t}\circ\dotsb\circ\mathcal{N}_{B_1\to B_2}^{1;h_2^t}\circ\mathcal{N}_{B_0\to B_1}^{0;h_1^t},
	\end{align}
	and we have defined $h_j^t\coloneqq (x_1,y_1,x_2,y_2,\dotsc,x_j,y_j)$ for all $j\in\{1,2,\dotsc,t-1\}$.
	
	Now, if Alice and Bob discard the history of their outcomes, then this corresponds to tracing out the classical history register $H_t$, and it results in the state
	\begin{align}
		\sigma_{A_tB_t}(t)&\coloneqq\Tr_{H_t}[\rho_{H_tA_tB_t}(t)]\\
		&=\sum_{h^t}\widetilde{\sigma}_{A_tB_t}(t;h^t)\\
		&=\sum_{h^t}(\mathcal{S}_{A_0\to A_t}^{t;h^t}\otimes\mathcal{T}_{B_0\to B_t}^{t;h^t})(\sigma_{A_0B_0}^0),\label{eq-LOCC_protocol_final_state}
	\end{align}
	which is of the form
	\begin{equation}\label{eq-LOCC_channel}
		\mathcal{L}_{AB\to\hat{A}\hat{B}}(\cdot)\coloneqq\sum_{x\in\SF{X}}(\mathcal{S}_{A\to\hat{A}}^x\otimes\mathcal{T}_{B\to\hat{B}}^x)(\cdot).
	\end{equation}
	Here, $\{\mathcal{S}^x\}_{x\in\SF{X}}$ and $\{\mathcal{T}^x\}_{x\in\SF{X}}$ are completely positive trace non-increasing maps such that the sum $\sum_{x\in\SF{X}}\mathcal{S}^x\otimes\mathcal{T}^x$ is a trace-preserving map.
	
	Note that the sum $\displaystyle\sum_{h^t}\mathcal{S}^{t;h^t}\otimes\mathcal{T}^{t;h^t}$ in \eqref{eq-LOCC_protocol_final_state} is indeed a trace-preserving map, because for every $j\in\{1,2,\dotsc,t-1\}$, history $h^{j-1}$ up to time $j-1$, and linear operator $X_{A_{j-1}B_{j-1}}$,
	\begin{align}
		&\sum_{x_j,y_j}\Tr\!\left[\left(\mathcal{M}_{A_{j-1}\to A_j}^{j-1;h^{j-1},x_j}\otimes\mathcal{N}_{B_{j-1}\to B_j}^{j-1;h^{j-1},x_j,y_j}\right)(X_{A_{j-1}B_{j-1}})\right]\nonumber\\
		&=\sum_{x_j}\Tr\!\left[\left(\mathcal{M}_{A_{j-1}\to A_j}^{j-1;h^{j-1},x_j}\otimes\sum_{y_j}\mathcal{N}_{B_{j-1}\to B_j}^{j-1;h^{j-1},x_j,y_j}\right)(X_{A_{j-1}B_{j-1}})\right]\\
		&=\sum_{x_j}\Tr_{A_j}\!\left[\mathcal{M}_{A_{j-1}\to A_j}^{j-1;h^{j-1},x_j}(\Tr_{B_{j-1}}[X_{A_{j-1}B_{j-1}}])\right]\\
		&=\Tr_{A_j}\!\left[\sum_{x_j}\mathcal{M}_{A_{j-1}\to A_j}^{j-1;h^{j-1},x_j}(\Tr_{B_{j-1}}[X_{A_{j-1}B_{j-1}}])\right]\\
		&=\Tr[X_{A_{j-1}B_{j-1}}],
	\end{align}
	where we have used the fact that $\displaystyle\sum_{y_j}\mathcal{N}_{B_{j-1}\to B_j}^{j-1;h^{j-1},x_j,y_j}$ and $\displaystyle\sum_{x_j}\mathcal{M}_{A_{j-1}\to A_j}^{j-1;h^{j-1},x_j}$ are trace-preserving maps. By applying this recursively at all time steps, it follows that for every linear operator $X_{A_0B_0}$,
	\begin{equation}
		\Tr\!\left[\sum_{h^t}(\mathcal{S}_{A_0\to A_t}^{t;h^t}\otimes\mathcal{T}_{A_0\to A_t}^{t;h^t})(X_{A_0B_0})\right]=\Tr[X_{A_0B_0}].
	\end{equation}
	So we conclude that the sum $\displaystyle\sum_{h^t}\mathcal{S}^{t;h^t}\otimes\mathcal{T}^{t;h^t}$ is a trace-preserving map.
	
	\begin{remark}[LOCC instruments]
		From \eqref{eq-cq_state_LOCC} and \eqref{eq-cq_state_LOCC_3}, the classical-quantum state $\rho_{H_tA_tB_t}(t)$ after $t$ rounds of an LOCC protocol is
		\begin{equation}
			\rho_{H_tA_tB_t}(t)=\sum_{h^t}\ketbra{h^t}{h^t}_{H_t}\otimes\left(\mathcal{S}_{A_0\to A_t}^{t;h^t}\otimes\mathcal{T}_{B_0\to B_t}^{t;h^t}\right)(\sigma_{A_0B_0}^0).
		\end{equation}
		Observe that this state has exactly the form of the output state of a quantum instrument channel. In particular, letting
		\begin{equation}
			\mathcal{L}_{A_0B_0\to A_tB_t}^{t;h^t}\coloneqq\mathcal{S}_{A_0\to A_t}^{t;h^t}\otimes\mathcal{T}_{B_0\to B_t}^{t;h^t},
		\end{equation}
		we see that the state $\rho_{H_tA_tB_t}(t)$ can be regarded as the output state of an \textit{LOCC instrument}, i.e., a finite set $\left\{\mathcal{L}_{AB\to\hat{A}\hat{B}}^{x}\right\}_{x\in\SF{X}}$ of completely positive trace non-increasing LOCC maps such that sum $\displaystyle\sum_{x\in\SF{X}}\mathcal{L}_{AB\to \hat{A}\hat{B}}^{x}$ is a trace-preserving map, and thus an LOCC quantum channel.
	\end{remark}

\toclesslab\section{Examples of elementary link generation}{sec-elem_link_examples}

\toclesslab\subsection{Ground-based transmission}{sec-network_setup_ground_based}

	The most common medium for quantum information transmission for communication purposes is photons traveling either through either free space or fiber-optic cables. These transmission media are modeled well by a bosonic pure-loss/attenuation channel $\mathcal{L}^{\eta}$~\cite{Serafini_book}, where $\eta\in(0,1]$ is the transmittance of the medium, which for fiber-optic or free-space transmission has the form $\eta=\e^{-\frac{L}{L_0}}$~\cite{SveltoBook,KJK_book,KGMS88_book}, where $L$ is the transmission distance and $L_0$ is the attenuation length of the fiber.
	
	Before the $k$ quantum systems corresponding to the source state $\rho^S$ are transmitted through the pure-loss channel, they are each encoded into $d$ bosonic modes with $d\geq 2$. A simple encoding is the following:
	\begin{align}
		\ket{0_d}&\coloneqq\ket{1,0,0,\dotsc,0},\label{eq-d_rail_encoding_0}\\
		\ket{1_d}&\coloneqq\ket{0,1,0,\dotsc,0},\\
		&\vdots\nonumber\\
		\ket{(d-1)_d}&\coloneqq\ket{0,0,0,\dotsc,1},\label{eq-d_rail_encoding}
	\end{align}
	sometimes called the \textit{$d$-rail encoding}. In other words, using $d$ bosonic modes, we form a qudit quantum system by defining the standard basis elements of the associated Hilbert space by the states corresponding to a single photon in each of the $d$ modes. We let
	\begin{equation}
		\ket{\text{vac}}\coloneqq\ket{0,0,\dotsc,0}
	\end{equation}
	denote the vacuum state of the $d$ modes, which is the state containing no photons.
	
	In the context of photonic state transmission, the source state $\rho^S$ is typically of the form $\ketbra{\psi^S}{\psi^S}$, where
	\begin{equation}\label{eq-source_state}
		\ket{\psi^S}=\sqrt{\smash[b]{p_0^S}}\ket{\text{vac}}+\sqrt{\smash[b]{p_1^S}}\ket{\psi_1^S}+\sqrt{\smash[b]{p_2^S}}\ket{\psi_2^S}+\dotsb,
	\end{equation}
	where $\ket{\psi_n^S}$ is a state vector with $n$ photons in total for each of the $k$ parties and the numbers $p_n^S\geq 0$ are probabilities, so that $\sum_{n=0}^{\infty} p_n^S=1$. For example, in the case $k=2$ and $d=2$, the following source state is generated from a parametric down-conversion process (see, e.g., Refs.~[\onlinecite{KB00,KGD+16}]):
	\begin{align}
		\ket{\psi^S}&=\sum_{n=0}^{\infty}\frac{\sqrt{n+1}r^n}{\e^q}\ket{\psi_n^S},\\
		\ket{\psi_n^S}&=\frac{1}{\sqrt{n+1}}\sum_{m=0}^n (-1)^m \ket{n-m,m;m,n-m},
	\end{align}
	where $r$ and $q$ are parameters characterizing the process. One often considers a truncated version of this state as an approximation, so that~\cite{KGD+16}
	\begin{equation}\label{eq-SPDC_state_2}
		\ket{\psi^S}=\sqrt{p_0}\ket{0,0;0,0}+\sqrt{\frac{p_1}{2}}(\ket{1,0;0,1}+\ket{0,1;1,0})+\sqrt{\frac{p_2}{3}}(\ket{2,0;0,2}+\ket{1,1;1,1}+\ket{0,2;2,0}),
	\end{equation}
	where $p_0+p_1+p_2=1$.
	
	Typically, the encoding into bosonic modes is not perfect, which means that a source state of the form \eqref{eq-source_state} is not ideal and that the desired state is given by one of the state vectors $\ket{\psi_j^S}$, and the other terms arise due to the naturally imperfect nature of the source. For example, for the state in \eqref{eq-SPDC_state_2}, the desired bipartite state is the maximally entangled state
	\begin{equation}\label{eq-photonic_source_ideal}
		\ket{\Psi^+}=\frac{1}{\sqrt{2}}(\ket{1,0;0,1}+\ket{0,1;1,0}).
	\end{equation}
	
	Once the source state is prepared, each mode is sent through the pure-loss channel. Letting
	\begin{equation}
		\mathcal{L}^{\eta,(d)}\coloneqq (\mathcal{L}^{\eta})^{\otimes d}
	\end{equation}
	denote the quantum channel that acts on the $d$ modes of each of the $k$ systems, the overall quantum channel through which the source state $\rho^S$ is sent is
	\begin{equation}
		\mathcal{S}^{\vec{\eta},(k;d)}\coloneqq \mathcal{L}^{\eta_1,(d)}\otimes\mathcal{L}^{\eta_2,(d)}\otimes\dotsb\otimes\mathcal{L}^{\eta_k,(d)},
	\end{equation}
	where $\vec{\eta}=(\eta_1,\eta_2,\dotsc,\eta_k)$ and $\eta_j$ is the transmittance of the medium to the $j^{\text{th}}$ node in the edge. The quantum state shared by the $k$ nodes after transmission from the source is then $\rho^{S,\text{out}}=\mathcal{S}^{\vec{\eta},(k;d)}(\rho^S)$.
	
	Now, it is known (see, e.g., Ref.~[\onlinecite{BH14}]) that the action of the bosonic pure-loss channel on any linear operator $\sigma_d$ encoded in $d$ modes according to the encoding in \eqref{eq-d_rail_encoding} is equivalent to the output of an erasure channel~\cite{BDS97,GBP97}. In general, a $d$-dimensional quantum erasure channel $\mathcal{E}_p^{(d)}$, with $p\in[0,1]$, is defined as follows. Consider the vector space $\mathbb{C}^d$ with orthonormal basis elements $\{\ket{0},\ket{1},\dotsc,\ket{d-1}\}$, and the vector space $\mathbb{C}^{d+1}$ with orthonormal basis elements $\{\ket{0},\ket{1},\dotsc,\ket{d-1},\ket{d}\}$. Then, for every linear operator $X\in\Lin(\mathbb{C}^d)$, $\mathcal{E}_p^{(d)}(X)=pX+(1-p)\ketbra{d}{d}$. Note that the output is an element of $\Lin(\mathbb{C}^{d+1})$. In particular, note that the vector $\ket{d}$ is orthogonal to the input vector space $\mathbb{C}^d$.
	
	\begin{lemma}[Pure-loss channel with a $d$-rail encoding~\cite{BH14}]\label{lem-pure_loss_erasure}
		Let $d\geq 2$. For every linear operator $X$ acting on a $d$-dimensional Hilbert space defined by the basis elements in \eqref{eq-d_rail_encoding_0}--\eqref{eq-d_rail_encoding}, we have that
		\begin{align}\label{eq-pure_loss_erasure}
			\mathcal{L}^{\eta,(d)}(X)&=(\mathcal{L}^{\eta})^{\otimes d}(\sigma_d)\\
			&=\eta X+(1-\eta)\Tr[X]\ketbra{\textnormal{vac}}{\textnormal{vac}}.
		\end{align}
	\end{lemma}
	
	\begin{proof}
		To start, the bosonic pure-loss channel has the following Kraus representation~\cite{FH08,SSS11}:
		\begin{equation}\label{eq-pure_loss_Kraus}
			\mathcal{L}^{\eta}(\rho)=\sum_{\ell=0}^{\infty} \frac{(1-\eta)^{\ell}}{\ell!}\sqrt{\eta}^{a^\dagger a} a^k\rho a^{k\dagger}\sqrt{\eta}^{a^\dagger a},
		\end{equation}
		where $a$ and $a^\dagger$ are the annihilation and creation operators of the bosonic mode, which are defined as $a\ket{n}=\sqrt{n}\ket{n-1}$ for all $n\geq 1$ (with $a\ket{0}=0$), and $a^\dagger\ket{n}=\sqrt{n+1}\ket{n+1}$ for all $n\geq 0$.
		
		Now, every linear operator $X$ acting on a $d$-dimensional space that is encoded into $d$ bosonic modes as in \eqref{eq-d_rail_encoding_0}--\eqref{eq-d_rail_encoding} can be written as
		\begin{equation}
			X=\sum_{\ell,\ell'=0}^{d-1}\alpha_{\ell,\ell'}\ketbra{\ell_d}{\ell'_d},
		\end{equation}
		for $\alpha_{\ell,\ell'}\in\mathbb{C}$. Using \eqref{eq-pure_loss_Kraus}, it is straightforward to show that
		\begin{align}
			\mathcal{L}^{\eta}(\ketbra{0}{0})&=\ketbra{0}{0},\\
			\mathcal{L}^{\eta}(\ketbra{0}{1})&=\sqrt{\eta}\ketbra{0}{1},\\
			\mathcal{L}^{\eta}(\ketbra{1}{0})&=\sqrt{\eta}\ketbra{1}{0},\\
			\mathcal{L}^{\eta}(\ketbra{1}{1})&=(1-\eta)\ketbra{0}{0}+\eta\ketbra{1}{1}.
		\end{align}
		Using this, we find that
		\begin{equation}
			(\mathcal{L}^{\eta})^{\otimes d}(\ketbra{\ell_d}{\ell'_d})=\left\{\begin{array}{l l} \eta\ketbra{\ell_d}{\ell_d}+(1-\eta)\ketbra{\text{vac}}{\text{vac}} & \text{if } \ell=\ell',\\ \eta\ketbra{\ell_d}{\ell'_d} & \text{if }\ell\neq\ell'. \end{array}\right.
		\end{equation}
		Therefore,
		\begin{align}
			(\mathcal{L}^{\eta})^{\otimes d}(X)&=\eta\sum_{\ell,\ell'=0}^{d-1}\alpha_{\ell,\ell'}\ketbra{\ell_d}{\ell'_d}+(1-\eta)\left(\sum_{\ell=0}^{d-1}\alpha_{\ell,\ell}\right)\ketbra{\text{vac}}{\text{vac}}\\
			&=\eta X+(1-\eta)\Tr[X]\ketbra{\text{vac}}{\text{vac}},
		\end{align}
		as required.
	\end{proof}
	
	After transmission from the source to the nodes, the heralding procedure typically involves doing measurements at the nodes to check whether all of the photons arrived. In the ideal case the quantum instrument $\{\mathcal{M}^0,\mathcal{M}^1\}$ for the heralding procedure corresponds simply to a measurement in the single-photon subspace defined by \eqref{eq-d_rail_encoding_0}--\eqref{eq-d_rail_encoding}. To be specific, let
	\begin{align}
		\Lambda^1&\coloneqq\Pi^{(d)}\\
		&\coloneqq\ketbra{0_d}{0_d}+\ketbra{1_d}{1_d}+\dotsb+\ketbra{(d-1)_d}{(d-1)_d},\label{eq-photonic_heralding_example_1}\\
		\Lambda^0&\coloneqq \mathbbm{1}_{\SF{H}_d}-\Lambda^0,\label{eq-photonic_heralding_example_2}
	\end{align}
	where $\Pi^{(d)}$ is the projection onto the $d$-dimensional single-photon subspace defined by \eqref{eq-d_rail_encoding_0}--\eqref{eq-d_rail_encoding}, and $\mathbbm{1}_{\SF{H}_d}$ is the identity operator of the full Hilbert space $\SF{H}_d$ of $d$ bosonic modes. Then, letting $\vec{x}\in\{0,1\}^k$ and defining
	\begin{equation}\label{eq-photonic_heralding_example_3}
		\Lambda^{\vec{x}}\coloneqq\Lambda^{x_1}\otimes\Lambda^{x_2}\otimes\dotsb\otimes\Lambda^{x_k},
	\end{equation}
	the maps $\mathcal{M}^0$ and $\mathcal{M}^1$ have the form
	\begin{align}
		\mathcal{M}^1(\cdot)&=\Lambda^{\vec{1}}(\cdot)\Lambda^{\vec{1}},\label{eq-photonic_heralding_1}\\
		\mathcal{M}^0(\cdot)&=\sum_{\substack{\vec{x}\in\{0,1\}^k\\\vec{x}\neq\vec{1}}}\Lambda^{\vec{x}}(\cdot)\Lambda^{\vec{x}}.\label{eq-photonic_heralding_0}
	\end{align}
	These maps correspond to perfect photon-number-resolving detectors. However, the detectors are typically noisy due to dark counts and other imperfections (see, e.g., Refs.~[\onlinecite{KGD+16}]), so that in practice the maps $\mathcal{M}^0$ and $\mathcal{M}^1$ will not have the ideal forms presented in \eqref{eq-photonic_heralding_1} and \eqref{eq-photonic_heralding_0}.
	
	
	Let
	\begin{align}
		\widetilde{\sigma}(0)&\coloneqq(\mathcal{M}^0\circ\mathcal{S})(\rho^S),\\
		\widetilde{\sigma}(1)&\coloneqq(\mathcal{M}^1\circ\mathcal{S})(\rho^S).
	\end{align}
	Then, if the source produces the ideal quantum state, such as the state in \eqref{eq-photonic_source_ideal}, so that $\rho^S=\Psi^+=\ketbra{\Psi^+}{\Psi^+}$, and if the heralding procedure is also ideal, then using \eqref{eq-pure_loss_erasure} we obtain
	\begin{align}
		\widetilde{\sigma}(1)&=\eta_1\eta_2\Psi^+,\\
		\widetilde{\sigma}(0)&=\eta_1(1-\eta_2)\frac{\Pi^{(2)}}{2}\otimes\ketbra{\text{vac}}{\text{vac}}\nonumber\\
		&\qquad+(1-\eta_1)\eta_2\ketbra{\text{vac}}{\text{vac}}\otimes\frac{\Pi^{(2)}}{2}\nonumber\\
		&\qquad+(1-\eta_1)(1-\eta_2)\ketbra{\text{vac}}{\text{vac}}\otimes\ketbra{\text{vac}}{\text{vac}},
	\end{align}
	which means that the transmission-heralding success probability as defined in \eqref{eq-elem_link_success_prob} is simply $p=\Tr[\widetilde{\sigma}(1)]=\eta_1\eta_2$.
	
	\begin{remark}[Multiplexing]\label{rem-multiplexing}
		In practice, in order to increase the transmission-heralding success probability, multiplexing strategies are used. The term ``multiplexing'' here refers to the use of a single transmission channel to send multiple signals simultaneously, with the signals being encoded into distinct (i.e., orthogonal) frequency modes; see, e,g., Ref.~[\onlinecite{GKF+15}]. If $M\geq 1$ distinct frequency modes are used, then the source state being transmitted is $(\rho^S)^{\otimes M}$. If $p$ denotes the probability that any single one of the signals is received and heralded successfully, then the probability that at least one of the $M$ signals is received and heralded successfully is $1-(1-p)^M$.
	\end{remark}

\toclesslab\subsection{Transmission from satellites}{sec-sat_architecture}

	Let us now consider the model of elementary link generation proposed in Ref.~[\onlinecite{KBD+19}], in which the entanglement sources are placed on satellites orbiting the earth. For further information on satellite-based quantum communication, we refer to Ref.~[\onlinecite{SJG+21}] for a review, and we refer to Refs.~[\onlinecite{BMH+13,Vogel2017weather,LKB+19,Vogel2019atmlinks}] for more detailed modeling of the satellite-to-ground quantum channel than what we consider here.
	
	When modeling photon transmission from satellites to ground stations, we must take into account background photons. Here, we analyze the scenario in which a source on board a satellite generates an entangled photon pair and distributes the individual photons to two parties, Alice ($A$) and Bob ($B$), on the ground. We allow the distributed photons to mix with background photons from an uncorrelated thermal source. Also, as before, we use the bosonic encoding defined in \eqref{eq-d_rail_encoding_0}--\eqref{eq-d_rail_encoding}, but we stick to $d=2$, i.e., qubit source states and thus bipartite elementary links. In this scenario, it is common for the two modes to represent the polarization degrees of freedom of the photons, so that
	\begin{align}
		\ket{H}&\equiv\ket{0_2}=\ket{1,0},\\
		\ket{V}&\equiv\ket{1_2}=\ket{0,1}
	\end{align}
	represent the state of one horizontally and vertically polarized photon, respectively.
	
	Let $\overline{n}$ be the average number of background photons. Then, as done in Ref.~[\onlinecite{KBD+19}], we can define an approximate thermal background state as
	\begin{equation}\label{eq-approx_th_state}
		\widetilde{\Theta}^{\overline{n}}\coloneqq(1-\overline{n})\ketbra{\text{vac}}{\text{vac}}+\frac{\overline{n}}{2}\left(\ketbra{H}{H}+\ketbra{V}{V}\right).
	\end{equation}
	The transmission channel from the satellite to the ground stations is then
	\begin{equation}\label{eq-noisy_transmission_channel}
		\mathcal{L}^{\eta_{\text{sg}},\overline{n}}(\rho_{A_1A_2})\coloneqq\Tr_{E_1E_2}\left[\left(U_{A_1E_1}^{\eta_{\text{sg}}}\otimes U^{\eta_{\text{sg}}}_{A_2E_2}\right)\left(\rho_{A_1A_2}\otimes\widetilde{\Theta}_{E_1E_2}^{\overline{n}}\right)\left(U^{\eta_{\text{sg}}}_{A_1E_1}\otimes U^{\eta_{\text{sg}}}_{A_2E_2}\right)^\dagger\right],
	\end{equation}
	where $U^{\eta_{\text{sg}}}$ is the beamsplitter unitary (see, e.g., Ref.~[\onlinecite{Serafini_book}]), and $A_1$ and $A_2$ refer to the horizontal and vertical polarization modes, respectively, of the dual-rail quantum system being transmitted; similarly for $E_1$ and $E_2$. Note that for $\overline{n}=0$, the transformation in \eqref{eq-noisy_transmission_channel} reduces to the one in \eqref{eq-pure_loss_erasure} with $d=2$.
	
	For a source state $\rho_{AB}^S$, with $A\equiv A_1A_2$ and $B\equiv B_1B_2$, the quantum state shared by Alice and Bob after transmission of the state $\rho_{AB}^S$ from the satellite to the ground stations is
	\begin{equation}\label{eq-transmission_noisy_output}
		\rho_{AB}^{S,\text{out}}=\left(\mathcal{L}_A^{\eta_{\text{sg}}^{(1)},\overline{n}_1}\otimes\mathcal{L}_B^{\eta_{\text{sg}}^{(2)},\overline{n}_2}\right)(\rho_{AB}^S),
	\end{equation}	
	where $\eta_{\text{sg}}^{(1)}$ and $\eta_{\text{sg}}^{(2)}$ are the transmittances to the ground stations and $\overline{n}_1$ and $\overline{n}_2$ are the corresponding thermal background noise parameters. In Sec.~\ref{sec-sats}, we look at a specific example of a source state $\rho_{AB}^S$, and thus provide an explicit form for the state $\rho_{AB}^{S,\text{out}}$. We also consider the heralding procedure defined by \eqref{eq-photonic_heralding_example_1}--\eqref{eq-photonic_heralding_0}, and thus provide explicit forms for the states $\sigma^0$ and $\tau^{\varnothing}$ in \eqref{eq-initial_link_state_success} and \eqref{eq-initial_link_state_failure} corresponding to success and failure, respectively, of the heralding procedure.
	
	The transmittance $\eta_{\text{sg}}$ generally depends on atmospheric conditions (such as turbulence and weather conditions) and on orbital parameters (such as altitude and zenith angle)~\cite{Vogel2017weather,Vogel2019atmlinks,LKB+19}. In general, if the satellite is at altitude $h$ and the path length from the satellite to the ground station is $L$, then
	\begin{equation}\label{eq-eta_sg}
		\eta_{\text{sg}}(L,h)=\eta_{\text{fs}}(L)\eta_{\text{atm}}(L,h),
	\end{equation}
	where
	\begin{align}
		\eta_{\text{fs}}(L)&=1-\exp\left(-\frac{2r^2}{w(L)^2}\right),\label{eq-fs_transmittance_2}\\
		w(L)&\coloneqq w_{0}\sqrt{1+\left(\frac{L}{L_{R}}\right)^2},\\
		L_{R}&\coloneqq\pi w_{0}^2\lambda^{-1},
	\end{align}
	and
	\begin{equation}\label{eq-atmospheric_transmittance_zenith_2}
		\eta_{\text{atm}}(L,h)=\left\{\begin{array}{l l} (\eta_{\text{atm}}^{\text{zen}})^{\sec\zeta} & \text{if } -\frac{\pi}{2}<\zeta<\frac{\pi}{2},\\[0.2cm] 0 & \text{if } |\zeta|\geq\frac{\pi}{2}, \end{array}\right.
	\end{equation}
	with $\eta_{\text{atm}}^{\text{zen}}$ the transmittance at zenith ($\zeta=0$). In general, the zenith angle $\zeta$ is given by
	\begin{equation}
		\cos\zeta=\frac{h}{L}-\frac{1}{2}\frac{L^2-h^2}{R_{\oplus}L}
	\end{equation}
	for a circular orbit of altitude $h$, with $R_{\oplus}\approx 6378$~km being the earth's radius. The following parameters thus characterize the total transmittance from satellite to ground: the initial beam waist $w_0$, the receiving aperture radius $r$, the wavelength $\lambda$ of the satellite-to-ground signals, and the atmospheric transmittance $\eta_{\text{atm}}^{\text{zen}}$ at zenith. Throughout the rest of this section, we take~\cite{KBD+19} $r=0.75$~m, $w_0=2.5$~cm, $\lambda=810$~nm, and $\eta_{\text{atm}}^{\text{zen}}=0.5$ at 810~nm~\cite{BMH+13}.
	
	After transmission, we assume a heralding procedure defined by post-selecting on coincident events using (perfect) photon-number-resolving detectors. One can justify this assumption because, in the high-loss and low-noise regimes ($\eta_{\text{sg}}^{(1)},\eta_{\text{sg}}^{(2)},\overline{n}\ll 1$), the probability of four-photon and three-photon occurrences is negligible compared to two-photon events. Therefore, upon successful heralding, the (unnormalized) quantum state shared by Alice and Bob is
	\begin{equation}\label{eq-initial_state_tilde_1_sats}
		\widetilde{\sigma}_{AB}(1)\coloneqq\Pi_{AB}\left(\mathcal{L}_A^{\eta_{\text{sg}}^{(1)},\overline{n}_1}\otimes\mathcal{L}_B^{\eta_{\text{sg}}^{(2)},\overline{n}_2}\right)(\rho_{AB}^S)\Pi_{AB},
	\end{equation}
	where
	\begin{equation}
		\Pi_{AB}\coloneqq(\ketbra{H}{H}_A+\ketbra{V}{V}_A)\otimes(\ketbra{H}{H}_B+\ketbra{V}{V}_B)
	\end{equation}
	is the projection onto the two-photon-coincidence subspace. Note that the projection $\Pi_{AB}$ is exactly the projection $\Lambda^1\otimes\Lambda^1$, with $\Lambda^1$ defined in \eqref{eq-photonic_heralding_example_1}. Then, the transmission-heralding success probability is, as per the definition in \eqref{eq-elem_link_success_prob},
	\begin{align}
		p&\coloneqq \Tr[\widetilde{\sigma}_{AB}(1)]\\
		&=\Tr\!\left[\Pi_{AB}\left(\mathcal{L}_A^{\eta_{\text{sg}}^{(1)},\overline{n}_1}\otimes\mathcal{L}_B^{\eta_{\text{sg}}^{(2)},\overline{n}_2}\right)(\rho_{AB}^S)\right].
	\end{align}
	
	Now, let us take the source state $\rho_{AB}^S$ to be the following:
	\begin{equation}\label{eq-sats_source_state}
		\rho_{AB}^S=f_S\Phi_{AB}^{+}+\left(\frac{1-f_S}{3}\right)(\Phi_{AB}^-+\Psi_{AB}^++\Psi_{AB}^-),
	\end{equation}
	where $f_S\in[0,1]$ and 
	\begin{align}
		\Phi_{AB}^{\pm}&\coloneqq\ketbra{\Phi^{\pm}}{\Phi^{\pm}}_{AB},\\
		\Psi_{AB}^{\pm}&\coloneqq\ketbra{\Psi^{\pm}}{\Psi^{\pm}}_{AB},\\
		\ket{\Phi^{\pm}}_{AB}&\coloneqq\frac{1}{\sqrt{2}}(\ket{H,H}_{AB}\pm\ket{V,V}_{AB}),\\
		\ket{\Psi^{\pm}}_{AB}&\coloneqq\frac{1}{\sqrt{2}}(\ket{H,V}_{AB}\pm\ket{V,H}_{AB}).
	\end{align}
	Using \eqref{eq-sats_source_state}, we obtain an explicit form for the (unnormalized) state $\widetilde{\sigma}_{AB}(1)$ in \eqref{eq-initial_state_tilde_1_sats}.
	
	\begin{proposition}[Quantum state of a satellite-to-ground elementary link~\cite{KBD+19}]\label{prop-noisy_transmission_output_Bell}
		Let $\eta_{\text{sg}}^{(1)},\eta_{\text{sg}}^{(2)},\overline{n}_1,\overline{n}_2\in[0,1]$, and consider the source state $\rho_{AB}^S$ given by \eqref{eq-sats_source_state}. Then, after successful heralding, the (unnormalized) state $\widetilde{\sigma}_{AB}(1)$ given by \eqref{eq-initial_state_tilde_1_sats} is equal to
		\begin{align}
			\widetilde{\sigma}_{AB}(1)&=\Pi_{AB}\left(\mathcal{L}_A^{\eta_{\text{sg}}^{(1)},\overline{n}_1}\otimes\mathcal{L}_B^{\eta_{\text{sg}}^{(2)},\overline{n}_2}\right)(\rho_{AB}^S)\Pi_{AB}\nonumber\\
			&=\frac{1}{2}\left(f_S(a+b)+\left(\frac{1-f_S}{3}\right)(a+2c-b)\right)\Phi_{AB}^+\nonumber\\
			&\quad +\frac{1}{2}\left(f_S(a-b)+\left(\frac{1-f_S}{3}\right)(a+2c+b)\right)\Phi_{AB}^-\nonumber\\
			&\quad +\frac{1}{2}\left(f_Sc+\left(\frac{1-f_S}{3}\right)(2a+c)\right)\Psi_{AB}^+\nonumber\\
			&\quad +\frac{1}{2}\left(f_Sc+\left(\frac{1-f_S}{3}\right)(2a+c)\right)\Psi_{AB}^-,\label{eq-rho_0_sats}
		\end{align}
		where
		\begin{equation}\label{eq-initial_state_params}
			a\coloneqq x_1x_2+y_1y_2,\quad b\coloneqq z_1z_2,\quad c\coloneqq x_1y_2+y_1x_2,
		\end{equation}
		and
		\begin{align}
			x_i&\coloneqq (1-\overline{n}_i)\eta_{\text{sg}}^{(i)}+\frac{\overline{n}_i}{2}\left(\left(1-2\eta_{\text{sg}}^{(i)}\right)^2+\left(\eta_{\text{sg}}^{(i)}\right)^2\right),\\
			y_i&\coloneqq \frac{\overline{n}_i}{2}\left(1-\eta_{\text{sg}}^{(i)}\right)^2,\\
			z_i&\coloneqq (1-\overline{n}_i)\eta_{\text{sg}}^{(i)}-\overline{n}_i\eta_{\text{sg}}^{(i)}\left(1-2\eta_{\text{sg}}^{(i)}\right),
		\end{align}
		for $i\in\{1,2\}$.
	\end{proposition}
	
	From \eqref{eq-rho_0_sats}, we have that the transmission-heralding success probability is given by
	\begin{equation}\label{eq-trans_succ_prob_sats}
		p=\Tr[\widetilde{\sigma}_{AB}(1)]=a+c=(x_1+y_1)(x_2+y_2),
	\end{equation}
	so that the quantum state shared by Alice and Bob conditioned on successful heralding is, as per the definition in \eqref{eq-initial_link_state_success},
	\begin{equation}\label{eq-rho_0_norm_sats}
		\sigma_{AB}^0=\frac{\widetilde{\sigma}_{AB}(1)}{p}.
	\end{equation}

\paragraph{Success probability and fidelity.}

	\begin{figure}
		\centering
		\includegraphics[scale=1]{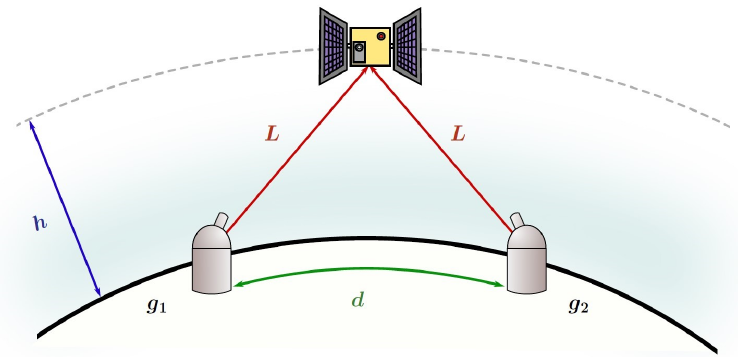}
		\caption{Optical satellite-to-ground transmission~\cite{KBD+19}. Two ground stations $g_1$ and $g_2$ are separated by a distance $d$ with a satellite at altitude $h$ at the midpoint. Both ground stations are the same distance $L$ away from the satellite, so that the total transmittance for two-qubit entanglement transmission (one qubit to each ground station) is $\eta_{\text{sg}}^2$, where $\eta_{\text{sg}}=\eta_{\text{fs}}\eta_{\text{atm}}$, with $\eta_{\text{fs}}$ given by \eqref{eq-fs_transmittance_2} and $\eta_{\text{atm}}$ given by \eqref{eq-atmospheric_transmittance_zenith_2}.}\label{fig-atmosphere_geometry}
	\end{figure}

	Let us now evaluate the quality of entanglement transmission from a satellite to two ground stations. For illustrative purposes, and for simplicity, we focus primarily on the simple scenario depicted in Fig.~\ref{fig-atmosphere_geometry}, in which a satellite passes over the midpoint between two ground stations, although the same analysis can be done even when this is not the case. Since the satellite is an equal distance away from both ground stations, we have $\eta_{\text{sg}}^{(1)}=\eta_{\text{sg}}^{(2)}$. We also let $\overline{n}_1=\overline{n}_2$. This means that $x_1=x_2\equiv x$, $y_1=y_2\equiv y$ and $z_1=z_2\equiv z$, so that
	\begin{equation}\label{eq-sat_transmission_symmetric}
		a=x^2+y^2,\quad b=z^2,\quad c=2xy\quad(\eta_{\text{sg}}^{(1)}=\eta_{\text{sg}}^{(2)}=\eta_{\text{sg}}\text{ and } \overline{n}_1=\overline{n}_2=\overline{n}).
	\end{equation}
	In this scenario, given a distance $d$ between the ground stations and an altitude $h$ for the satellite, by simple geometry the distance $L$ between the satellite and either ground station is given by
	\begin{equation}\label{eq-link_distance_symmetric}
		L=\sqrt{4R_{\oplus}(R_{\oplus}+h)\sin^2\left(\frac{d}{4R_{\oplus}}\right)+h^2},
	\end{equation}
	where $R_{\oplus}$ is the radius of the earth.
	
	\begin{figure}
		\centering
		\includegraphics[scale=1]{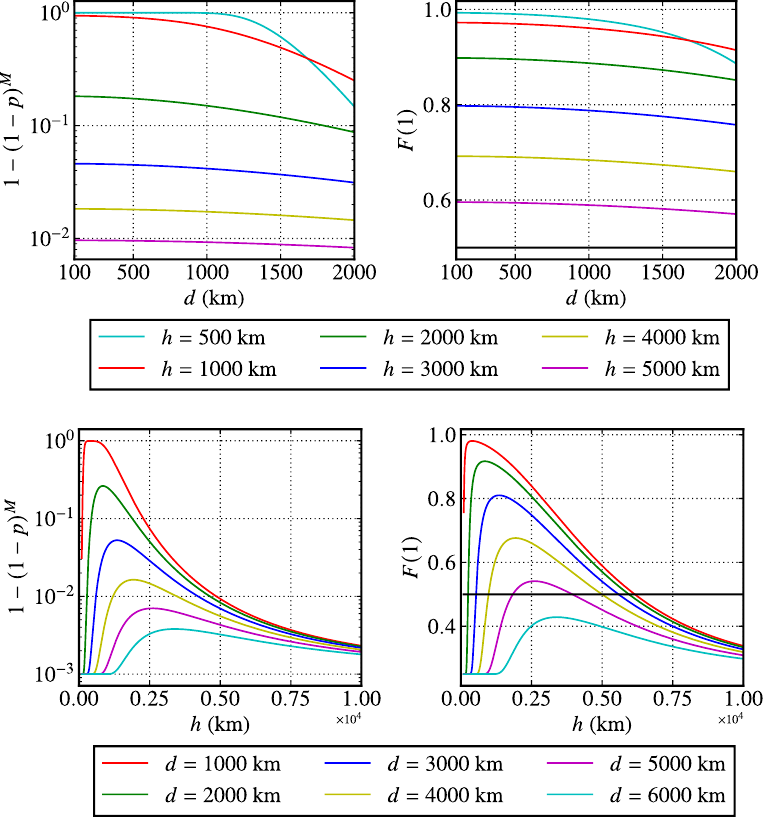}
		\caption{Plots of the transmission-heralding success probability as well as the initial fidelity of the quantum state $\sigma_{AB}^0$ conditioned on successful heralding for the situation depicted in Fig.~\ref{fig-atmosphere_geometry}, in which $\eta_{\text{sg}}^{(1)}=\eta_{\text{sg}}^{(2)}=\eta_{\text{sg}}$ and $\overline{n}_1=\overline{n}_2=\overline{n}$. Indicated is the threshold fidelity of $\frac{1}{2}$ beyond which the state $\sigma_{AB}^0$ is entangled (see Proposition~\ref{prop-initial_state_entanglement}). The success probability is shown in a multiplexing setting with $M=10^5$ (see Remark~\ref{rem-multiplexing}). Also, we have let $\overline{n}=10^{-4}$ and $f_S=1$.}\label{fig-initial_fid_prob}
	\end{figure}

	Now, let us consider the transmission-heralding success probability $p$ in \eqref{eq-trans_succ_prob_sats}. Due to the altitude of the satellites, there typically has to be multiplexing of the signals (see Remark~\ref{rem-multiplexing}) in order to maintain a high probability of both ground stations receiving the entangled state. In Fig.~\ref{fig-initial_fid_prob}, we plot the success probability with multiplexing, which is given by $1-(1-p)^M$, where $M$ is the number of distinct frequency modes used for multiplexing.
	
	We also plot in Fig.~\ref{fig-initial_fid_prob} the fidelity of the initial state, which is given by
	\begin{align}
		F(1)&=\bra{\Phi^+}\sigma_{AB}^0\ket{\Phi^+}=\frac{1}{p}\widetilde{F}(1),\label{eq-initial_fid_sats}\\
		\widetilde{F}(1)&=\bra{\Phi^+}\widetilde{\sigma}_{AB}(1)\ket{\Phi^+}\\
		&=\frac{1}{2}f_S(a+b)+\frac{1}{2}\left(\frac{1-f_S}{3}\right)(a+2c-b),\label{eq-initial_fid_tilde}
	\end{align}
	with $a,b,c$ given by \eqref{eq-initial_state_params} in general and by \eqref{eq-sat_transmission_symmetric} in the special case depicted in Fig.~\ref{fig-atmosphere_geometry}.
	
	The fidelity of $\sigma_{AB}^0$ with respect to $\Phi_{AB}^+$ is related in a simple way to the entanglement of $\sigma_{AB}^0$. In particular, by the partial positive transpose (PPT) criterion~\cite{Peres96,HHH96}, $\sigma_{AB}^0$ is entangled if and only if its fidelity with respect to $\Phi_{AB}^+$ is strictly greater than $\frac{1}{2}$, and this leads to constraints on the loss and noise parameters of the satellite-to-ground transmission. 
	
	\begin{proposition}\label{prop-initial_state_entanglement}
		The quantum state $\sigma_{AB}^0$ after successful satellite-to-ground transmission, as defined in \eqref{eq-rho_0_norm_sats}, is entangled if and only if the fidelity of the source state in \eqref{eq-sats_source_state} satisfies $f_S>\frac{1}{2}$, and
		\begin{equation}\label{eq-initial_state_ent_cond_pf}
			2(f_S-1)a+(4f_S-1)b-(1+2f_S)c>0,
		\end{equation}
		with $a,b,c$ given by \eqref{eq-initial_state_params} in general and by \eqref{eq-sat_transmission_symmetric} in the special case depicted in Fig.~\ref{fig-atmosphere_geometry}.
	\end{proposition}
	
	\begin{proof}
		Observe that the state $\sigma_{AB}^0$ is a Bell-diagonal state of the form
		\begin{equation}\label{eq-initial_state_entanglement_pf1}
			\sigma_{AB}^0=(\alpha+\beta)\Phi_{AB}^++(\alpha-\beta)\Phi_{AB}^-+\gamma\Psi_{AB}^++\gamma\Psi_{AB}^-,
		\end{equation}
		where $\alpha,\beta,\gamma\geq 0$ (when $f_S>\frac{1}{2}$). Indeed, the coefficient of $\Phi_{AB}^+$ in \eqref{eq-rho_0_sats} can be written as
		\begin{equation}
			\frac{1}{2}f_Sa+\frac{1}{2}\left(\frac{1-f_S}{3}\right)(a+2c)+\frac{1}{2}f_Sb-\frac{1}{2}\left(\frac{1-f_S}{3}\right)b,
		\end{equation}
		and the coefficient of $\Phi_{AB}^-$ in \eqref{eq-rho_0_sats} can be written as
		\begin{equation}
			\frac{1}{2}f_Sa+\frac{1}{2}\left(\frac{1-f_S}{3}\right)(a+2c)-\left(\frac{1}{2}f_Sb-\frac{1}{2}\left(\frac{1-f_S}{3}\right)b\right).
		\end{equation}
		We can thus make the following identifications:
		\begin{align}
			\alpha&\equiv\frac{1}{a+c}\left(\frac{1}{2}f_Sa+\frac{1}{2}\left(\frac{1-f_S}{3}\right)(a+2c)\right),\label{eq-rho0_alpha}\\
			\beta&\equiv \frac{1}{a+c}\left(\frac{1}{2}f_Sb-\frac{1}{2}\left(\frac{1-f_S}{3}\right)b\right),\label{eq-rho0_beta}\\
			\gamma&\equiv \frac{1}{2}f_Sc+\frac{1}{2}\left(\frac{1-f_S}{3}\right)(2a+c).\label{eq-rho0_gamma}
		\end{align}
		Now, using the PPT criterion~\cite{Peres96,HHH96}, we have that $\sigma_{AB}^0$ is entangled if and only if $\bra{\Phi^+}\sigma_{AB}^0\ket{\Phi^+}>\frac{1}{2}$. Then, from \eqref{eq-initial_fid_sats}, we have that
		\begin{equation}
			\bra{\Phi^+}\sigma_{AB}^0\ket{\Phi^+}=\frac{1}{2}f_S\frac{a+b}{a+c}+\frac{1}{2}\left(\frac{1-f_S}{3}\right)\frac{a+2c-b}{a+c},
		\end{equation}
		so we require
		\begin{equation}
			\frac{1}{2}f_S\frac{a+b}{a+c}+\frac{1}{2}\left(\frac{1-f_S}{3}\right)\frac{a+2c-b}{a+c}>\frac{1}{2}.
		\end{equation}
		Simplifying this leads to
		\begin{equation}
			2(f_S-1)a+(4f_S-1)b-(1+2f_S)c>0,
		\end{equation}
		as required.
	\end{proof}
	
	\begin{figure}
		\centering
		\includegraphics[scale=1]{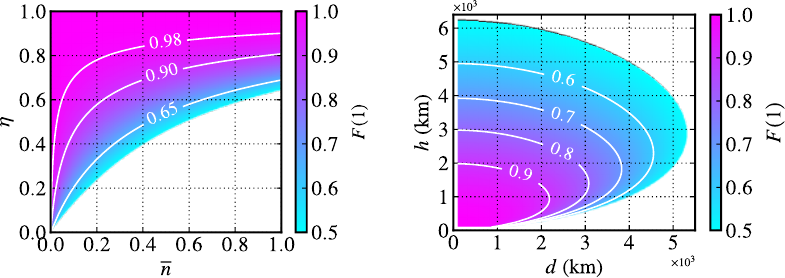}
		\caption{Plots of the entanglement region for the state $\sigma_{AB}^0$ obtained after successful satellite-to-ground transmission for the scenario depicted in Fig.~\ref{fig-atmosphere_geometry}. The regions are defined by the condition $F(1)>\frac{1}{2}$, with $F(1)$ the fidelity of the state $\sigma_{AB}^0$ with the maximally entangled state; see \eqref{eq-initial_fid_sats} and Proposition~\ref{prop-initial_state_entanglement}. For both plots, we assume $f_S=1$. For the right-hand plot, we take $\overline{n}_1=\overline{n}_2=10^{-4}$.}\label{fig-ent_sg}
	\end{figure}

	Now, for the scenario depicted in Fig.~\ref{fig-atmosphere_geometry}, we have that $x_1=x_2=x$, $y_1=y_2=y$, and $z_1=z_2=z$, so that from \eqref{eq-sat_transmission_symmetric} we have $a=x^2+y^2$, $b=z^2$, and $c=2xy$. Substituting this into \eqref{eq-initial_state_ent_cond_pf} leads to $2(f_S-1)(x^2+y^2)+(4f_S-1)z^2-2(1+2f_S)xy>0$ as the condition for $\sigma_{AB}^0$ to be entangled. We plot this condition in Fig.~\ref{fig-ent_sg}. The inequality gives us the colored regions, and the values within the regions are obtained by evaluating the fidelity according to \eqref{eq-initial_fid_sats}.

\paragraph{Key rates for QKD.}
	
	Let us also consider key rates for quantum key distribution (QKD) between Alice and Bob, who are at the ends of the elementary link whose quantum state is $\sigma_{AB}^0$ (conditioned on successful transmission and heralding), as given by \eqref{eq-rho_0_norm_sats}. We consider the BB84, six-state, and device-independent (DI) QKD protocols, and we calculate the secret key rates using known asymptotic secret key rate formulas, which we review (along with other necessary background on QKD) in Appendix~\ref{sec-QKD}.
	
	Recalling from the proof of Proposition~\ref{prop-initial_state_entanglement} that $\sigma_{AB}^0$ is a quantum state of the form
	\begin{equation}
		\sigma_{AB}^0=(\alpha+\beta)\Phi_{AB}^++(\alpha-\beta)\Phi_{AB}^-+\gamma\Psi_{AB}^++\gamma\Psi_{AB}^-,
	\end{equation}
	with $\alpha,\beta,\gamma$ defined in \eqref{eq-rho0_alpha}--\eqref{eq-rho0_gamma}, it is easy to show using \eqref{eq-rhoAB_Qx}--\eqref{eq-rhoAB_Qz} that the quantum bit-error rates for the BB84 and six-state protocols are
	\begin{align}
		Q_{\text{BB84}}^{(d,h)}&=\frac{1}{2}(Q_x+Q_z)=\frac{3}{4}-\frac{1}{2}\beta-\alpha,\label{eq-QBER_BB84_sats}\\
		Q_{\text{6-state}}^{(d,h)}&=\frac{1}{3}(Q_x+Q_y+Q_z)=\frac{2}{3}(1-(\alpha+\beta)).\label{eq-QBER_6state_sats}
	\end{align}
	For the device-independent protocol, we assume that the correlation is such that the quantum bit-error rate is $Q_{\text{DI}}^{(d,h)}=Q_{\text{6-state}}^{(d,h)}$ and $S^{(d,h)}=2\sqrt{2}(1-2Q_{\text{DI}}^{(d,h)})$. Then, assuming that $M$ signals per second are transmitted from the satellite, the secret-key rate (in units of secret key bits per second) is given by $\widetilde{K}=pMK$, where $p=a+c$ is the success probability of elementary link generation and $K$ is the asymptotic secret key rate per copy of the state $\sigma_{AB}^0$, which depends on the protocol under consideration. Using the formulas in Appendix~\ref{sec-QKD}, we obtain
	\begin{align}
		\widetilde{K}_{\text{BB84}}(d,h)&=M(a+c)K_{\text{BB84}}(Q_{\text{BB84}}^{(d,h)})\label{eq-key_rate_BB84_sats}\\
		\widetilde{K}_{\text{6-state}}(d,h)&=M(a+c)K_{\text{6-state}}(Q_{\text{6-state}}^{(d,h)})\label{eq-key_rate_6state_sats}\\
		\widetilde{K}_{\text{DI}}(d,h)&=M(a+c)K_{\text{DI}}(Q_{\text{DI}}^{(d,h)},S^{(d,h)})\label{eq-key_rate_DI_sats}
	\end{align}
	We plot these secret key rates in Fig.~\ref{fig-key_rates_sats}.
	
	\begin{figure}
		\centering
		\includegraphics[scale=1]{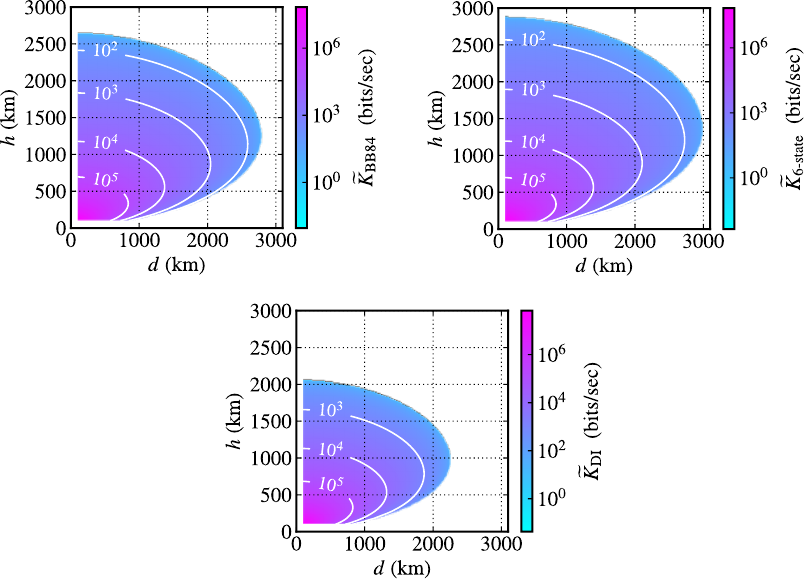}
		\caption{Asymptotic secret key rates for the BB84, six-state, and device-independent (DI) quantum key distribution protocols for the scenario depicted in Fig.~\ref{fig-atmosphere_geometry}. When calculating the error rates in \eqref{eq-QBER_BB84_sats} and \eqref{eq-QBER_6state_sats}, we take $f_S=1$. To calculate the key rates in \eqref{eq-key_rate_BB84_sats}, \eqref{eq-key_rate_6state_sats}, and \eqref{eq-key_rate_DI_sats}, we have taken $M=10^9$.}\label{fig-key_rates_sats}
	\end{figure}
	
	In Fig.~\ref{fig-key_rates_sats}, notice that the region of non-zero secret key rate is largest for the six-state protocol, with the region for the BB84 protocol being smaller and the region for the DI protocol being even smaller. This is due to the fact that the error threshold for the DI protocol is the smallest among the three protocols, with the error threshold for the BB84 protocol slightly larger, and the error threshold for the 6-state protocol the largest.

\toclesslab\section{Policies for satellite-to-ground entanglement distribution}{sec-sats}

	In this section, we present an example of an analysis of elementary links based on the satellite-to-ground transmission model presented in Sec.~\ref{sec-sat_architecture} based on Ref.~[\onlinecite{KBD+19}].
	
	\bigskip

\toclesslab\subsection{Quantum memory model}{sec-quantum_memory_sats}

	Having examined the quantum state immediately after successful transmission and heralding, let us now consider a particular model of decoherence for the quantum memories in which the transmitted qubits are stored. For illustrative purposes, we consider a simple amplitude damping decoherence model for the quantum memories. The amplitude damping channel $\mathcal{A}_{\gamma}$ is a qubit channel, with $\gamma\in[0,1]$, such that~\cite{NC00_book}
	\begin{align}
		\mathcal{A}_{\gamma}(\ketbra{0}{0})&=\ketbra{0}{0},\\
		\mathcal{A}_{\gamma}(\ketbra{0}{1})&=\sqrt{1-\gamma}\ketbra{0}{1},\\
		\mathcal{A}_{\gamma}(\ketbra{1}{0})&=\sqrt{1-\gamma}\ketbra{1}{0},\\
		\mathcal{A}_{\gamma}(\ketbra{1}{1})&=\gamma\ketbra{0}{0}+(1-\gamma)\ketbra{1}{1}.
	\end{align}
	Note that for $\gamma=0$ we recover the noiseless (identity) channel. We can relate $\gamma$ to the coherence time of the quantum memory, which we denote by $t_{\text{coh}}$, as follows [\onlinecite[Sec.~3.4.3]{Preskill20}]:
	\begin{equation}\label{eq-AD_gamma_coh}
		\gamma\coloneqq 1-\e^{-\frac{1}{t_{\text{coh}}}}.
	\end{equation}
	Note that infinite coherence time corresponds to an ideal quantum memory, meaning that the quantum channel is noiseless. Indeed, by relating the noise parameter $\gamma$ to the coherence time as in \eqref{eq-AD_gamma_coh}, we have that $t_{\text{coh}}=\infty\Rightarrow\gamma=0$.
	
	For $m\in\mathbb{N}_0$ applications of the amplitude damping channel, it is straightforward to show that
	\begin{align}
		\mathcal{A}_{\gamma}^{\circ m}(\ketbra{0}{0})&=\ketbra{0}{0},\\
		\mathcal{A}_{\gamma}^{\circ m}(\ketbra{0}{1})&=\sqrt{\lambda_m}\ketbra{0}{1},\\
		\mathcal{A}_{\gamma}^{\circ m}(\ketbra{1}{0})&=\sqrt{\lambda_m}\ketbra{1}{0},\\
		\mathcal{A}_{\gamma}^{\circ m}(\ketbra{1}{1})&=(1-\lambda_m)\ketbra{0}{0}+\lambda_m\ketbra{1}{1},
	\end{align}
	where $\lambda_m\coloneqq\e^{-\frac{m}{t_{\text{coh}}}}=(1-\gamma)^m$. Then, for all $m\in\mathbb{N}_0$,
	\begin{align}
		\sigma_{AB}(m)&\coloneqq(\mathcal{A}_{\gamma}^{\circ m}\otimes\mathcal{A}_{\gamma}^{\circ m})(\sigma_{AB}^0)\\
		&=\left(\alpha\lambda_m^2+\left(\beta-\frac{1}{2}\right)\lambda_m+\frac{1}{2}\right)\Phi_{AB}^++\left(\alpha\lambda_m^2+\left(-\beta-\frac{1}{2}\right)\lambda_m+\frac{1}{2}\right)\Phi_{AB}^-\\
		&\quad +\lambda_m\left(\frac{1}{2}-\alpha\lambda_m\right)\Psi_{AB}^++\lambda_m\left(\frac{1}{2}-\alpha\lambda_m\right)\Psi_{AB}^-\\
		&\quad+\frac{1}{2}(1-\lambda_m)\left(\ketbra{\Phi^+}{\Phi^-}_{AB}+\ketbra{\Phi^-}{\Phi^+}_{AB}\right),\label{eq-elem_link_m_sats}
	\end{align}
	where $\alpha$ and $\beta$ are given by \eqref{eq-rho0_alpha} and \eqref{eq-rho0_beta}, respectively. Note that we have assumed that the memories corresponding to systems $A$ and $B$ have the same coherence time. It follows that
	\begin{equation}
		f(m)=\bra{\Phi^+}(\mathcal{A}_{\gamma}^{\circ m}\otimes\mathcal{A}_{\gamma}^{\circ m})(\sigma_{AB}^0)\ket{\Phi^+}=\alpha\lambda_m^2+\left(\beta-\frac{1}{2}\right)\lambda_m+\frac{1}{2},\label{eq-fid_decay_sats}
	\end{equation}
	for all $m\in\mathbb{N}_0$. Note that $f(m)\leq f(0)$ for all $m\in\mathbb{N}_0$.


\toclesslab\subsection{Memory-cutoff policy}{sec-sats_policies_mem_cutoff}

	Let us now consider the memory-cutoff policy, which we defined in Sec.~\ref{sec-elem_link_policies}. In what follows, we make use of the following definitions for the deterministic decision functions corresponding to the memory-cutoff policy:
	\begin{align}
		d^{t^{\star}}\!(m)&\coloneqq\left\{\begin{array}{l l} 0 & \text{if }m\in\{0,1,\dotsc,t^{\star}-1\},\\ 1 & \text{if }m=-1,t^{\star}, \end{array}\right.\label{eq-mem_cutoff_policy_decision_func}\\
		d^{\infty}(m)&\coloneqq\left\{\begin{array}{l l} 0 & \text{if }m\in\{0,1,2,\dotsc\}, \\ 1 & \text{if }m=-1. \end{array}\right.\label{eq-mem_cutoff_policy_infty}
	\end{align}
	Using \eqref{eq-avg_fid_tilde_tInfty} and \eqref{eq-avg_fid_tInfty}, along with the expression for $f(m)$ in \eqref{eq-fid_decay_sats}, for every cutoff $t^{\star}\in\mathbb{N}_0$ we obtain
	\begin{align}
		\lim_{t\to\infty}\widetilde{F}^{t^{\star}}\!\!(t)&=\frac{p}{1+t^{\star}p}\sum_{m=0}^{t^{\star}}\!\!\left(\alpha\lambda_m^2+\left(\beta-\frac{1}{2}\right)\lambda_m+\frac{1}{2}\right),\\[0.5cm]
		\lim_{t\to\infty}F^{t^{\star}}\!\!(t)&=\frac{1}{t^{\star}+1}\sum_{m=0}^{t^{\star}}\left(\alpha\lambda_m^2+\left(\beta-\frac{1}{2}\right)\lambda_m+\frac{1}{2}\right).
	\end{align}
	Then, using the fact that $\lambda_m=\e^{-\frac{m}{t_{\text{coh}}}}$, it is straightforward to show that
	\begin{align}
		\sum_{m=0}^{t^{\star}}\lambda_m&=\e^{-\frac{t^{\star}}{2t_{\text{coh}}}}\frac{\sinh\left(\frac{1+t^{\star}}{2t_{\text{coh}}}\right)}{\sinh\left(\frac{1}{2t_{\text{coh}}}\right)},\\
		\sum_{m=0}^{t^{\star}}\lambda_m^2&=\e^{-\frac{t^{\star}}{t_{\text{coh}}}}\frac{\sinh\left(\frac{1+t^{\star}}{t_{\text{coh}}}\right)}{\sinh\left(\frac{1}{t_{\text{coh}}}\right)}.
	\end{align}
	Therefore, in the steady-state limit,
	\begin{align}
		\lim_{t\to\infty}\widetilde{F}^{t^{\star}}\!\!(t)&=\frac{\alpha p \e^{-\frac{t^{\star}}{t_{\text{coh}}}}}{1+t^{\star}p}\frac{\sinh\left(\frac{1+t^{\star}}{t_{\text{coh}}}\right)}{\sinh\left(\frac{1}{t_{\text{coh}}}\right)}+\frac{p \e^{-\frac{t^{\star}}{2t_{\text{coh}}}}}{1+t^{\star}p}\left(\beta-\frac{1}{2}\right)\frac{\sinh\left(\frac{1+t^{\star}}{2t_{\text{coh}}}\right)}{\sinh\left(\frac{1}{2t_{\text{coh}}}\right)}+\frac{1}{2}\frac{(t^{\star}+1)p}{1+t^{\star}p},\\[0.5cm]
		\lim_{t\to\infty}F^{t^{\star}}\!\!(t)&=\frac{\alpha\e^{-\frac{t^{\star}}{t_{\text{coh}}}}}{t^{\star}+1}\frac{\sinh\left(\frac{1+t^{\star}}{t_{\text{coh}}}\right)}{\sinh\left(\frac{1}{t_{\text{coh}}}\right)}+\frac{\e^{-\frac{t^{\star}}{2t_{\text{coh}}}}}{t^{\star}+1}\left(\beta-\frac{1}{2}\right)\frac{\sinh\left(\frac{1+t^{\star}}{2t_{\text{coh}}}\right)}{\sinh\left(\frac{1}{2t_{\text{coh}}}\right)}+\frac{1}{2}.
	\end{align}

	For $t^{\star}=\infty$, from \eqref{eq-avg_fid_tilde_sats}, we obtain
	\begin{equation}
		\widetilde{F}^{\infty}(t)=\sum_{m=0}^{t-1}\left(\alpha\lambda_m^2+\left(\beta-\frac{1}{2}\right)\lambda_m+\frac{1}{2}\right)p(1-p)^{t-1-m}
	\end{equation}
	for all $t\geq 1$. Evaluating the sums leads to
	\begin{equation}\label{eq-fid_tilde_tstarInf}
		\widetilde{F}^{\infty}(t)=\frac{\alpha p\e^{\frac{2}{t_{\text{coh}}}}\left(\e^{-\frac{2t}{t_{\text{coh}}}}-(1-p)^t\right)}{1-\e^{\frac{2}{t_{\text{coh}}}}(1-p)}+\left(\beta-\frac{1}{2}\right)\frac{p\e^{\frac{1}{t_{\text{coh}}}}\left(\e^{-\frac{t}{t_{\text{coh}}}}-(1-p)^t\right)}{1-\e^{\frac{1}{t_{\text{coh}}}}(1-p)}+\frac{1}{2}\left(1-(1-p)^t\right).
	\end{equation}
	Then, for all $p\in(0,1]$, we obtain $\lim_{t\to\infty}\widetilde{F}^{\infty}(t)=\frac{1}{2}$.
	
	Let us now focus primarily on the $t^{\star}=\infty$ memory-cutoff policy by considering an example. Consider the situation depicted in Fig.~\ref{fig-atmosphere_geometry}, in which we have two ground stations separated by a distance $d$ and a satellite at altitude $h$ that passes over the midpoint between the two ground stations. Now, given that the ground stations are separated by a distance $d$, it takes time at least $\frac{2d}{c}$ to perform the heralding procedure, as this is the round-trip communication time between the ground stations ($c$ is the speed of light). We thus take the duration of each time step in the decision process for the elementary link to be $\frac{2d}{c}$. If the coherence time of the quantum memories is $x$ seconds, then $t_{\text{coh}}=\frac{xc}{2d}$ time steps. In Fig.~\ref{fig-tstarInf_sats}, we plot the quantities $\widetilde{F}^{\infty}(t)$ (solid lines), $F^{\infty}(t)$ (dashed lines), and $X^{\infty}(t)$ (dotted lines) for the $t^{\star}=\infty$ memory-cutoff policy under this scenario.
	
	\begin{figure}
		\centering
		\includegraphics[scale=1]{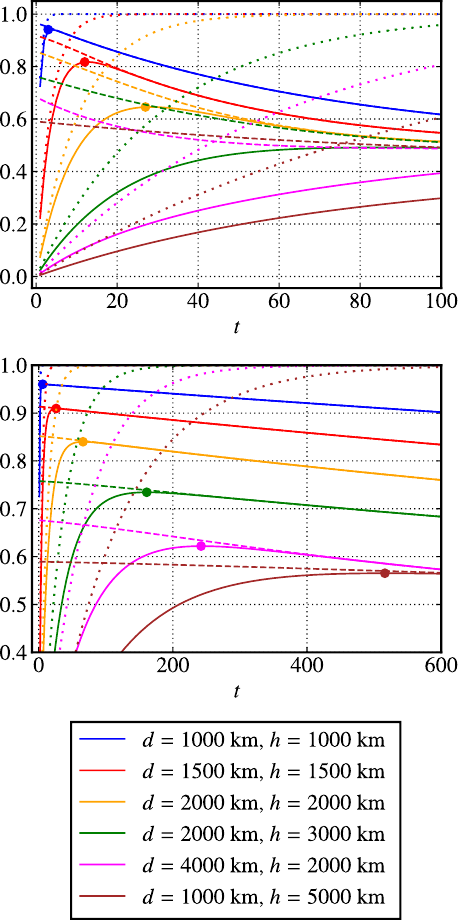}
		\caption{The $t^{\star}=\infty$ memory-cutoff policy for satellite-to-ground elementary link generation for various ground distances $d$ and satellite altitudes $h$, according to the situation depicted in Fig.~\ref{fig-atmosphere_geometry}. The solid lines are $\widetilde{F}^{\infty}(t)$ (as given by \eqref{eq-fid_tilde_tstarInf}), the dashed lines are $F^{\infty}(t)$, and the dotted lines are $X^{\infty}(t)=1-(1-p)^t$ (see \eqref{cor-link_status_Pr1}), where $p=1-(1-(a+c))^M$, with $a$ and $c$ given by \eqref{eq-sat_transmission_symmetric} and $M=10^5$. We let $f_S=1$ be the fidelity of the source, we let $\overline{n}_1=\overline{n}_2=10^{-4}$ be the average number of background photons, and we take the memory coherence times to be 1 s (top) and 60 s (bottom). The dots are placed at the maxima of the curves for $\widetilde{F}^{\infty}(t)$.}\label{fig-tstarInf_sats}
	\end{figure}
	
	In Fig.~\ref{fig-tstarInf_sats}, we can see the trade-off between the quantities $\widetilde{F}$, $F$, and $X$. On the one hand, the fidelity $F^{\infty}(t)$ is always highest at time $t=1$, as we expect, but at this point the probability $X^{\infty}(t)$ that the elementary link is active is simply $p$. Since we want not only a high fidelity for the elementary link but also a high probability that the elementary link is active, by optimizing $\widetilde{F}$ it is possible to achieve a higher elementary link activity probability at the expense of a slightly lower fidelity. Specifically, in Fig.~\ref{fig-tstarInf_sats}, we see that for every choice of $d$ and $h$ there exists a time step $t_{\text{crit}}\geq 1$ at which $\widetilde{F}$ is maximal. At this point, the elementary link activity probability is $1-(1-p)^{t_{\text{crit}}}$, which in many cases is dramatically greater than $p$, while the fidelity $F^{\infty}(t_{\text{crit}})$ is only slightly lower than the fidelity at time $t=1$. Therefore, by waiting until time $t_{\text{crit}}$, it is possible to obtain an elementary link that is almost deterministically active, while incurring only a slight decrease in the fidelity. The time $t_{\text{crit}}$, obtained by optimizing the quantity $\widetilde{F}^{\infty}(t)$ with respect to time $t$ and can be found using the formula in \eqref{eq-fid_tilde_tstarInf}, can be viewed as the optimal time $t$ that should be chosen for the quantum network protocol presented in Fig.~\ref{fig-QDP_protocol}. We refer to Ref.~[\onlinecite{CRDW19}] for an argument similar to the one presented here, except that in Ref.~[\onlinecite{CRDW19}] the time $t_{\text{crit}}$ is obtained by considering a desired value of the fidelity $F^{\infty}(t)$ rather than by optimizing $\widetilde{F}^{\infty}(t)$ with respect to $t$, which is what we do here.

	\bigskip

\toclesslab\subsection{Forward recursion policy}{sec-forward_recursion_policy_sats}

	The forward recursion policy is defined as the time-homogeneous policy such that the action at time $t$ is equal to the one that maximizes the quantity $\widetilde{F}^{\pi}(t+1)$ at the next time step. The corresponding decision function is~\cite{Kha21b}
	\begin{equation}\label{eq-forward_greedy_decision_func}
		d^{\text{FR}}(m)=\left\{\begin{array}{l l} 1 & \text{if }m=-1, \\ 0 & \text{if }m\geq 0\text{ and } f(m+1)>pf(0), \\ 1 & \text{if }m\geq 0\text{ and } f(m+1)\leq pf(0). \end{array}\right.
	\end{equation}
	
	Observe that if $p=1$, then the second condition in \eqref{eq-forward_greedy_decision_func} is always false, because of the fact that $f(m)\leq f(0)$ for all $m\in\mathbb{N}_0$; see \eqref{eq-fid_decay_sats}. Therefore, when $p=1$, we have that $d_t^{\text{FR}}=d_t^0$, i.e., the forward recursion policy is equal to the $t^{\star}=0$ memory-cutoff policy; see \eqref{eq-mem_cutoff_policy_infty}. We now show that the forward recursion policy reduces to a memory-cutoff policy even when $p<1$.
		
	\begin{proposition}\label{prop-forward_recursion_sats}
		Consider satellite-to-ground bipartite elementary link generation with $\overline{n}_1=\overline{n}_2=0$ and $f_S=1$, and let $p\in(0,1)$ be the transmission-heralding success probability, as given by \eqref{eq-trans_succ_prob_sats}. Let $t_{\text{coh}}$ be the coherence time of the quantum memories, as defined in Sec.~\ref{sec-quantum_memory_sats}. Then, for all $t\geq 1$,
		\begin{equation}
			d^{\textnormal{FR}}=\left\{\begin{array}{l l} d^{\infty} & \text{if }p\leq\frac{1}{2}, \\[0.2cm] d^{t^{\star}} & \text{if }p>\frac{1}{2},\end{array}\right.
		\end{equation}
		where
		\begin{equation}\label{eq-forward_greedy_cutoff_sats_spec}
			t^{\star}=\left\lceil-\frac{t_{\textnormal{coh}}}{2}\ln(2p-1)-1\right\rceil.
		\end{equation}
		In other words, if $p\leq\frac{1}{2}$, then the forward recursion policy is equal to the $t^{\star}=\infty$ memory-cutoff policy; if $p>\frac{1}{2}$, then the forward recursion policy is equal to the $t^{\star}$ memory-cutoff policy, with $t^{\star}$ given by \eqref{eq-forward_greedy_cutoff_sats_spec}.
	\end{proposition}
	
	\begin{remark}
		The result of Proposition~\ref{prop-forward_recursion_sats} goes beyond elementary link generation with satellites, because we assumed that $\overline{n}_1=\overline{n}_2=0$ and $f_S=1$. As a result of these assumptions, the result of Proposition~\ref{prop-forward_recursion_sats} applies to every elementary link generation scenario (such as ground-based elementary link generation as described in Sec.~\ref{sec-network_setup_ground_based}) in which the transmission channel is a pure-loss channel, the heralding procedure is described by \eqref{eq-photonic_heralding_example_1}--\eqref{eq-photonic_heralding_0}, the source state is equal to the target state, and the quantum memories are modeled as in Sec.~\ref{sec-quantum_memory_sats}.
	\end{remark}

	\begin{proof}
		For the state $\sigma_{AB}^0$ as given by \eqref{eq-initial_link_state_success}, using \eqref{eq-fid_decay_sats} the second condition in \eqref{eq-forward_greedy_decision_func} translates to 
		\begin{align}
			&\alpha\lambda_{m+1}^2+\left(\beta-\frac{1}{2}\right)\lambda_{m+1}+\frac{1}{2}>p(\alpha+\beta)\label{eq-forward_greedy_condition_sats}\\
			\Rightarrow & p<\frac{\alpha\lambda_{m+1}^2}{\alpha+\beta}+\frac{\left(\beta-\frac{1}{2}\right)\lambda_{m+1}}{\alpha+\beta}+\frac{1}{2(\alpha+\beta)}.\label{eq-forward_greedy_condition_sats2}
		\end{align}
		In the case $\overline{n}_1=\overline{n}_2=0$ and $f_S=1$, we have that $\alpha=\beta=\frac{1}{2}$, so that the inequality in \eqref{eq-forward_greedy_condition_sats2} becomes
		\begin{equation}\label{eq-forward_greedy_condition_sats3}
			p<\frac{1}{2}\left(\e^{-\frac{2(m+1)}{t_{\text{coh}}}}+1\right),
		\end{equation}
		Now, this inequality is satisfied for all $m\in\mathbb{N}_0$ if and only if $p\leq\frac{1}{2}$. In other words, if $p\leq\frac{1}{2}$, then for all possible memory times the action is to wait if the elementary link is currently active, meaning that the decision function in \eqref{eq-forward_greedy_decision_func} becomes
		\begin{equation}
			d^{\text{FR}}(m)=\left\{\begin{array}{l l} 1 & \text{if }m=-1, \\ 0 & \text{if }m\geq 0, \end{array}\right.
		\end{equation}
		which is precisely the decision function $d^{\infty}$ for the $t^{\star}=\infty$ memory-cutoff policy; see \eqref{eq-mem_cutoff_policy_infty}.
		
		For $p\in\left(\frac{1}{2},1\right)$, whether or not the inequality in \eqref{eq-forward_greedy_condition_sats3} is satisfied depends on the memory time $m$. Consider the largest value of $m$ for which the inequality is satisfied, and denote that value by $m_{\max}$. Since the action is to wait, at the next time step the memory value will be $m_{\max}+1$, which by definition will not satisfy the inequality in \eqref{eq-forward_greedy_condition_sats}. This means that, for all memory times strictly less than $m_{\max}+1$, the forward recursion policy dictates that the ``wait'' action should be performed if the elementary link is currently active. As soon as the memory time is equal to $m_{\max}+1$, then the forward recursion policy dictates that the ``request'' action should be performed. This means that $m_{\max}+1$ is a cutoff value. In particular, by rearranging the inequality in \eqref{eq-forward_greedy_condition_sats3}, we obtain
		\begin{equation}
			m<-\frac{t_{\text{coh}}}{2}\ln(2p-1)-1,
		\end{equation}
		which means that
		\begin{equation}
			m_{\max}=\left\lfloor -\frac{t_{\text{coh}}}{2}\ln(2p-1)-1 \right\rfloor,
		\end{equation}
		and
		\begin{equation}
			t^{\star}=1+m_{\max}=\left\lceil -\frac{t_{\text{coh}}}{2}\ln(2p-1)-1 \right\rceil,
		\end{equation}
		as required.
	\end{proof}
	
	Observe that the cutoff in \eqref{eq-forward_greedy_cutoff_sats_spec} is equal to zero for all $p\geq\frac{1}{2}\left(1+\e^{-\frac{2}{t_{\text{coh}}}}\right)$. This means that $p=1$ is not the only transmission-heralding success probability for which the forward recursion policy is equal to the $t^{\star}=0$ memory-cutoff policy. Intuitively, for $\frac{1}{2}\left(1+\e^{-\frac{2}{t_{\text{coh}}}}\right)\leq p\leq 1$, the transmission-heralding success probability is high enough that it is not necessary to store the quantum state in memory---for the purposes of maximizing the value of $\widetilde{F}$, it suffices to request a new quantum state at every time step. At the other extreme, for $0\leq p\leq\frac{1}{2}$, the probability is too low to keep requesting---for the purposes of maximizing the value of $\widetilde{F}$, it is better to keep the quantum state in memory indefinitely.

\toclesslab\subsection{Backward recursion policy}{sec-backward_recursion_policy_sats}

	Finally, to end this section, let us consider the backward recursion policy, which we know to be optimal from Theorem~\ref{thm-opt_policy}. We perform the policy optimization for small times, just as a proof of concept. 
	
	
	
	In Fig.~\ref{fig-opt_pol_sats}, we plot optimal values of $\widetilde{F}^{\pi}(t+1)$ for a single elementary link, except now we plot them as a function of the ground station distance $d$ and the satellite altitude $h$ as per the situation depicted in Fig.~\ref{fig-atmosphere_geometry}. We also plot the elementary link activity probability $X^{\pi}(t+1)$ and the expected fidelities $F^{\pi}(t+1)$ associated with the optimal policies. As before, we assume that $f_S=1$, but unlike before we assume that $\overline{n}_1=\overline{n}_2=10^{-4}$, and we consider multiplexing with $M=10^5$ distinct frequency modes per transmission. We assume a coherence time of 1 s throughout. For small distance-altitude pairs, we find that the optimal value is reached within five time steps. For these cases, it is worth pointing out that the optimal value of $\widetilde{F}^{\pi}(t+1)$ corresponds to an elementary link activity probability $X^{\pi}(t+1)$ of nearly one, while the fidelity (although it drops, as expected) does not drop significantly, meaning that the elementary link can still be useful for performing entanglement distillation of parallel elementary links or for creating virtual links. It is also interesting to point out that for a ground distance separation of $d=2000$~km, the optimal values for satellite altitude $h=1000$~km is higher than for $h=500$~km. This result can be traced back to the top-left panel of Fig.~\ref{fig-initial_fid_prob}, in which we see that the transmission-heralding success probability curves for $h=500$ km and $h=1000$~km cross over at around 1700~km, so that $h=1000$~km has a higher probability than $h=500$~km when $d=2000$~km.
	
	\begin{figure}
		\centering
		\includegraphics[scale=1]{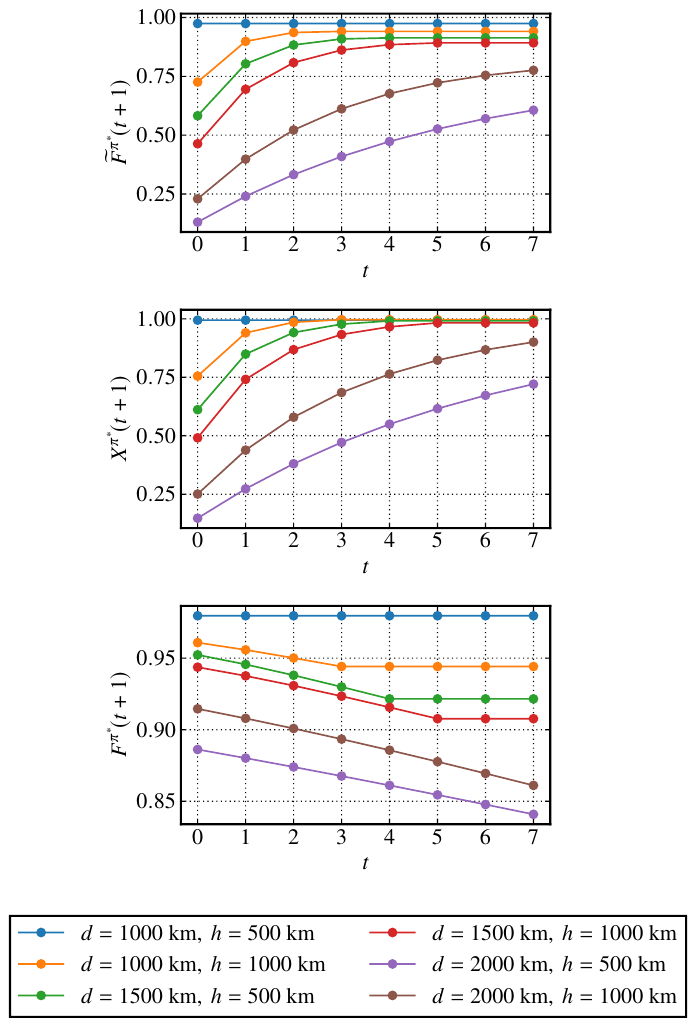}
		\caption{Optimal values of $\widetilde{F}^{\pi}(t+1)$, along with the associated values of $X^{\pi}(t+1)$ and fidelities $F^{\pi}(t+1)$, for a single elementary link distributed by a satellite to two ground stations, according to the symmetric situation depicted in Fig.~\ref{fig-atmosphere_geometry}. We assume that $f_S=1$ and that $\overline{n}_1=\overline{n}_2=10^{-4}$, and we assume that the quantum memories have a coherence time of 1~s. We also assume multiplexing with $M=10^5$ distinct frequency modes per transmission.}\label{fig-opt_pol_sats}
	\end{figure}

	\bigskip

\toclesslab\section{Overview of quantum key distribution}{sec-QKD}

	In this section, we provide a brief overview of quantum key distribution (QKD). We refer to Refs.~[\onlinecite{GRG+02,SBPC+09,Lut14,MyhrThesis,kaur2020,XXQ+20,PAB+19}] for pedagogical introductions and reviews of state-of-the-art QKD research. 
	
	Let us consider the following scenario of so-called entanglement-based QKD. Suppose that Alice and Bob have access to a source that distributes entangled states $\rho_{AB}$ to them, and that their task is to use many copies of this quantum state to distill a secret key. The general strategy of Alice and Bob is to measure their quantum systems. Based on their measurement statistics, they decide whether or not to use their classical measurement data to distill a secret key. The measurement statistics are of the form
	\begin{equation}\label{eq-QKD_correlation}
		p_{AB}(x,y|a,b)\coloneqq\Tr\!\left[\left(\Pi_{A}^{a,x}\otimes \Lambda_B^{b,y}\right)\rho_{AB}\right],\quad x\in\SF{X},\,y\in\SF{Y},\,a\in\SF{A},\,b\in\SF{B},
	\end{equation}
	where $\SF{A}$ and $\SF{B}$ are finite sets of POVMs, such that $\{\Pi_A^{a,x}\}_{x\in\SF{X}}$ is a POVM for Alice's measurement for all $a\in\SF{A}$ and $\{\Lambda_B^{b,y}\}_{y\in\SF{Y}}$ is a POVM for Bob's measurement for all $b\in\SF{B}$. 
	
	\paragraph*{BB84 and six-state protocols.} Two well-known device-dependent protocols that we discuss here are the BB84~\cite{BB84} and six-state~\cite{B98,BG99} protocols. The original formulation of these protocols is as so-called prepare-and-measure protocols, which do not require Alice and Bob to share entanglement. However, these protocols can be viewed from an entanglement-based point of view, in which Alice and Bob possess an entangled state; see Ref.~[\onlinecite{MyhrThesis}] for a discussion on the equivalence of entanglement-based and prepare-and-measure-based protocols, and Ref.~[\onlinecite{TL17}] for a more general discussion of the security of prepare-and-measure-based and entanglement-based QKD protocols. In this device-dependent scenario, we explicitly assume that the state $\rho_{AB}$ is a two-qubit state, and the correlation in \eqref{eq-QKD_correlation} is given by measurement of the qubit Pauli observables $X$, $Z$, and $Y=\I XZ$. In other words, the sets $\SF{A}$ and $\SF{B}$ indicate which observable to be measured, and the sets $\SF{X}$ and $\SF{Y}$ contain the outcomes of the measurements. It can be shown via certain symmetrization procedures that, without loss of generality, $\rho_{AB}$ is a Bell-diagonal state; see Refs.~[\onlinecite{MyhrThesis,WKKG19}] for details. It then suffices to estimate the following three quantities, called \textit{quantum bit-error rates (QBERs)}, in order to characterize the eavesdropper's knowledge:
	\begin{align}
		Q_x&\coloneqq\Tr[(\ketbra{+}{+}_A\otimes\ketbra{-}{-}_B)\rho_{AB}]+\Tr[(\ketbra{-}{-}_A\otimes\ketbra{+}{+}_B)\rho_{AB}]\label{eq-rhoAB_Qx}\\
		&=\frac{1}{2}(1-\Tr[(X\otimes X)\rho_{AB}]),\\
		Q_y&\coloneqq\Tr[(\ketbra{+\I}{+\I}_A\otimes\ketbra{-\I}{-\I}_B)\rho_{AB}]+\Tr[(\ketbra{-\I}{-\I}_A\otimes\ketbra{+\I}{+\I}_B)\rho_{AB}]\label{eq-rhoAB_Qy}\\
		&=\frac{1}{2}(1+\Tr[(Y\otimes Y)\rho_{AB}]),\\
		Q_z&\coloneqq\Tr[(\ketbra{0}{0}_A\otimes\ketbra{1}{1}_B)\rho_{AB}]+\Tr[(\ketbra{1}{1}_A\otimes\ketbra{0}{0}_B)\rho_{AB}]\label{eq-rhoAB_Qz}\\
		&=\frac{1}{2}(1-\Tr[(Z\otimes Z)\rho_{AB}]),
	\end{align}
	where $\ket{\pm}=\frac{1}{\sqrt{2}}(\ket{0}\pm\ket{1})$ and $\ket{\pm\I}=\frac{1}{\sqrt{2}}(\ket{0}\pm\I\ket{1})$. For example, $Q_x$ is simply the probability that Alice and Bob's measurement outcomes disagree when they both measure the observable $X$, and similarly for $Q_y$ and $Q_z$.
	
	A standard figure of merit for QKD protocols is the number of secret key bits obtained per copy of the source state; see, e.g., Ref.~[\onlinecite{WKKG19}] for precise definitions. For the BB84 protocol, the asymptotic secret key rate is~\cite{MY98,LC99,BBB+00,SP00,Mayers01,BBB+06}
	\begin{equation}\label{eq-key_rate_BB84}
		K_{\text{BB84}}(Q)=1-2h_2(Q),
	\end{equation}
	where $Q=\frac{1}{2}(Q_x+Q_z)$ and
	\begin{equation}
		h_2(Q)\coloneqq-Q\log_2(Q)-(1-Q)\log_2(1-Q)
	\end{equation}
	is the binary entropy. For the six-state protocol, the asymptotic secret key rate is~\cite{B98,Lo01}.
	\begin{equation}\label{eq-key_rate_6state}
		K_{\text{6-state}}(Q)=1+\left(1-\frac{3Q}{2}\right)\log_2\left(1-\frac{3Q}{2}\right)+\frac{3Q}{2}\log_2\left(\frac{Q}{2}\right),
	\end{equation}
	where $Q=\frac{1}{3}(Q_x+Q_y+Q_z)$.
	
	\begin{remark}
		The QBERs $Q_x,Q_y,Q_z$ in \eqref{eq-rhoAB_Qx}, \eqref{eq-rhoAB_Qy}, and \eqref{eq-rhoAB_Qz} have a useful interpretation in terms of the fidelity of an arbitrary two-qubit state $\rho_{AB}$ to the maximally entangled state $\Phi_{AB}$. In particular,
		\begin{equation}\label{eq-rhoAB_fid_Bell_QBERs}
			\bra{\Phi}\rho_{AB}\ket{\Phi}=1-\frac{1}{2}(Q_x+Q_y+Q_z)
		\end{equation}
		for every two-qubit state $\rho_{AB}$. It is easy to see this by noting that
		\begin{equation}
			\Phi_{AB}=\frac{1}{4}\left(\mathbbm{1}_A\otimes\mathbbm{1}_B+X_A\otimes X_B-Y_A\otimes Y_B+Z_A\otimes Z_B\right).
		\end{equation}
		Then, using the definitions in \eqref{eq-rhoAB_Qx}, \eqref{eq-rhoAB_Qy}, and \eqref{eq-rhoAB_Qz}, we obtain \eqref{eq-rhoAB_fid_Bell_QBERs}.
	\end{remark}
	
	\paragraph*{Device-independent protocols.} The device-independent protocol that we present here is the one introduced in Refs.~[\onlinecite{AMP06,ABG+07}], and the basic idea behind the protocol comes from the protocol in Ref.~[\onlinecite{Eke91}]. The security of the protocol is based on violation of a Bell inequality, specifically the CHSH inequality~\cite{CHSH69} (see Ref.~[\onlinecite{Scar13}] for a pedagogical introduction). In this protocol, unlike the device-dependent protocols shown above, it is not required to assume that $\rho_{AB}$ is a two-qubit state. However, like the device-dependent protocols considered above, there are symmetrization procedures and other reductions from which it can be argued that $\rho_{AB}$ is a two-qubit Bell-diagonal state without loss of generality; see Refs.~[\onlinecite{AMP06,ABG+07}] for details. The correlation in \eqref{eq-QKD_correlation} is given by measurement of observables $P_A^0,P_A^1,P_A^2$ for system $A$ and observables $Q_B^1,Q_B^2$ for system $B$, and we assume that they all have spectral decompositions of the form
	\begin{align}
		P_A^j&=\Pi_A^{j,0}-\Pi_A^{j,1},\quad j\in\{0,1,2\},\\
		T_B^k&=\Lambda_B^{k,0}-\Lambda_B^{k,1},\quad k\in\{1,2\}.
	\end{align}
	In other words, $\SF{A}=\{0,1,2\}$, $\SF{B}=\{1,2\}$, and $\SF{X}=\SF{Y}=\{0,1\}$.
	
	Two quantities in this case characterize the secret key rate:
	\begin{equation}
		S\coloneqq \Tr\!\left[\left(P_A^1\otimes T_B^1+P_A^1\otimes T_B^2+P_A^2\otimes T_B^1-P_A^2\otimes T_B^2\right)\rho_{AB}\right],
	\end{equation}
	and the quantum bit-error rate (QBER) $Q$, which is defined as
	\begin{equation}
		Q\coloneqq \Tr[(\Pi_A^{0,0}\otimes\Lambda_B^{1,1})\rho_{AB}]+\Tr[(\Pi_A^{0,1}\otimes\Lambda_B^{1,0})\rho_{AB}].
	\end{equation}
	As with the QBERs defined previously, the QBER here is the probability that the outcomes of Alice and Bob disagree when a measurement of $P_A^0$ is performed by Alice and a measurement of $T_B^1$ is performed by Bob. The asymptotic secret key rate is then~\cite{ABG+07,PAB+09}
	\begin{equation}\label{eq-key_rate_DIQKD}
		K_{\text{DI}}(Q,S)=1-h_2(Q)-h_2\left(\frac{1+\sqrt{(S/2)^2-1}}{2}\right).
	\end{equation}

\toclesslab\section{Proof of Theorem~\ref{thm-elem_link_steady_state_dist}}{sec-elem_link_steady_state_dist_pf}

	To prove this, we use \eqref{eq-MDP_func_inf}. First of all, it is straightforward to show that the transition matrix $P^d$ given by the definition in \eqref{eq-transition_matrix_decision} is equal to
	\begin{multline}
		P^d=\left(1-p\overline{\alpha}(-1)\right)\ketbra{-1}{-1}+p\overline{\alpha}(-1)\ketbra{0}{-1}\\+\left(1-p\overline{\alpha}(m^{\star})\right)\ketbra{-1}{m^{\star}}+p\overline{\alpha}(m^{\star})\ketbra{0}{m^{\star}}\\+\sum_{m=0}^{m^{\star}-1}(\alpha(m)\ketbra{m+1}{m}+p\overline{\alpha}(m)\ketbra{0}{m}+(1-p)\overline{\alpha}(m)\ketbra{-1}{m}).
	\end{multline}
	With this, we can verify that the vector $\ket{M(\infty)}_d\coloneqq\sum_{m=-1}^{m^{\star}}s_d(m)\ket{m}$ is a unit-eigenvalue probability vector of $P^d$, i.e., that $P^d\ket{M(\infty)}_d=\ket{M(\infty)}_d$. This is the unique such vector, because the Markov chain defined by the transition matrix $P^d$ is ergodic, which can be straightforwardly verified. Therefore, by ergodicity, the stationary vector $\ket{M(\infty)}$ is unique and $\lim_{t\to\infty} (P^d)^{t-1}=\ketbra{M(\infty)}{\gamma}$; see, e.g., Ref.~[\onlinecite[Theorem~A.2]{Put14_book}]. Therefore, using \eqref{eq-MDP_func_inf}, we obtain the desired result.
	
	\bigskip

\toclesslab\section{Proof of Theorem~\ref{thm-opt_pol_mem_cutoff}}{sec-thm_opt_pol_mem_cutoff_pf}

	The inequality
	\begin{equation}
		\sup_d\lim_{t\to\infty}\widetilde{F}^{(d,d,\dotsc)}(t)\geq \lim_{t\to\infty}\widetilde{F}^{t^{\star}}\!(t)
	\end{equation}
	is certainly true for all $t^{\star}\in\{0,1,\dotsc,m^{\star}\}$, which implies that
	\begin{equation}
		\sup_d\lim_{t\to\infty}\widetilde{F}^{(d,d,\dotsc)}(t)\geq\max_{t^{\star}\in\{0,1,\dotsc,m^{\star}\}}\lim_{t\to\infty}\widetilde{F}^{t^{\star}}\!(t).
	\end{equation}
	
	Now, to prove the opposite inequality, we show that for every decision function $d$ there exists a $t^{\star}\in\{0,1,\dotsc,m^{\star}\}$ such that $\lim_{t\to\infty}\widetilde{F}^{(d,d,\dotsc)}(t)\leq\lim_{t\to\infty}\widetilde{F}^{t^{\star}}\!(t)$. To this end, let $d$ be an arbitrary decision function. We first observe that, from Theorem~\ref{thm-elem_link_steady_state_dist}, the steady-state probability distribution $s_d$ of the memory storage time is such that $s_d(m)\leq s_d(0)$ for all $m\in\{1,2,\dotsc,m^{\star}\}$. On the other hand, for all $t^{\star}\in\mathbb{N}_0$, we have $s_{t^{\star}}(m)=\frac{p}{1+t^{\star}p}$ for all $m\in\{0,1,\dotsc,t^{\star}\}$. We thus need a cutoff value $t^{\star}\in\{0,1,\dotsc,m^{\star}\}$ satisfying $\frac{p}{1+t^{\star}p}\geq s_d(0)$. Rearranging this inequality leads to the condition
	\begin{equation}
		t^{\star}\leq \frac{1}{s_d(0)}-\frac{1}{p}=\frac{N_d}{p(1-\alpha(-1))}-\frac{1}{p}.
	\end{equation}
	Now, from \eqref{eq-elem_link_steady_state_dist_normalization}, we have that
	\begin{equation}
		N_d\leq 1+p(1-\alpha(-1))(1+m^{\star}),
	\end{equation}
	which implies that
	\begin{equation}
		t^{\star}\leq\frac{1}{p(1-\alpha(-1))}+1+m^{\star}-\frac{1}{p}\leq \frac{1}{p}\left(\frac{1}{1-\alpha(-1)}-1\right)+1+m^{\star}.
	\end{equation}
	Now, because $\alpha(-1)$ is a probability, if $\alpha(-1)\in[0,1)$, then $\frac{1}{1-\alpha(-1)}-1\geq 0$, which implies that the right-most quantity in the above inequality is positive and strictly greater than $m^{\star}$. (Note that if $\alpha(-1)=1$, then $\lim_{t\to\infty}\widetilde{F}^{(d,d,\dotsc)}(t)=0$.) Therefore, we can set $t^{\star}=m^{\star}$, resulting in
	\begin{equation}
		\lim_{t\to\infty}\widetilde{F}^{(d,d,\dotsc)}(t)=\sum_{m=0}^{m^{\star}}f(m)s_d(m)\leq\frac{p}{1+m^{\star}p}\sum_{m=0}^{m^{\star}}f(m)=\lim_{t\to\infty}\widetilde{F}^{m^{\star}}\!(t)\leq\max_{t^{\star}\in\{0,1,\dotsc,m^{\star}\}}\lim_{t\to\infty}\widetilde{F}^{t^{\star}}\!(t),
	\end{equation}
	as required. This completes the proof.

\bigskip

\toclesslab\section{Proofs from Sec.~\ref{sec-joining_protocols}}{sec-joining_protocols_pf}

\toclesslab\subsection{Proof of Proposition~\ref{prop-ent_swap_post_fid}}{app-ent_swap_post_fid_pf}

	Let $\rho_{A\vec{R}_1\dotsb\vec{R}_nB}$ be an arbitrary state. Then,
	\begin{multline}
		\bra{\Phi}_{AB}\mathcal{L}_{A\vec{R}_1\dotsb\vec{R}_nB\to AB}^{\ES_n}\left(\rho_{A\vec{R}_1\dotsb\vec{R}_nB}\right)\ket{\Phi}_{AB}\\=\sum_{\vec{x},\vec{z}\in[d]^{\times n}}\left(\bra{\Phi^{a,b}}_{AB}\otimes\bra{\Phi^{z_1,x_1}}_{R_1^1R_1^2}\otimes\dotsb\otimes\bra{\Phi^{z_n,x_n}}_{R_n^1R_n^2}\right)\left(\rho_{A\vec{R}_1\dotsb\vec{R}_nB}\right)\\\left(\ket{\Phi^{a,b}}_{AB}\otimes\ket{\Phi^{z_1,x_1}}_{R_1^1R_1^2}\otimes\dotsb\otimes\ket{\Phi^{z_n,x_n}}_{R_n^1R_n^2}\right),\label{eq-ent_swap_output_fid_pf1}
	\end{multline}
	where
	\begin{equation}
		a\coloneqq z_1+\dotsb+z_n,\quad b\coloneqq x_1+\dotsb+x_n.
	\end{equation}
	Using
	\begin{align}
		\ket{\Phi^{z,x}}&=(Z^zX^x\otimes\mathbbm{1})\ket{\Phi}\\
		&=\frac{1}{\sqrt{d}}\sum_{k=0}^{d-1}\e^{\frac{2\pi\I(k+x)z}{d}}\ket{k+x,k}
	\end{align}
	and
	\begin{equation}\label{eq-comp_to_Bell}
		\ket{j,k}=\frac{1}{\sqrt{d}}\sum_{z,x=0}^{d-1}\e^{\frac{-2\pi\I jz}{d}}\delta_{j,k+x}\ket{\Phi^{z,x}},
	\end{equation}
	we obtain
	\begin{align}
		&\ket{\Phi^{a,b}}_{AB}\otimes\ket{\Phi^{z_1,x_1}}_{\vec{R}_1}\otimes\ket{\Phi^{z_2,x_2}}_{\vec{R}_2}\otimes\dotsb\otimes\ket{\Phi^{z_n,x_n}}_{\vec{R}_n}\nonumber\\
		&\quad=\frac{1}{\sqrt{d^{n+1}}}\sum_{k_0,k_1,\dotsc,k_n=0}^{d-1}\e^{\frac{2\pi\I(k_0+b)a}{d}}\left(\prod_{\ell=1}^n\e^{\frac{2\pi\I(k_{\ell}+x_{\ell})z_{\ell}}{d}}\right)\ket{k_0+b,k_0}_{AB}\ket{k_1+x_1,k_1}_{R_1^1R_1^2}\nonumber\\
		&\qquad\qquad\qquad\qquad\qquad\qquad\qquad\qquad\qquad\qquad\qquad\ket{k_2+x_2,k_2}_{R_2^1R_2^2}\dotsb\ket{k_n+x_n,k_n}_{R_n^1R_n^2}\\
		&\quad=\frac{1}{\sqrt{d^{n+1}}}\sum_{k_0,k_1,\dotsc,k_n=0}^{d-1}\e^{\frac{2\pi\I(k_0+b)a}{d}}\left(\prod_{\ell=1}^n\e^{\frac{2\pi\I(k_{\ell}+x_{\ell})z_{\ell}}{d}}\right)\ket{k_0+b,k_1+x_1}_{AR_1^1}\nonumber\\
		&\qquad\qquad\qquad\qquad\qquad\qquad\qquad\qquad\qquad\qquad\qquad\ket{k_1,k_2+x_2}_{R_1^2R_2^1}\dotsb\ket{k_n,k_0}_{R_n^2B}.
	\end{align}
	Now,
	\begin{align}
		\ket{k_0+b,k_1+x_1}_{AR_1^1}&=\frac{1}{\sqrt{d}}\sum_{z_0',x_0'=0}^{d-1}\e^{\frac{-2\pi\I(k_0+b)z_0'}{d}}\delta_{k_0+b,k_1+x_1+x_0'}\ket{\Phi^{z_0',x_0'}}_{AR_1^1}\\
		\ell\in\{1,\dotsc,n-1\}:\ket{k_{\ell},k_{\ell+1}+x_{\ell+1}}_{R_{\ell}^2R_{\ell+1}^1}&=\frac{1}{\sqrt{d}}\sum_{z_{\ell}',x_{\ell}'=0}^{d-1}\e^{\frac{-2\pi\I k_{\ell}z_{\ell}'}{d}}\delta_{k_{\ell},k_{\ell+1}+x_{\ell+1}+x_{\ell}'}\ket{\Phi^{z_{\ell}',x_{\ell}'}}_{R_{\ell}^2R_{\ell+1}^1},\\
		\ket{k_n,k_0}_{R_n^2B}&=\frac{1}{\sqrt{d}}\sum_{z_n',x_n'=0}^{d-1}\e^{\frac{-2\pi\I k_n z_n'}{d}}\delta_{k_n,k_0+x_n'}\ket{\Phi^{z_n',x_n'}}_{R_n^2B}
	\end{align}
	Therefore,
	\begin{multline}
		\ket{\Phi^{a,b}}_{AB}\otimes\ket{\Phi^{z_1,x_1}}_{\vec{R}_1}\otimes\ket{\Phi^{z_2,x_2}}_{\vec{R}_2}\otimes\dotsb\otimes\ket{\Phi^{z_n,x_n}}_{\vec{R}_n}\\=\frac{1}{d^{n+1}}\sum_{\substack{k_0,\dotsc,k_n=0\\z_0',\dotsc,z_n'=0\\x_0',\dotsc,x_n'=0}}^{d-1} \e^{\frac{2\pi\I(k_0+b)a}{d}}\left(\prod_{\ell=1}^n\e^{\frac{2\pi\I(k_{\ell}+x_{\ell})z_{\ell}}{d}}\right)\e^{\frac{-2\pi\I(k_0+b)z_0'}{d}}\delta_{k_0+b,k_1+x_1+x_0'}\left(\prod_{\ell=1}^{n-1} \e^{\frac{-2\pi\I k_{\ell}z_{\ell}'}{d}}\delta_{k_{\ell},k_{\ell+1}+x_{\ell+1}+x_{\ell}'} \right)\e^{\frac{-2\pi\I k_n z_n'}{d}}\delta_{k_n,k_0+x_n'}\\\ket{\Phi^{z_0',x_0'}}_{AR_1^1}\bigotimes_{\ell=1}^{n-1}\ket{\Phi^{z_{\ell}',x_{\ell}'}}_{R_{\ell}^2R_{\ell+1}^1}\ket{\Phi^{z_n',x_n'}}_{R_n^2B}.
	\end{multline}
	Evaluating the sums with respect to $k_0,\dotsc,k_n$, starting with $k_n$ and proceeding backwards to $k_0$, we obtain
	\begin{multline}
		\ket{\Phi^{a,b}}_{AB}\otimes\ket{\Phi^{z_1,x_1}}_{\vec{R}_1}\otimes\ket{\Phi^{z_2,x_2}}_{\vec{R}_2}\otimes\dotsb\otimes\ket{\Phi^{z_n,x_n}}_{\vec{R}_n}\\=\frac{1}{d^n}\sum_{\substack{z_0',\dotsc,z_n'=0\\x_0',\dotsc,x_n'=0}}^{d-1}\e^{-\frac{2\pi\I}{d}ab}\left(\prod_{\ell=1}^n \e^{\frac{2\pi\I}{d}(x_{\ell}+x_{\ell}'+\dotsb+x_n+x_n')z_{\ell}}\right)\left(\prod_{\ell=1}^n \e^{-\frac{2\pi\I}{d}(x_{\ell}'+x_{\ell+1}+x_{\ell+1}'+\dotsb+x_n+x_n')z_{\ell}'}\right)\e^{\frac{2\pi\I}{d}(z_1'+\dotsb+z_n')b}\\\ket{\Phi^{2a-z_1'-\dotsb-z_n',-x_1'-\dotsb-x_n'}}_{AR_1^1}\bigotimes_{\ell=1}^{n-1}\ket{\Phi^{z_{\ell}',x_{\ell}'}}_{R_{\ell}^2R_{\ell+1}^1}\ket{\Phi^{z_n',x_n'}}_{R_n^2B},
	\end{multline}
	where for the sum with respect to $k_0$ we used the identity
	\begin{equation}\label{eq-exp_sum}
		\sum_{k=0}^{d-1}\e^{\frac{2\pi\I k\alpha}{d}}=d\delta_{\alpha,0},
	\end{equation}
	which holds for all $\alpha\in\{0,1,\dotsc,d-1\}$. Now, observe that
	\begin{equation}\label{eq-Bell_spec}
		\ket{\Phi^{2a-z_1'-\dotsb-z_n',-x_1'-\dotsb-x_n'}}_{AR_1^1}=Z_A^{2a}\ket{\Phi^{-z_1'-\dotsb-z_n',-x_1'-\dotsb-x_n'}}_{AR_1^1}.
	\end{equation}
	Using this, along with the fact that $(Z_A^{z})^{\dagger}Z_A^{z}=\mathbbm{1}$ for all $z\in\{0,1,\dotsc,d-1\}$, and after much simplification and repeated use of \eqref{eq-exp_sum}, we obtain
	\begin{multline}
		\sum_{\vec{x},\vec{z}\in[d]^{\times n}}\left(\bra{\Phi^{a,b}}_{AB}\otimes\bra{\Phi^{z_1,x_1}}_{R_1^1R_1^2}\otimes\dotsb\otimes\bra{\Phi^{z_n,x_n}}_{R_n^1R_n^2}\right)\left(\rho_{A\vec{R}_1\dotsb\vec{R}_nB}\right)\\\left(\ket{\Phi^{a,b}}_{AB}\otimes\ket{\Phi^{z_1,x_1}}_{R_1^1R_1^2}\otimes\dotsb\otimes\ket{\Phi^{z_n,x_n}}_{R_n^1R_n^2}\right)\\=\sum_{\vec{z}',\vec{x}'\in[d]^{\times n}}\left(\bra{\Phi^{-z_1'-\dotsb-z_n',-x_1'-\dotsb-x_n'}}_{AR_1^1}\bigotimes_{\ell=1}^{n-1}\bra{\Phi^{z_{\ell}',x_{\ell}'}}_{R_{\ell}^2R_{\ell+1}^1}\bra{\Phi^{z_n',x_n'}}_{R_n^2B}\right)\left(\rho_{A\vec{R}_1\dotsb\vec{R}_nB}\right)\\\left(\ket{\Phi^{-z_1'-\dotsb-z_n',-x_1'-\dotsb-x_n'}}_{AR_1^1}\bigotimes_{\ell=1}^{n-1}\ket{\Phi^{z_{\ell}',x_{\ell}'}}_{R_{\ell}^2R_{\ell+1}^1}\ket{\Phi^{z_n',x_n'}}_{R_n^2B}\right),
	\end{multline}
	which leads to the desired result.

	\bigskip

\toclesslab\subsection{Proof of Proposition~\ref{prop-GHZ_ent_swap_post_fid}}{app-GHZ_ent_swap_post_fid_pf}

	Let $\rho_{A\vec{R}_1\dotsb\vec{R}_nB}$ be an arbitrary state. We then have
	\begin{multline}\label{eq-ent_swap_GHZ_post_fid_pf1}
		\bra{\GHZ_{n+2}}\mathcal{L}_{A\vec{R}_1\dotsb\vec{R}_nB\to AR_1^1\dotsb R_n^1B}^{\GHZ;n}\left(\rho_{A\vec{R}_1\dotsb\vec{R}_nB}\right)\ket{\GHZ_{n+2}}\\=\frac{1}{2}\sum_{x,x'=0}^1\sum_{\vec{x}\in\{0,1\}^n}\bra{x,x,\dotsc,x}L^{x_n}_n\dotsb L^{x_2}_2 L^{x_1}_1\left(\rho_{A\vec{R}_1\dotsb\vec{R}_nB}\right)L^{x_1\dagger}_{1}L^{x_2\dagger}_{2}\dotsb L^{x_n\dagger}_{n}\ket{x',x',\dotsc,x'},
	\end{multline}
	where
	\begin{align}
		L^{x_j}_{j}&\coloneqq\bra{x_j}_{R_j^2}\text{CNOT}_{\vec{R}_j}X_{R_{j+1}^1}^{x_j}\\
		&=\bra{x_j}_{R_j^2}\left(\sum_{x'=0}^1\ketbra{x'}{x'}_{R_j^1}\otimes X_{R_j^2}^{x'}\right)(\mathbbm{1}_{\vec{R}_j}\otimes X_{R_{j+1}^1}^{x_j})\\
		&=\sum_{x'=0}^1\ketbra{x'}{x'}_{R_j^1}\otimes\bra{x_j+x'}_{R_j^2}\otimes X_{R_{j+1}^1}^{x_j}.
	\end{align}
	Then,
	\begin{multline}
		L^{x_n}_{n}\dotsb L^{x_2}_2 K^{x_1}_1=\sum_{x_1',\dotsc,x_n'=0}^1\ketbra{x_1',\dotsc,x_n'}{x_1',x_2'+x_1,x_3'+x_2,\dotsc,x_n'+x_{n-1}}_{R_1^1R_2^1\dotsb R_n^1}\\\otimes\bra{x_1+x_1',x_2+x_2',\dotsc,x_n+x_n'}_{R_1^2R_2^2\dotsb R_n^2}\otimes X_B^{x_n},
	\end{multline}
	so that, using \eqref{eq-comp_to_Bell} with $d=2$,
	\begin{align}
		&\bra{x,x,\dotsc,x}_{AR_1^1R_2^1\dotsb R_n^1B}L^{x_n}_n\dotsb L^{x_2}_2 L^{x_1}_1\nonumber\\
		&\quad=\bra{x}_A\bra{x,x+x_1,x+x_2,\dotsc,x+x_{n-1}}_{R_1^1R_2^1\dotsb R_n^1}\bra{x_1+x,x_2+x,\dotsc,x_n+x}_{R_1^2R_2^2\dotsb R_n^2}\bra{x+x_n}_B\\
		&\quad=\bra{x,x}_{AR_1^1}\bra{x+x_1,x+x_1}_{R_1^2R_2^1}\bra{x+x_2,x+x_2}_{R_2^2R_3^1}\dotsb\bra{x+x_n,x+x_n}_{R_n^2B}\\
		&\quad=\frac{1}{\sqrt{2^{n+1}}}\sum_{\vec{z}\in\{0,1\}^n}^1(-1)^{z_1x}(-1)^{z_2(x+x_1)}\dotsb(-1)^{z_{n+1}(x+x_n)}\bra{\Phi^{z_1,0}}_{AR_1^1}\bra{\Phi^{z_2,0}}_{R_1^2R_2^1}\dotsb\bra{\Phi^{z_{n+1},0}}_{R_n^2B}.
	\end{align}
	We substitute this into \eqref{eq-ent_swap_GHZ_post_fid_pf1}, simplify, and then make use of the following identity:
	\begin{equation}\label{eq-bit_string_sum_spec}
		\sum_{\vec{\gamma}\in\{0,1\}^n}(-1)^{\vec{\gamma}^{\t}\vec{x}}=2^n\delta_{\vec{x},\vec{0}}.
	\end{equation}
	This leads to
	\begin{multline}
		\bra{\GHZ_{n+2}}\mathcal{L}_{A\vec{R}_1\dotsb\vec{R}_nB\to AR_1^1\dotsb R_n^1B}^{\GHZ;n}\left(\rho_{A\vec{R}_1\dotsb\vec{R}_nB}\right)\ket{\GHZ_{n+2}}\\=\sum_{z_2,\dotsc,z_{n+1}=0}^1\bra{\Phi^{z_2+\dotsb+z_{n+1},0}}\bra{\Phi^{z_2,0}}\dotsb\bra{\Phi^{z_{n+1},0}}\left(\rho_{A\vec{R}_1\dotsb\vec{R}_nB}\right)\ket{\Phi^{z_2+\dotsb+z_{n+1},0}}\ket{\Phi^{z_2,0}}\dotsb\ket{\Phi^{z_{n+1},0}}.
	\end{multline}
	This holds for every state $\rho_{A\vec{R}_1\dotsb\vec{R}_nB}$, so it holds for the tensor product state in the statement of the proposition, thus completing the proof.

\bigskip

\toclesslab\subsection{Proof of Proposition~\ref{prop-graph_state_dist_post_fid}}{app-graph_state_dist_post_fid_pf}

	Let $\rho_{A_1^nB_1^n}$ be an arbitrary $2n$-qubit state. Then, by definition of the channel $\mathcal{L}^{(G)}$, we have that
	\begin{equation}
		\bra{G}\mathcal{L}^{(G)}\left(\rho_{A_1^nR_1^n}\right)\ket{G}=\sum_{\vec{\gamma}\in\{0,1\}^n}\left(\bra{G^{\vec{\gamma}}}_{A_1^n}\otimes\bra{G^{\vec{\gamma}}}_{R_1^n}\right)\left(\rho_{A_1^nR_1^n}\right)\left(\ket{G^{\vec{\gamma}}}_{A_1^n}\otimes\ket{G^{\vec{\gamma}}}_{R_1^n}\right),
	\end{equation}
	where we recall the definition of $\ket{G^{\vec{\gamma}}}$ in \eqref{eq-graph_state_x}. Now,
	\begin{equation}
		\ket{G^{\vec{\gamma}}}_{A_1^n}\otimes\ket{G^{\vec{\gamma}}}_{R_1^n}=\frac{1}{2^n}\sum_{\vec{\alpha},\vec{\beta}\in\{0,1\}^n}(-1)^{\gamma_1(\alpha_1+\beta_1)+\dotsb\gamma_n(\alpha_n+\beta_n)}(-1)^{\frac{1}{2}\vec{\alpha}^{\t}A(G)\vec{\alpha}+\frac{1}{2}\vec{\beta}^{\t}A(G)\vec{\beta}}\ket{\vec{\alpha}}_{A_1^n}\otimes\ket{\vec{\beta}}_{R_1^n},
	\end{equation}
	and, for all $\vec{\alpha},\vec{\beta}\in\{0,1\}^n$,
	\begin{equation}
		\ket{\vec{\alpha}}_{A_1^n}\otimes\ket{\vec{\beta}}_{R_1^n}=\frac{1}{\sqrt{2^n}}\sum_{\vec{x},\vec{z}\in\{0,1\}^n}(-1)^{\alpha_1z_1+\dotsb+\alpha_nz_n}\delta_{\beta_1,\alpha_1+x_1}\dotsb\delta_{\beta_n,\alpha_n+x_n}\ket{\Phi^{z_1,x_1}}_{A_1R_1}\otimes\dotsb\otimes\ket{\Phi^{z_n,x_n}}_{A_nR^n},
	\end{equation}
	where we have used \eqref{eq-comp_to_Bell}. Then,
	\begin{multline}
		\ket{G^{\vec{\gamma}}}_{A_1^n}\otimes\ket{G^{\vec{\gamma}}}_{R_1^n}\\=\frac{1}{(2^n)^{\frac{3}{2}}}\sum_{\vec{\alpha},\vec{x},\vec{z}\in\{0,1\}^n}(-1)^{\vec{\gamma}^{\t}\vec{x}+\vec{\alpha}^{\t}\vec{z}}(-1)^{\frac{1}{2}\vec{\alpha}^{\t}A(G)\vec{\alpha}+\frac{1}{2}(\vec{\alpha}+\vec{x})^{\t}A(G)(\vec{\alpha}+\vec{x})}\ket{\Phi^{x_1,z_1}}_{A_1R_1}\otimes\dotsb\otimes\ket{\Phi^{z_n,x_n}}_{A_nR_n}.
	\end{multline}
	Now, because $A(G)$ is a symmetric matrix, we have that $\vec{\alpha}^{\t}A(G)\vec{x}=\vec{x}^{\t}A(G)\vec{\alpha}$. We thus obtain
	\begin{equation}
		(-1)^{\frac{1}{2}\vec{\alpha}^{\t}A(G)\vec{\alpha}+\frac{1}{2}(\vec{\alpha}+\vec{x})^{\t}A(G)(\vec{\alpha}+\vec{x})}=(-1)^{\vec{\alpha}^{\t}A(G)\vec{x}+\frac{1}{2}\vec{x}^{\t}A(G)\vec{x}},
	\end{equation}
	so that
	\begin{equation}
		\ket{G^{\vec{\gamma}}}_{A_1^n}\otimes\ket{G^{\vec{\gamma}}}_{R_1^n}=\frac{1}{(2^n)^{\frac{3}{2}}}\sum_{\vec{\alpha},\vec{x},\vec{z}\in\{0,1\}^n}(-1)^{\vec{\gamma}^{\t}\vec{x}+\vec{\alpha}^{\t}\vec{z}}(-1)^{\frac{1}{2}\vec{x}^{\t}A(G)\vec{x}+\vec{\alpha}^{\t}A(G)\vec{x}}\ket{\Phi^{z_1,x_1}}_{A_1R_1}\otimes\dotsb\otimes\ket{\Phi^{z_n,x_n}}_{A_nR_n}.
	\end{equation}
	Therefore, using \eqref{eq-bit_string_sum_spec}, we find that
	\begin{multline}
		\sum_{\vec{\gamma}\in\{0,1\}^n}\left(\bra{G^{\vec{\gamma}}}_{A_1^n}\otimes\bra{G^{\vec{\gamma}}}_{R_1^n}\right)\left(\rho_{A_1^nR_1^n}\right)\left(\ket{G^{\vec{\gamma}}}_{A_1^n}\otimes\ket{G^{\vec{\gamma}}}_{R_1^n}\right)\\=\frac{1}{(2^n)^2}\sum_{\vec{\alpha},\vec{\alpha}',\vec{z},\vec{z}',\vec{x}\in\{0,1\}^n}(-1)^{\vec{\alpha}^{\t}(A(G)\vec{x}+\vec{z})+\vec{\alpha}^{\prime\t}(A(G)\vec{x}+\vec{z}')}\left(\bra{\Phi^{z_1,x_1}}_{A_1R_1}\otimes\dotsb\otimes\bra{\Phi^{z_n,x_n}}_{A_nR_n}\right)\left(\rho_{A_1^nR_1^n}\right)\\\left(\ket{\Phi^{z_1',x_1'}}_{A_1R_1}\otimes\dotsb\otimes\ket{\Phi^{z_n',x_n}}_{A_nR_n}\right).
	\end{multline}
	Using \eqref{eq-bit_string_sum_spec} two more times in the summation with respect to $\vec{\alpha}$ and $\vec{\alpha}'$ finally leads to 
	\begin{multline}
		\sum_{\vec{\gamma}\in\{0,1\}^n}\left(\bra{G^{\vec{\gamma}}}_{A_1^n}\otimes\bra{G^{\vec{\gamma}}}_{R_1^n}\right)\left(\rho_{A_1^nR_1^n}\right)\left(\ket{G^{\vec{\gamma}}}_{A_1^n}\otimes\ket{G^{\vec{\gamma}}}_{R_1^n}\right)\\=\sum_{\vec{x}\in\{0,1\}^n}\left(\bra{\Phi^{z_1,x_1}}_{A_1R_1}\otimes\dotsb\otimes\bra{\Phi^{z_n,x_n}}_{A_nR_n}\right)\left(\rho_{A_1^nR_1^n}\right)\left(\ket{\Phi^{z_1,x_1}}_{A_1R_1}\otimes\dotsb\otimes\ket{\Phi^{z_n,x_n}}_{A_nR_n}\right),
	\end{multline}
	where $\vec{z}=A(G)\vec{x}$. Since this holds for every state $\rho_{A_1^nR_1^n}$, it holds for the tensor product state in the statement of the proposition, which completes the proof.
	
	\bigskip
	
\toclesslab\section{Proof of Eq.~\eqref{eq-exp_waiting_time_tInfty}}{sec-exp_waiting_time_tInfty_pf}
	
	By definition,
	\begin{equation}
		\Pr[W_{E'}(t_{\text{req}})=t]_{\infty}=\Pr[X_{E'}(t_{\text{req}}+1)=0,\dotsc,X_{E'}(t_{\text{req}}+t)=1]_{\infty}.
	\end{equation}
	Note that
	\begin{equation}\label{eq-waiting_time_tstarInf_pf1}
		\Pr[W_{E'}(t_{\text{req}})=1]_{\infty}=\Pr[X_{E'}(t_{\text{req}}+1)=1]_{\infty}=(1-(1-p)^{t_{\text{req}}+1})^M=p_{t_{\text{req}}+1}^M,
	\end{equation}
	which holds because all of the elementary links are generated independently and because they all have the same success probability.
	
	Now, for $t\geq 2$, our first goal is to prove that
	\begin{equation}\label{eq-waiting_time_tstarInf_pf4}
		\Pr[W_{E'}(t_{\text{req}})=t]_{\infty}=(1-(1-p_{t_{\text{req}}+1})(1-p)^{t-1})^M-(1-(1-p_{t_{\text{req}}+1})(1-p)^{t-2})^M.
	\end{equation}
	In order to prove this, let us for the moment take $t_{\text{req}}=0$. Then, $X_{E'}(1)=0$ means that at least one of the $M$ elementary links is not active in the first time step, and the same for all subsequent time steps except for the $t^{\text{th}}$ time step, in which all of the $M$ elementary links are active. Then, because $t^{\star}=\infty$, the links that are active in the first time step always remain active. This means that we can evaluate $\Pr[W_{E'}(0)=t]_{\infty}$ by counting the number of elementary links that are inactive at each time step. For example, for $t=2$, we obtain
	\begin{align}
		\Pr[X_{E'}(1)=0,X_{E'}(2)=1]_{\infty}&=\sum_{k_1=1}^M\binom{M}{k_1}\underbrace{(1-p)^{k_1}}_{\substack{k_1\text{inactive links}\\\text{in the first}\\\text{time step}}} \underbrace{p^{M-k_1}}_{\substack{M-k_1\text{ active}\\\text{links in the}\\\text{first time step}}} \underbrace{p^{k_1}}_{\substack{\text{remaining }k_1\\\text{ inactive links}\\\text{succeed in the}\\\text{second time step}}}\\
		&=p^M\sum_{k_1=1}^M\binom{M}{k_1}(1-p)^{k_1}.
	\end{align}
	Similarly, for $t=3$, we find that
	\begin{align}
		&\Pr[X_{E'}(1)=0,X_{E'}(2)=0,X_{E'}(3)=1]_{\infty}\nonumber\\
		&\qquad\qquad\qquad\qquad=\sum_{k_1=1}^M\binom{M}{k_1}(1-p)^{k_1}p^{M-k_1}\sum_{k_2=1}^{k_1}\binom{k_1}{k_2}(1-p)^{k_2}p^{k_1-k_2}p^{k_2}\\
		&\qquad\qquad\qquad\qquad=p^M\sum_{k_1=1}^M\sum_{k_2=1}^{k_1}\binom{M}{k_1}\binom{k_1}{k_2}(1-p)^{k_1}(1-p)^{k_2}.
	\end{align}
	In general, then, for all $t\geq 2$,
	\begin{align}
		&\Pr[W_{E'}(0)=t]_{\infty}\nonumber=\Pr[X_{E'}(1)=0,\dotsc,X_{E'}(t)=1]_{\infty}\\
		&\quad=p^M\sum_{k_1=1}^M\sum_{k_2=1}^{k_1}\sum_{k_3=1}^{k_2}\dotsb\sum_{k_{t-1}=1}^{k_{t-2}}\binom{M}{k_1}\binom{k_1}{k_2}\binom{k_2}{k_3}\dotsb\binom{k_{t-2}}{k_{t-1}}(1-p)^{k_1}(1-p)^{k_2}\dotsb(1-p)^{k_{t-1}}\\
		&\quad=\sum_{k_1=1}^M\binom{M}{k_1}(1-p)^{k_1}p^{M-k_1}\underbrace{p^{k_1}\sum_{k_2=1}^{k_1}\sum_{k_3=1}^{k_2}\dotsb\sum_{k_{t-1}=1}^{k_{t-2}}\binom{k_1}{k_2}\binom{k_2}{k_3}\dotsb\binom{k_{t-2}}{k_{t-1}}(1-p)^{k_2}\dotsb(1-p)^{k_{t-1}}}_{\Pr[W_{k_1}^{(\infty)}(0)=t-1]}\\
		&\quad=\sum_{k_1=1}^M \binom{M}{k_1}(1-p)^{k_1}p^{M-k_1}\Pr[W_{k_1}(0)=t-1]_{\infty} \label{eq-waiting_time_tstarInf_pf3}
	\end{align}
	Using this, we can immediately prove the following result by induction on $t$:
	\begin{equation}\label{eq-waiting_time_tstarInf_pf2}
		\Pr[W_{E'}(0)=t]_{\infty}=(1-(1-p)^t)^M-(1-(1-p)^{t-1})^M.
	\end{equation}
	Indeed, from \eqref{eq-waiting_time_tstarInf_pf1}, we immediately have that this result holds for $t=1$. Similarly, using the fact that
	\begin{equation}
		\sum_{k_1=1}^M\binom{M}{k_1}(1-p)^{k_1}=-1+(2-p)^M=\frac{1}{p^M}\left((1-(1-p)^2)^M-(1-(1-p))^M\right),
	\end{equation}
	we see that \eqref{eq-waiting_time_tstarInf_pf2} holds for $t=2$ as well. Now, assuming that \eqref{eq-waiting_time_tstarInf_pf2} holds for all $t\geq 2$, using \eqref{eq-waiting_time_tstarInf_pf3} we find that
	\begin{align}
		\Pr[W_{E'}(0)=t+1]_{\infty}&=\sum_{k_1=1}^M \binom{M}{k_1}(1-p)^{k_1}p^{M-k_1}\Pr[W_{k_1}(0)=t]_{\infty}\\
		&=\sum_{k_1=1}^M\binom{M}{k_1}(1-p)^{k_1}p^{M-k_1}\left((1-(1-p)^t)^{k_1}-(1-(1-p)^{t-1})^{k_1}\right)\\
		&=(1-(1-p)^{t+1})^M-(1-(1-p)^t)^M,
	\end{align}
	as required. Therefore, \eqref{eq-waiting_time_tstarInf_pf2} holds for all $t\geq 1$.
	
	We are now in a position to prove \eqref{eq-waiting_time_tstarInf_pf4}. Recall that for the $t^{\star}=\infty$ policy, $\Pr[X(t)=1]_{\infty}=1-(1-p)^t=p_t$. Therefore, at time step $t_{\text{req}}+1$, the probability that $k_1\geq 1$ elementary links are inactive is $(1-p_{t_{\text{req}}+1})^{k_1}$ and the probability that $M-k_1$ elementary links are active is $p_{t_{\text{req}}+1}^{M-k_1}$. In the subsequent time steps, each inactive elementary link from the previous time step is active with probability $p$ and inactive with probability $1-p$. Therefore,
	\begin{align}
		\Pr[W_{E'}(t_{\text{req}})=t]_{\infty}&=\sum_{k_1=1}^M\binom{M}{k_1}(1-p_{t_{\text{req}}+1})^{k_1}p_{t_{\text{req}}+1}^{M-k_1}\sum_{k_2=1}^{k_1}\binom{k_1}{k_2}(1-p)^{k_2}p^{k_2-k_1}\dotsb\nonumber\\
		&\qquad\qquad\qquad\qquad\qquad\qquad\dotsb\sum_{k_{t-1}=1}^{k_{t-2}}\binom{k_{t-2}}{k_{t-1}}(1-p)^{k_{t-1}}p^{k_{t-2}-k_{t-1}}p^{k_{t-1}}\\
		&=\sum_{k_1=1}^M\binom{M}{k_1}(1-p_{t_{\text{req}}+1})^{k_1}p_{t_{\text{req}}+1}^{M-k_1}p^{k_1}\sum_{k_2=1}^{k_1}\dotsb\nonumber\\
		&\qquad\qquad\qquad\qquad\qquad\dotsb\sum_{k_{t-1}=1}^{k_{t-2}}\binom{k_1}{k_2}\dotsb\binom{k_{t-2}}{k_{t-1}}(1-p)^{k_2}\dotsb(1-p)^{k_{t-1}}\\
		&=\sum_{k_1=1}^M\binom{M}{k_1}(1-p_{t_{\text{req}}+1})^{k_1}p_{t_{\text{req}}+1}^{M-k_1}\Pr[W_{k_1}(0)=t-1]_{\infty}\\
		&=(1-(1-p_{t_{\text{req}}+1})(1-p)^{t-1})^M-(1-(1-p_{t_{\text{req}}+1})(1-p)^{t-2})^M,
	\end{align}
	which is precisely \eqref{eq-waiting_time_tstarInf_pf4}.
	
	Now, for brevity, let $\widetilde{q}\equiv 1-p_{t_{\text{req}}+1},\quad q\equiv 1-p$. Then,
	\begin{align}
		\Pr[W_{E'}(t_{\text{req}})=t]_{\infty}&=(1-\widetilde{q}q^{t-1})^M-(1-\widetilde{q}q^{t-2})^M\\
		&=\sum_{k=0}^M\binom{M}{k}(-1)^k(\widetilde{q}q^{t-1})^k-\sum_{k=0}^M\binom{M}{k}(-1)^k(\widetilde{q}q^{t-2})^k\\
		&=\sum_{k=1}^M\binom{M}{k}(-1)^k\widetilde{q}^k(q^{t-1})^k(1-q^{-k}).
	\end{align}
	Then, using the fact that
	\begin{equation}
		\sum_{t=2}^\infty t (q^k)^{t-1}=\frac{q^k(2-q^k)}{(1-q^k)^2},
	\end{equation}
	we obtain
	\begin{align}
		\mathbb{E}[W_{E'}(t_{\text{req}})]_{\infty}&=\sum_{t=1}^\infty t \Pr[W_{E'}(t_{\text{req}})=t]_{\infty}\\
		&=(1-\widetilde{q})^M+\sum_{k=1}^M\binom{M}{k}(-1)^k\widetilde{q}^k\left(\frac{q^k(2-q^k)}{(1-q^k)^2}\right)(1-q^{-k})\\
		&=(1-\widetilde{q})^M+\sum_{k=1}^M\binom{M}{k}(-1)^{k+1}\widetilde{q}^k\left(1+\frac{1}{1-q^k}\right)\\
		&=(1-\widetilde{q})^M+\sum_{k=1}^M\binom{M}{k}(-1)^{k+1}\widetilde{q}^k\left(1+\frac{1}{p_k}\right)\\
		&=(1-\widetilde{q})^M+\sum_{k=1}^M\binom{M}{k}(-1)^{k+1}\widetilde{q}^k+\sum_{k=1}^M\binom{M}{k}(-1)^{k+1}\frac{(1-p_k)^{t_{\text{req}}+1}}{p_k}\\
		&=1+\sum_{k=1}^M\binom{M}{k}(-1)^{k+1}\frac{(1-p_k)^{t_{\text{req}}+1}}{p_k},
	\end{align}
	where in the second-last line we used the fact that $\widetilde{q}^k=(1-p_k)^{t_{\text{req}}+1}$. Finally, using the fact that $1=\sum_{k=1}^M\binom{M}{k}(-1)^{k+1}$, we obtain
	\begin{equation}
		\mathbb{E}[W_{E'}(t_{\text{req}})]_{\infty}=\sum_{k=1}^M\binom{M}{k}(-1)^{k+1}\left(1+\frac{(1-p_k)^{t_{\text{req}}+1}}{p_k}\right),
	\end{equation}
	as required.

\begin{center}\rule{0.5\textwidth}{1pt}\end{center}

\twocolumngrid

\bibliography{refs_main_new.bib}

\end{document}